%% file: 00_main.tex
\documentclass[sigconf,9pt]{acmart}
\usepackage{graphicx}
\usepackage{balance}  
\usepackage{color,xcolor}
\usepackage[linesnumbered,ruled,vlined]{algorithm2e}
\usepackage{amsmath}
\usepackage{url}
\usepackage{bm}
\usepackage{enumerate}
\usepackage{afterpage}
\usepackage{caption}
\usepackage{subcaption}
\usepackage{multirow}
\usepackage{booktabs}
\usepackage{hyperref}
\hypersetup{colorlinks=true, citecolor=red, filecolor=blue, linkcolor=blue, urlcolor=blue, bookmarksdepth=3}

\theoremstyle{plain}

\newtheorem{lemma}{Lemma}

\theoremstyle{definition}
\newtheorem{definition}{Definition}

\newcommand{\algoref}[1]{Algo.~\ref{#1}}
\newcommand{\figref}[1]{Fig.~\ref{#1}}
\newcommand{\tabref}[1]{Table~\ref{#1}}

\newcommand{\equref}[1]{Eq.~(\ref{#1})}

\newcommand{\lemref}[1]{Lemma~\ref{#1}}

\newcommand{\defref}[1]{Def.~\ref{#1}}

\newcommand{\fakeparagraph}[1]{\vspace{0.3mm}\noindent\textbf{#1.}}

\newcommand{\eg}{\emph{e.g.},\xspace}
\newcommand{\ie}{\emph{i.e.},\xspace}

\newcommand{\query}[1]{disquery{#1}}
\newcommand{\pick}[1]{arr{#1}}
\newcommand{\deliver}[1]{arr{#1}}
\newcommand{\travel}[1]{ComTravel{#1}}
\newcommand{\arr}[1]{arr{#1}}
\newcommand{\num}[1]{num{#1}}
\newcommand{\latest}[1]{latest{#1}}
\newcommand{\pickloc}[1]{pckl{#1}}

\newcommand{\chengdu}{\textbf{\textit{Chengdu}}}
\newcommand{\haikou}{\textbf{\textit{Haikou}}}

\newcommand{\cubic}{\textit{Cubic Time Algorithm}}
\newcommand{\quadra}{\textit{Quadratic Time Algorithm}}
\newcommand{\linear}{\textit{Linear Time Algorithm}}
\newcommand{\nil}{\textit{NIL}}

\newlength{\oldtextfloatsep}
\newlength{\oldfloatsep}

\ifodd 0

\else

\fi

\AtBeginDocument{%
  \providecommand\BibTeX{{%
    \normalfont B\kern-0.5em{\scshape i\kern-0.25em b}\kern-0.8em\TeX}}}

\begin{document}

\title{A Fast Insertion Operator for Ridesharing over Time-Dependent Road Networks}

\author{Zengyang Gong}
\affiliation{%
  \institution{The Hong Kong University of Science and Technology}
  \city{Hong Kong SAR}
  \country{China}
}
\email{zgongae@connect.ust.hk}
\author{Yuxiang Zeng}
\affiliation{%
  \institution{The Hong Kong University of Science and Technology}
  \city{Hong Kong SAR}
  \country{China}
}
\email{yzengal@cse.ust.hk}
\author{Lei Chen}
\affiliation{%
  \institution{The Hong Kong University of Science and Technology}
  \city{Hong Kong SAR}
  \country{China}
}
\email{leichen@cse.ust.hk}

\input{01_abstract}

\maketitle

\input{02_intro}
\input{03_definition}
\input{05_query}

\input{06_experiment}
\input{04_related}
\input{07_conclusion}

\balance
\bibliographystyle{ACM-Reference-Format}
\bibliography{reference}

\end{document}

%% file: 01_abstract.tex
\begin{abstract}\label{sec:abstract}
Ridesharing has become a promising travel mode recently due to the economic and social benefits. As an essential operator,  ``insertion operator'' has been extensively studied over static road networks. When a new request appears, the insertion operator is used to find the optimal positions of a worker's current route to insert the origin and destination of this request and minimize the travel time of this worker. Previous works study how to conduct the insertion operation efficiently in static road networks, however, in reality, the route planning should be addressed by considering the dynamic traffic scenario (\ie a time-dependent road network). Unfortunately, under this setting, existing solutions to the insertion operator become the bottleneck of a route planning algorithm in ridesharing, because the feasibility and travel time of the worker to serve the new request are dependent on the pickup and deliver time of this request, instead of directly calculating over the static road network. This paper studies the insertion operator over time-dependent road networks. Specially, we first introduce a baseline insertion method by calculating the arrival time along the new route from scratch, it takes $O(n^3)$ time, where $n$ is the total number of requests assigned to the worker. To reduce the high time complexity, we calculate the compound travel time functions along the route to speed up the calculation of the travel time between vertex pairs belonging to the route, as a result time complexity of an insertion can be reduced to $O(n^2)$. Finally, we further improve the method to a linear-time insertion algorithm by showing that it only needs $O(1)$ time to find the best position of current route to insert the origin when linearly enumerating each possible position to
insert the destination of the new request. Evaluations on two real-world and large-scale datasets show that our methods can accelerate the existing insertion algorithm by up to 25.79 times.
\end{abstract}

%% file: 02_intro.tex
\section{Introduction}\label{sec:intro}
Ridesharing is a service for a worker to provide the shared ride for multiple requests with similar traveling schedules. In real world, various extensive existing applications like car-pooling, food delivery and last-mile logistics are based on ridesharing service \cite{spatialReview,TongJoS17}. There is a set of workers and dynamically incoming requests in the service, when a new request appears, a fundamental function is to find a worker with the minimized travel time to deliver this request. Because of tremendous social and economic benefits of ridesharing, \eg{ relieving traffic congestion and improving utilization of road}, ridesharing has been extensively studied  \cite{huangyan2014}, \cite{yuzheng2015}, \cite{peng2017}.

The insertion operator has been recognized as a fundamental operation to solve route planning problem in ridesharing and has been well studied over static road networks scenario \cite{yuxiang2018vldb}, \cite{yuxiang2019icde}, \cite{yuxiang2020tkde}. Given a worker with a feasible route to serve the assigned requests, when a new request appears, the insertion operator attempts to find the appropriate positions in the route to insert origin and destination of the new request while keeping the route's feasibility and minimize the total travel time of the new route. In the literature, this operator is widely adapted to plan routes in the ridesharing service \cite{yuzheng2015}, \cite{huangyan2014}, \cite{yuxiang2018vldb}, \cite{DBLP:journals/tods/TongZZCX22}, \cite{yuxiang2019icde}.

However, in real-world applications, the travel time of an edge on the road network is changing over time. Thus, a road network is recently formalized as a time-dependent graph where each edge is associated with a time-dependent weight function \cite{shortestquery}. Existing insertion operators work inefficiently over time-dependent road network, which become a bottleneck of ridesharing service. Therefore, in this paper, to serve requests in real-time, we study a fast insertion operator over time-dependent road networks.  

\begin{figure}[t]
	\centering
    \begin{subfigure}[b]{0.38\textwidth}
		\includegraphics[width=\textwidth]{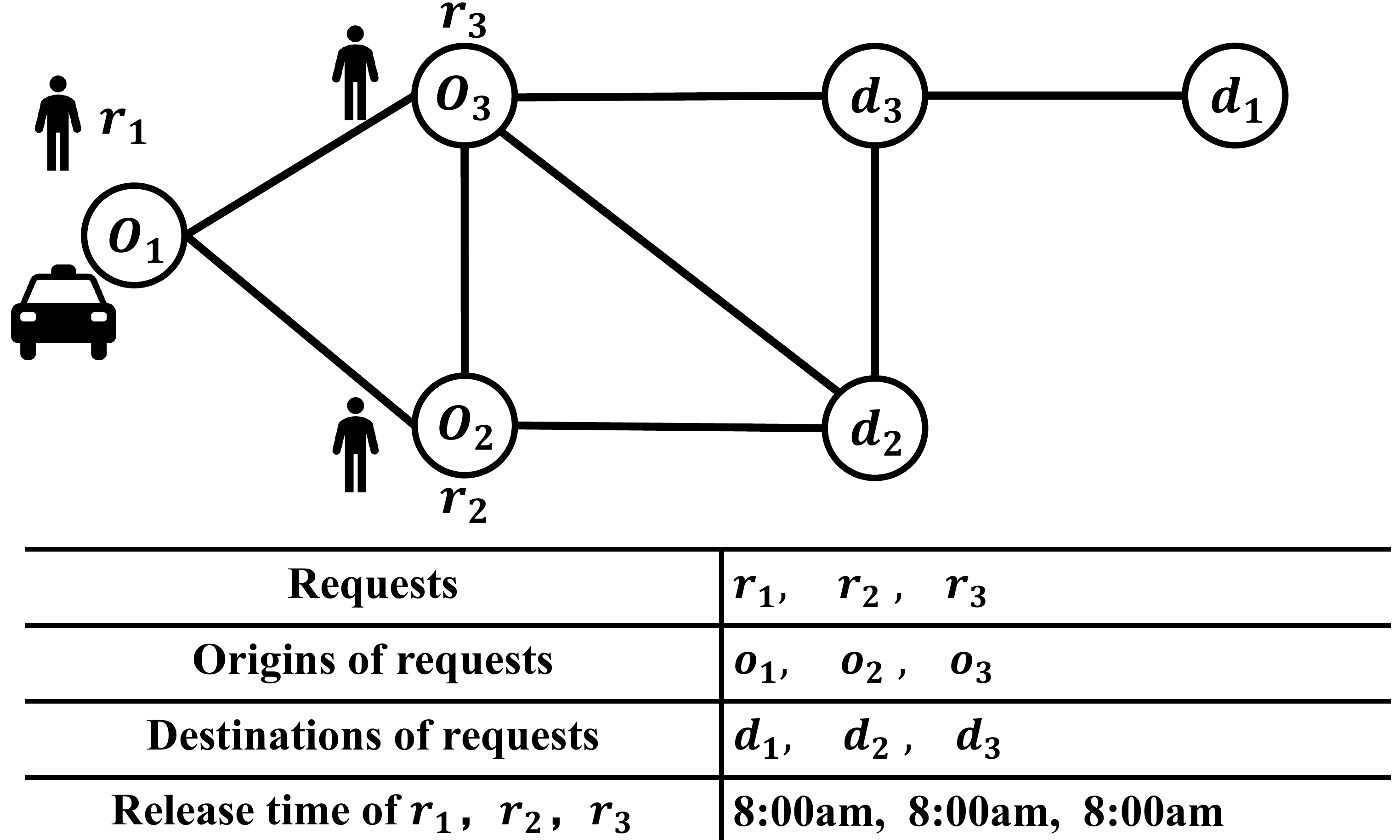}
		\vspace{-4ex}
		\caption{\footnotesize{Time-dependent road network}}
		\label{fig: example1a}
		\vspace{2ex}
	\end{subfigure}
	
	\begin{subfigure}[b]{0.39\textwidth}
		\includegraphics[width=\textwidth]{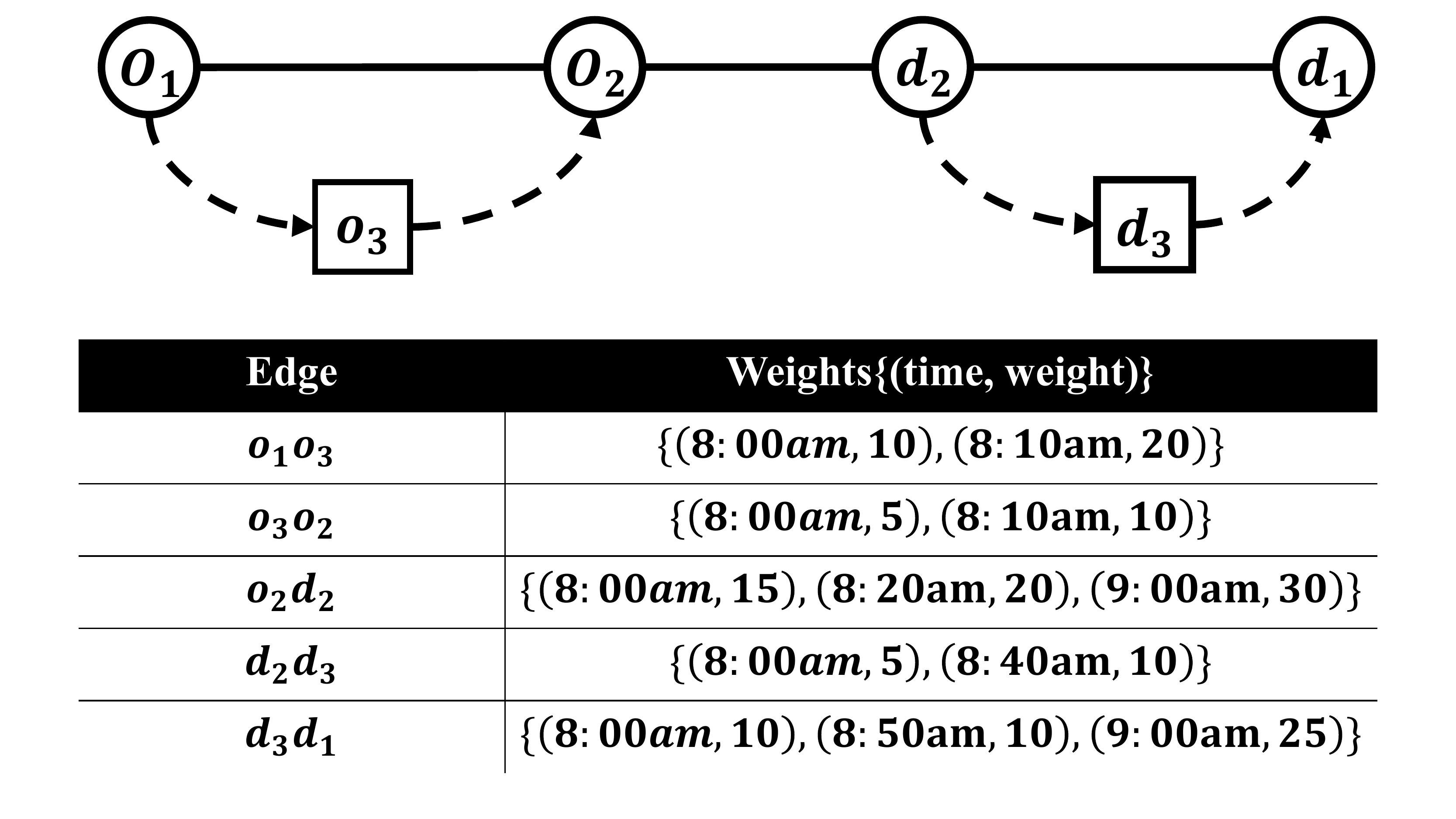}
		\vspace{-4ex}
		\caption{\footnotesize{Insertion operator}}
		\label{fig:example1b}
	\end{subfigure}
	\vspace{-2ex}
	\caption{Insertion operator over a time-dependent road network.}
	\label{fig:example1}
	\vspace{-4ex}
\end{figure}

We use the following example to show the difficulties and challenges of insertion operator over time-dependent road network.
\fakeparagraph{Example 1} \textit{ \figref{fig: example1a} shows an example of ridesharing service over time-dependent road network. There are three requests $r_1, r_2$ and $r_3$. At 8:00am, three requests $r_1$, $r_2$ and $r_3$ are released, and assigned to the worker locates in $o_1$. We assume $\langle o_1, o_2, d_2, d_1 \rangle$ is the route with the shortest time to serve $r_1$ and $r_2$, which is represented as the solid line in \figref{fig:example1b}. For $r_3$ with the origin-destination pair $o_3$ and $d_3$ on the road network, to serve this request, the insertion operator tries to insert new request's origin and destination before $o_2$ and $d_1$ respectively, as the dash lines shown in \figref{fig:example1b}. We can get a new candidate route $\langle o_{1},o_{3},o_{2}, d_{2}, d_{3}, d_{1} \rangle$. Specially, each edge of the route over time-dependent road network is associated time-dependent weights. For instance, for the edge $\langle o_1, o_3 \rangle$ on this road network is associated with a set \{(8:00am, 10),(8:10am, 20)\}, \eg{ (8:00am, 10) for $\langle o_1, o_3 \rangle$ says it takes 10 minutes from $o_1$ to $o_3$ at 8:00am, and (8:10am, 20) represents it takes 20 minutes at 8:10am.} The dynamic travel cost makes it more challenging to apply the insertion operator efficiently, as explained below. }

\fakeparagraph{Challenges} For the candidate route of serving the new request $\langle o_1, o_3, o_2, d_2, d_3, d_1 \rangle$ shown in \figref{fig:example1b}, over the well-studied static road network, the travel time along this route could be updated directly by adding constant travel time of 4 edges over the static road network \eg{$\langle o_1, o_3 \rangle$ and $\langle o_3, o_2 \rangle$ for inserting $o_3$, $\langle d_2, d_3 \rangle$ and $\langle d_3, d_1 \rangle$ for inserting $d_3$}. When we further to calculate the travel time of this route, the arrival time of each vertex belongs to the route could be ignored. However, in this example over time-dependent road networks, the newly arrival time of the vertex is dependent on the arrival time of the previous vertex. Assume start from $o_1$ at time $t_1$ and $t_1 = 8:00am$ we could calculate the travel time of this route from scratch in a recursive way
\begin{gather*}
    arr(o_1) = t_1 \\
    arr(o_3) = arr(o_1) + \query{(o_1, o_3, arr(o_1))} \\
    \cdots \\
    arr(d_3) = arr(d_2) + \query{(d_2, d_3, arr(d_2))} \\
    arr(d_1) = arr(d_3) + \query{(d_3, d_1, arr(d_3))}
\end{gather*}

$arr(o_1)$ indicates the arrive time to vertex $o_1$ in this route, and $\query{(o_1, o_3, arr(o_1))}$ represents the shortest travel time query over time-dependent road network, it returns the shortest travel time from $o_1$ to $o_3$ when start from $o_1$ at time $arr(o_1)$. This recursive method incurs huge number of shortest travel time queries. The shortest travel time query over the large-scale time-dependent road network is more time-consuming than the query over the static road network \cite{shortestquery}, therefore invoking the shortest travel time query frequently makes the insertion operator becomes the bottleneck of ridesharing service over time-dependent road network.   

In this paper, we study the insertion operator for ridesharing over time-dependent road networks. The main idea of our solution is as follows. We first study the baseline method \cite{yuxiang2018vldb} with $O(n^3)$ time complexity which invokes large amounts of shortest travel time queries over the time-dependent road network and makes the insertion become the efficiency bottleneck. Then, we reduce the time complexity to $O(n^2)$ by compounding the edge weight functions, checking the constrains and calculating the objective in $O(1)$ time. Finally, we reduce our method to $O(n)$ time complexity, by enumerating the potential position of the route to insert destination and finding the optimal position to insert origin of the new request in $O(1)$ time.

The main contributions are summarised as follows:
\begin{itemize}
  \item
  We are the first to study an efficient insertion operator for ridesharing over time-dependent road networks.
  \item
  We design a novel method to reduce the time complexity from cubic in baseline method to quadric by compounding the weight functions of road segments.
  \item
  Then, we propose an optimal solution with the time complexity of $O(n)$ and reduce the memory cost space from $O(n^2)$ in quadric time method to $O(n)$.
  \item
  Extensive experiments on real datasets demonstrate that our solution can speed
  up the insertion operator by 2.16 to 25.79 times on time-dependent road networks.
\end{itemize}

The rest of this paper is organized as follows. We first present the problem statement in Section \ref{sec:definition}. Then, our algorithms are elaborated in Section \ref{sec:methodology}. We evaluate our algorithms in Section \ref{sec:experiment}. We review the related works in Section \ref{sec:related}. Finally, Section \ref{sec:conclusion} concludes this paper.

%% file: 03_definition.tex
\section{Problem Definition}\label{sec:definition}

\begin{table}[t]
	\centering
	\caption{Summary of major notations.}
	\vspace{-2ex}
	\label{table:notations}
    \begin{small}
    \resizebox{0.45\textwidth}{!} {%
	\begin{tabular}{|c|c|}
		\hline
		Notation & Description \\
		\hline
		$G(V,E,F)$ & time-dependent road network\\
		\hline
		$\query{(u,v,t)}$ & shortest travel time query over $G(V,E,F)$\\
		\hline
		$\travel{(u,v,t)}$ & compound travel time function over route $u$ to $v$ \\
		\hline
		$v_w, c_w$ & location and capacity of worker \\
		\hline
		$o_r, d_r$ & origin and destination of request \\ 
		\hline
		$\pick{(o_r)}, \deliver{(d_r)}$ & pickup time and deliver time of request $r$ \\
		\hline
		$\arr{[k]}$ & arrival time at $v_k$ along the route of worker \\
		\hline
        $\num{[k]}$ & number of picked request at $v_k$ along the route of worker \\
		\hline
		$\latest{[k]}$ & the latest arrival time at $v_k$ to keep feasibility of the route \\
		\hline
		$\pickloc{[k]}$ & the $i$ value when insert $d_r$ before $v_k$ \\
		\hline
	\end{tabular}
	}
	\end{small}
  \vspace{-3.9ex}
\end{table}


In this section, we introduce the notations and the formal definition of insertion operator on time-dependent road networks. Major notations in this work are summarized in \tabref{table:notations}.

\subsection{Time-dependent Road Network}
\begin{definition}[Time-dependent Road Network]
\textit{A directed graph $\textit{G(V,E,F)}$ is utilized to model the dynamic road network, where $\textit{V}$ is the set of vertices and each vertex $v \in \textit{V}$ represents one geo-location (\eg{ road intersection}); $\textit{E} \subset \textit{V} \times \textit{V}$ is the set of road edges. For each edge $(u,v) \in \textit{E}$ indicating a directed edge from $u$ to $v$, there is a weight function $f_{u,v}(t) \in \textit{F}$ associated to it, and $t$ is a variable indicating the starting time. The value of $f_{u,v}(t)$ denotes the travel time to $v$ when starting from $u$ at time $t$, which is always non-negative.}
\end{definition}

Following the conventions in existing studies \cite{shortestquery} \cite{yuanyeicde2019} \cite{yuanyeicde2021}, we use piecewise linear functions to model the time-dependent road networks. To reduce the space cost, a set of interpolation points $P = \{(t_1, w_1), (t_2, w_2), \dots, (t_k, w_k)\}$ is saved to fit $f_{u,v}(t)$. In each point $(t_i, w_i)$, it indicates when starting from $u$ at $t_i$, it will cost $w_i$ unit time to arrival at $v$, for any $x \leq y$, $t_x \leq t_y$. A straight line connected two successive points $(t_i, w_i), (t_{i+1}, w_{i+1})$ fits the linear function in the time domain $[t_i, t_{i+1}]$. For example, if the worker starts from $u$ at $t \in [t_1, t_2]$, his travel time will be $f_{u,v}(t) = w_1 + (t - t_1) \frac{w_2 - w_1}{t_2 - t_1}$. Then the time function associated with $(u,v)$ can be formalized as    

\begin{equation}
    f_{u,v}(t) = \begin{cases}
                    w_1, & t = t_1 \\
                    w_1 + (t - t_1) \frac{w_2 - w_1}{t_2 - t_1}, & t_1 < t \leq t_2\\
                    \cdots\\
                    w_{k-1} + (t - t_{k-1}) \frac{w_k - w_{k-1}}{t_{k} - t_{k-1}}, & t_{k-1} < t \leq t_k\\
                    w_k, &   t = t_k
                    \end{cases} 
\end{equation}

The function domain is $[t_1, t_k]$ for $f_{u,v}(t)$, where $t_1$ is the earliest departure time from beginning of the edge, and $t_k$ is the last available start time.

\textbf{FIFO Property.}
Following the existing works \cite{shortestquery} \cite{yuanyeicde2019} \cite{yuanyeicde2021}, edges of $\textit{G}$ satisfy the first-in-first-out (FIFO) property, and it implies that if there are two workers driving on the same edge $(u,v)$, the worker starts earlier from the $u$ then it will also arrive at $v$ earlier. For every $f_{u,v}(t) \in F$, and starting times $t_1 \leq t_2$, we can derive $t_1+f_{u,v}(t_1) \leq t_2+f_{u,v}(t_2)$ based on this property.

\subsection{Problem Statement}
In this subsection we first introduce some basic concepts, and then we formally give the definition of the insertion operator for ridesharing over time-dependent road networks (``time-dependent insertion'' for short).

\begin{definition}[Worker]
\textit{A worker is denoted by ${w} = \langle v_{w}, c_{w} \rangle$ with an initialized location $v_{w} \in {V}$ and a capacity $c_{w}$, and the number of passengers $w$ can carry at any time can not exceed $c_{w}$.} 
\end{definition}

\begin{definition}[Request]
\textit{A request is denoted by ${r}=\langle o_{r},d_{r},t_{r},e_{r},c_{r} \rangle$. This request appears at time $t_r$, which indicates the information of this request is only known until $t_r$. $o_{r} \in {V}$ is the origin of the request and $d_{r} \in {V}$ indicates the destination. This request can be completed if it is picked up at $o_r$ by a worker after release time $t_{r}$ and delivered to $d_r$ before deadline time $e_{r}$. The capacity $c_r$ indicates the number of passengers in this request $r$.}
\end{definition}

In real-world ridesharing platforms such as Uber \cite{Uber} and Didi Chuxing \cite{Didi}, a worker picks up the request when he/she arrives at $o_r$ and delivers it at $d_r$, the pickup time and deliver time are denoted as $\pick{(o_r)}$ and $\deliver{(d_r)}$, respectively. If a worker is feasible to serve this request, the deliver time cannot exceed the deadline time specified by the request \ie{ $\deliver{(d_r)} \leq e_r $}.

\begin{definition}[Route]
\textit{A sequence $\mathcal{S_R} = \{ \langle v_0, v_1, \cdots, v_n \rangle | arr(v_0) = t_0\}$ is the route of this worker to serve requests set $R$, $R = \{r_0, r_1, \cdots,$ $r_{|R|-1} ,r_{|R|}\}$ represents the currently undelivered requests by $w$. The current location of $w$ is 
$v_0$, and $t_0$ indicates the start time, $w$ starts this route from $v_0$ at time $t_0$. Except for $v_0$, each vertex corresponds to the origin or the destination of one request in ${R}$, i.e. $v_0 = v_{w}$ and for all $1 \leq i \leq n,  v_{i} \in \{o_r|r \in {R}\} \cup \{d_r|r \in {R}\}$.}

\end{definition}

\fakeparagraph{Travel Time}
Given a route $\mathcal{S_R}$ with successive vertexes $v_0,$\dots$,v_n$, for $k$ from 0 to n-1, the arrival time at $v_{k+1}$ depends on the starting time from $v_k$. If we use an array $\arr{[k]}$ to record the arrival time at $v_k$, then the time dependency can be represented as $\arr{[k+1]} = \arr{[k]} + \query{(v_k, v_{k+1}, \arr{[k]})}$ where the $\query{(v_k, v_{k+1},}$ ${ arr[k])}$ is the shortest travel time query over time-dependent road networks, this function indicates the shortest travel time from $v_k$ to $v_{k+1}$ when starting at time $arr[k]$. Thus, we can calculate the travel time of the whole route from the first vertex $v_0$ at $t_0$ in a recursive way: $\arr{[1]} = t_0 + \query{(v_0, v_1, t_0)}, \arr{[2]} = \arr{[1]} + \query{(v_1, v_2, \arr{[1]})},$ $\ldots,$ $\arr{[n]} = \arr{[n-1]} + \query{(v_{n-1}, v_n,}$ ${ \arr{[n-1]})}$.

\fakeparagraph{Feasibility}
We call a route $\mathcal{S_R}$ is feasible for worker $w$ to serve request set $R$, if the following four constrains are satisfied,

\begin{itemize}
    \item \textbf{Completion Constraint.} Worker $w$ should complete serving all requests assigned to it.
	\item \textbf{Order Constraint.} $\forall r \in {R}$, $d_r$ lies behind $o_r$, a request should be picked up by a worker first and then delivered to its destination. 
	\item \textbf{Deadline Constraint.} $\forall r \in {R}$, the worker $w$ delivers $r$ to his/her destination $d_r$ at time $\deliver{(d_r)}$ before its deadline time $e_r$.
	\item \textbf{Capacity Constraint.} At any time, the total number of the requests that a worker $w$ has picked up but not completed does not exceed the capacity $c_{w}$. 
\end{itemize}

\begin{definition}[Time-dependent Insertion]
\textit{Given a worker $w$ and his/her associated feasible route $\mathcal{S_R}$ start at $t_0$, when a new request $r^+$ appears at time $t_{r^+}$, the operator tries to insert $o_{r^+}$ at the $i$-th position of $\mathcal{S_R}$ (i.e. before vertex $v_i$) and insert $d_{r^+}$ at the $j$-th position of $\mathcal{S_R}$ (i.e. before vertex $v_j$) to get a new feasible route $\mathcal{S_{R^+}}$ where $R^+ = R \cup {r^+}$ with the minimum travel time. }
\end{definition}

\figref{fig:insertion} shows all possible insertion position pairs $(i,j)$, $v_0$ is the current location of $w$ and $v_n$ is the last vertex, $\forall i,j, 1 \leq i \leq j \leq n+1 $.

\begin{figure}[t]
	\centering
	\begin{subfigure}[b]{0.23\textwidth}
		\includegraphics[width=\textwidth]{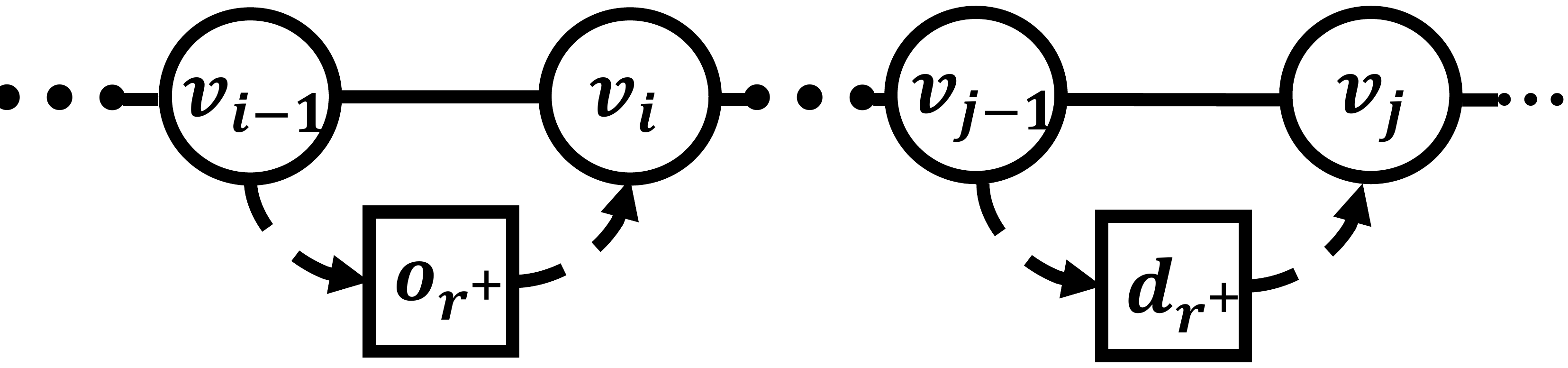}
		\vspace{-1ex}
		\caption{\footnotesize{$0 < i < j < n+1$}}
		\label{fig:insertionCase1}
	\end{subfigure}
	~~
	\begin{subfigure}[b]{0.21\textwidth}
		\includegraphics[width=\textwidth]{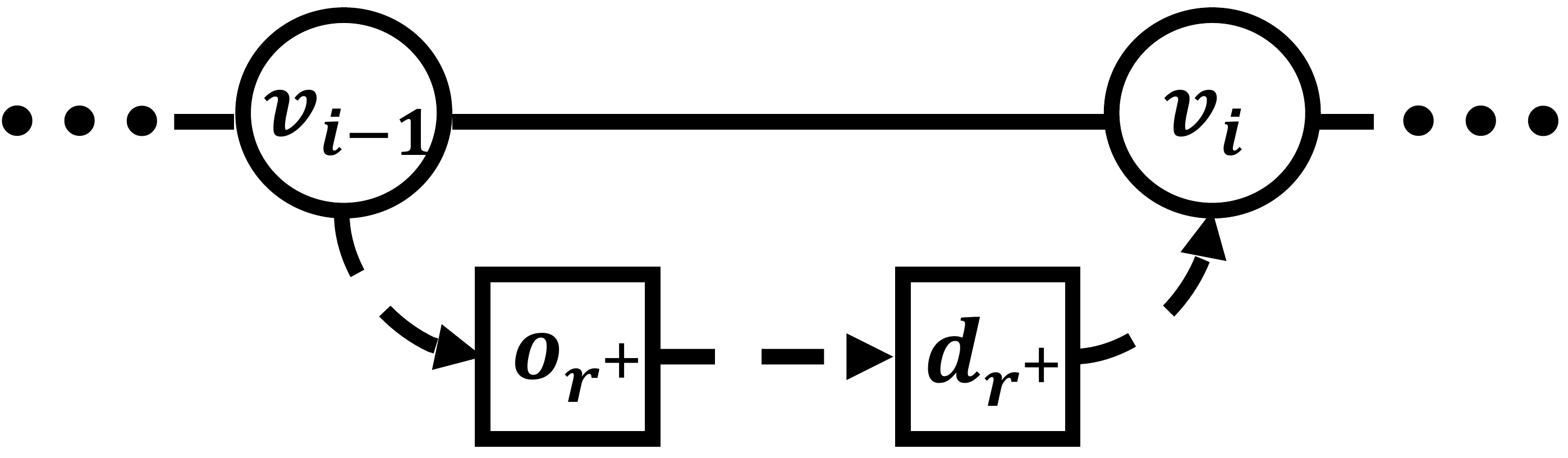}
		\vspace{-1.8ex}
		\caption{\footnotesize{$0 < i=j<n+1$}}
		\label{fig:insertionCase2}
	\end{subfigure}
	
	\begin{subfigure}[b]{0.21\textwidth}
		\includegraphics[width=\textwidth]{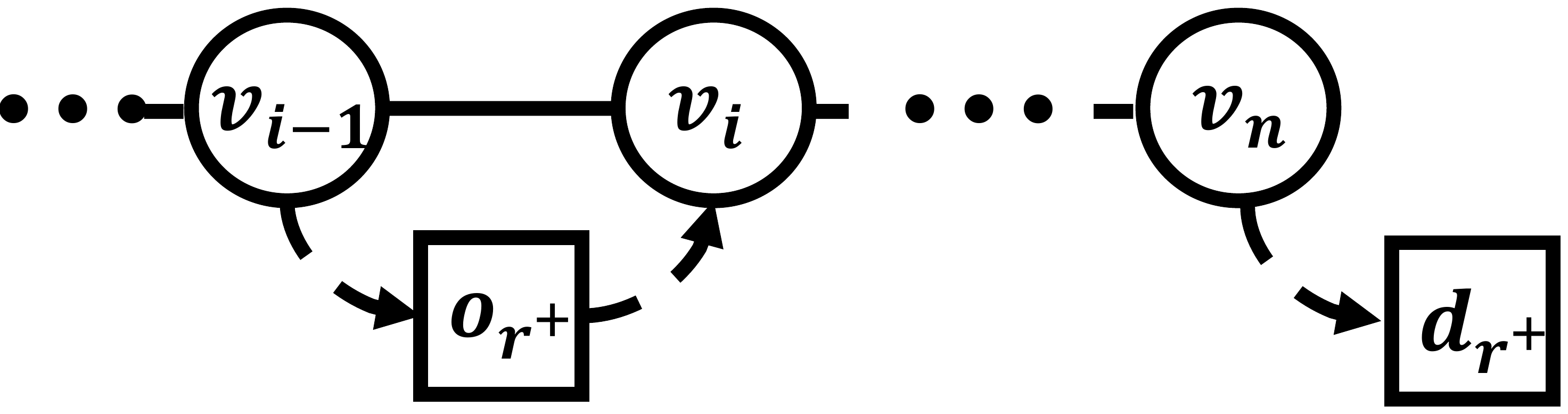}
		\vspace{-1ex}
		\caption{\footnotesize{$0 < i < j = n+1$}}
		\label{fig:insertionCase4}
	\end{subfigure}
	~~	
	\begin{subfigure}[b]{0.19\textwidth}
		\includegraphics[width=\textwidth]{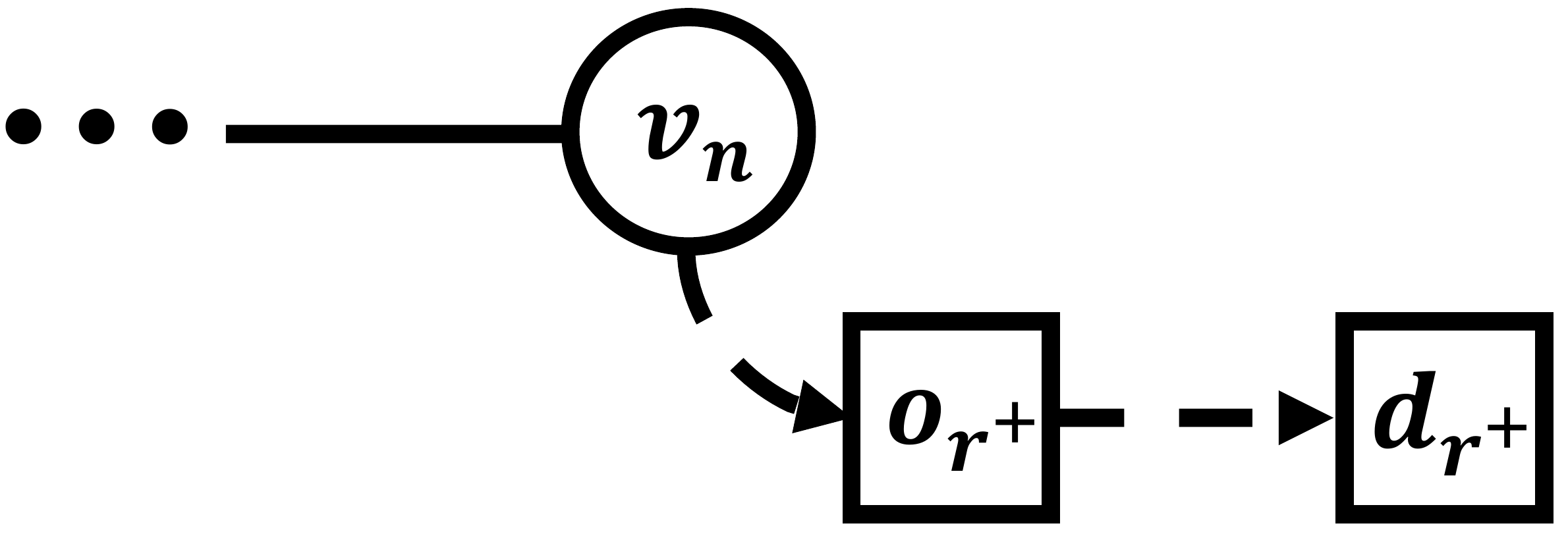}
  		\vspace{-2ex}
		\caption{\footnotesize{$i=j=n+1$}}
		\label{fig:insertionCase3}
	\end{subfigure}
    \vspace{-1ex}
    \caption{Insertion cases of all possible pairs$(i,j)$.}
	\label{fig:insertion}
	\vspace{-4ex}
\end{figure}

\subsection{Baseline Method}
We first extend the existing insertion operator over static road networks \cite{yuxiang2019icde} to the time-dependent road networks as the baseline method shown in \algoref{alg:n3}. 

\fakeparagraph{Basic Idea} The intuitions behind baseline method are as follows. $(1)$ We enumerate all potential location pairs $(i,j)$ where $1 \leq i \leq j \leq n+1$ for inserting the new request's origin $o_{r^+}$ and new request's destination $d_{r^+}$ to get a new route $S_{R^+}$; $(2)$ From the first vertex to the last vertex belong to $S_{R^+}$, we sequentially compute the travel time and check the feasibility; $(3)$ If there is no violation of the constraints until the last vertex of this new route and the total travel time is shorter than that of the currently optimal route, the algorithm replaces optimal route $S^*$ by $S_{R^+}$.

\fakeparagraph{Complexity Analysis} In lines 2-3 enumerating all potential pairs $(i,j)$ take $O(n^2)$ time cost. Lines 5-7 calculate new arrival time along the vertexes after inserting $o_{r^+}$ and $d_{r^+}$ until the last vertex, it takes $O(n)$ time. Line 6 involves huge number of travel time queries $\query{(v_{k-1}, v_k, \arr{[k]})}$ over time-dependent road networks, we assume the time cost of this query is $O(1)$ time as with previous work \cite{yuxiang2018vldb}\cite{yuxiang2019icde}\cite{yuxiang2020tkde}. We follow this assumption, although the query is not as efficient as the query function over the static road network \cite{shortestquery}. Thus, the total time cost of \algoref{alg:n3} is $O(n^3)$ and the space cost is $O(n)$ to maintain the route of the worker.

\begin{algorithm}
\caption{Baseline: Cubic Time Insertion}
\label{alg:n3}
\LinesNumbered 
\KwIn{a worker $w$ with a feasible route $S_R$ and a request $r^+$}
\KwOut{new route $S^*$ for $w$ and its travel cost $obj^*$}
$S^* \gets S_R, obj^* \gets +\infty, arr[0] \gets t_0 $ \;
\For{$i \gets $ 1 to $n+1$}{
    \For{$j \gets $ $i$ to $n+1$}{
    $S_{R^+}$ $\gets$ insert $o_{r^+}$ before $v_i$ and $d_{r^+}$ before $v_j$ in $S_R$ \;
        \For{$\langle v_{k-1}, v_k \rangle$ $\in$ $S_{R^+}$}{
            $arr[k] \gets \arr{[k-1]} + \query{(v_{k-1}, v_k, arr[k-1])}$ \;
            $obj \gets arr[k]$        
        }
    \If{$S_{R^+}$ is feasible $and$ $ obj < obj^*$}{ 
    $S^* \gets S_{R^+}$, $obj^* \gets obj$}
    }
}
\KwRet{$S^*$, $obj^*$}\;
\end{algorithm}

\fakeparagraph{Example 2} \textit{Back to the settings in Example 1. For $r_3$, suppose we want to check the feasibility and calculate the total travel time of insertion $(1,3)$. From the previous work, if the road network is modeled as a static graph, we can get the total travel directly as: $arr(d_3) + dis(o_1,o_3)+ dis(o_3,o_2) + dis(d_2,d_3)+ dis(d_3,d_1)$, where $dis(\cdot,\cdot)$ is the shortest distance between any two vertices in the static graph. However, due to the time-dependent travel costs of road segments in reality, arrival time at each vertex is dependent on the arrival time at the previous vertex along the route. For example, the arrival time at $d_2$ can not be updated directly from $o_3$, it is dependent on the new arrival time at $o_2$. It is same for $d_1$, new arrival time at $d_1$ can be calculated only if the arrival time at $d_3$ is known. Therefore, if we want to calculate the arrival time at each vertex in this new candidate route, we need to calculate the time from scratch at $o_1$. As shown in the previous section, if the worker starts from $o_1$ at $t_1$, we calculate the arrival time at each vertex along the route as :  $arr(o_3) = arr(o_1) + \query{(o_1, o_3, t_1)}, arr(o_2) = arr(o_3) + \query{(o_3, o_2, arr(o_3))}, arr(d_2) = arr(o_2) + \query{(o_2, d_2, arr(o_2))},$ $  \cdots  arr(d_1) = arr(d_3) + \query{(d_3, d_1, arr(d_3))}$, $arr(d_1)$ is the total travel time of this new route. For insertion $(1,3)$, the baseline method invokes 5 times $\query{}$ function.    
}

%% file: 05_query.tex
\section{Our Methodology}\label{sec:methodology}

To address the shortcomings of the baseline solution, we first propose the compound travel time functions to accelerate the travel time computing and feasibility checking to $O(1)$ time, thus the time complexity of time-dependent insertion is reduced to $O(n^2)$ by enumerating all $(i,j)$ values. After that, we prove our efficient insertion operator can find the optimal value of $i$ in $O(1)$ time when $j$ is given along a feasible route of the worker over the time-dependent road network. Therefore, both time and space complexity can be improved to $O(n)$ by enumerating $j$ along the feasible route of the worker. 

\subsection{Quadratic Time Algorithm}
\fakeparagraph{Basic Idea}
The insertion operator takes $O(n^2)$ time to enumerate all possible $(i,j)$ pairs to insert new origin-destination pair, and it only needs $O(1)$ time to calculate the new arrival time of each vertex and check the feasibility of route $S_{R^+}$ instead of taking $O(n)$ time in the baseline method. We first compound the edge weight functions from time-dependent road network to compound travel time functions, and further use an array $\num{}$ to record the number of picked but not delivered requests along the route. With the help of compound travel time functions and $\num{}$, checking feasibility and calculating travel time for a new route can be accelerated to $O(1)$ time. To maintain the compound travel time functions between any two vertexes of the current feasible route $\mathcal{S_R}$, the space cost of this method is $O(n^2)$.       

\subsubsection{Compound Travel Time Functions}
\
\newline
To accelerate the checking and computing, for a feasible route $S_R = \langle v_0, v_1, \dots v_n \rangle$, we compute the compound travel time functions $\travel{(v_x,v_y,t_x)} (\forall x,y, 0< x < y < n)$ to represent the travel time when starting from vertex $v_x$ at time $t_x$ to vertex $v_y$. For each successive sub-route $\langle v_k, v_{k+1} \rangle \in S_R, \forall k \in [0,n-1]$, we can get an edge weight function $f_{v_k, v_{k+1}}(t_k) \in F$ associated with it. If we start from first vertex $v_0$ at $t_0$, compounding $t_0$ to edge weight function of the sub-route $\langle v_0,v_1 \rangle$ we can compute the compound travel time function $f_{0,1}(t_0)$ for this sub-route, then recursively compounding $ t_1 = f_{0,1}(t_0)$ into $f_{1,2}(t_1)$, we can compute the compound travel time function for the sub-route $\langle v_0, v_1, v_2 \rangle$, where $travel(v_0,v_2,t_0) = f_{1,2}(f_{0,1}(t_0)) - t_0$. Following the procedure in a recursive way, we give the formal definition of compound travel time function along a route:
\begin{definition}[Compound Travel Time Function]
\label{def:TravelFunction}
\textit{Given two vertexes $v_x$ and $v_y$ associated with the sub-route $\langle v_x, v_{x+1}, \cdots, v_y \rangle$ in $S_R$, the compound travel time function $\travel{(v_x,v_y,t_x)}$ indicates the travel cost of the route when start from $v_x$ to $v_y$ at time $t_x$, where $\travel{(v_x,v_y,t_x)} = f_{y-1,y}(f_{y-2,y-1}(\cdots f_{x,x+1}(t_x))) - t_x$.}
\end{definition}

\subsubsection{Computing Arrival Time}
\
\newline
After applying operator $insert(i,j)$, $o_{r^+}$ is inserted before $v_i$ and $d_{r^+}$ is inserted before $v_j$ in $S_R$,
we get a new potential route $S_{R^+}$. We need to calculate the arrival time at each vertex that belongs to this new route to check the feasibility and calculate the total travel time of $S_{R^+}$. New arrival time to $v_i$ and $v_j$ is dependent on the pickup time and deliver time of $r^+$ respectively, which can be updated by $\query{}$ over the road network. For vertexes in sub-route $\langle v_{i+1}, v_{i+2}, \cdots, v_{j-1} \rangle$, their new arrival time is effected by inserting $o_{r^+}$ before $v_i$, compound travel time functions $\travel{(v_i,*,*)}$ can be utilized to update the arrival time of these vertexes in $O(1)$ time which avoids recursively invoking $\query{}$ from scratch.
Similarly, for each vertex in $\langle v_{j+1}, v_{j+2}, \cdots, v_n \rangle$, the new arrival time is dependent on the arrival time to $v_j$ which can be calculated by compound travel time functions $\travel{(v_j,*,*)}$.

Along the route $S_R$, inserting $o_{r^+}$ and $d_{r^+}$ at $i$-th position and $j$-th position, the pickup time of the new request $\pick{(o_{r^+})}$, which is also the arrival time to $o_{r^+}$, is dependent on the arrival time of the previous vertex $\arr{[i-1]}$ in the route. As for the delivery time of the new request $\deliver{(d_{r^+})}$, it depends on the new arrival time of the previous vertex after picking up the new request $\arr{'[j-1]}$ when $i<j$ (see \figref{fig:insertionCase1} and \figref{fig:insertionCase4}) or the pickup time of the request $\pick{(o_r)}$ when $i = j$ (see \figref{fig:insertionCase2} and \figref{fig:insertionCase3}).

We first calculate the pickup time of $r^+$, origin $o_{r^+}$ is inserted before $v_i$ and after $v_{i-1}$, so the pickup time can be calculated from arrival time to $v_{i-1}$:
\begin{equation}
\label{equ:pickup}
     \pick{(o_{r^+})} = \arr{[i-1]} + \query{(v_{i-1}, o_{r^+}, arr[i-1])}
\end{equation}

Similarly, after inserting $d_{r^+}$, the deliver time is dependent on the new arrival to $v_{j-1}$ after picking up $r^+$ or new request's pickup time, where $\arr{'[j-1]}$ is the new arrival time to $v_{j-1}$ after inserting $o_{r^+}$ before $v_i$:
\begin{equation}
\label{equ:deliver}
    \deliver{(d_{r^+})} = \begin{cases}
                    arr'[j-1] + \query{(v_{j-1}, d_{r^+}, arr'[j-1])}, & i < j \\
                    \pick{(o_{r^+})} + \query{(o_{r^+}, d_{r^+}, \pick{(o_{r^+})})}, & i = j\\ 
                    \end{cases}
\end{equation}

According the locations of the vertexes in $S_R$, we can utilize the compound travel time functions to calculate the new arrival time $\arr{'[k]}$ of each vertex $v_k$ in $O(1)$ time based on the following rules:
\begin{equation}
\label{equ:arr}
    \arr{'[k]} = \begin{cases}
                    \arr{[k]}, & k < i\\
                    \pick{(o_{r^+})} +  \query{(o_{r^+}, v_i, \pick{(o_{r^+})}}, & k = i \\
                    \arr{'[i]} + \travel{(v_i, v_k, \arr{'[i]})}, & i < k < j\\
                    \deliver{(d_{r^+})} + \query{(d_{r^+}, v_j, \deliver{(d_{r^+})})} , &  k = j \\
                    \arr{'[j]} + \travel{(v_j, v_k, \arr{'[j]})}, & j < k \leq n\\
                    \end{cases}
\end{equation}

In \equref{equ:arr}, vertexes locate before $v_i$ are not effected by serving the new request, the arrival times retain the same in original route $S_R$. For vertexes $v_i$ and $v_j$, the new arrival time are dependent on the pickup time and deliver time of the new request, two $\query{}$ functions are invoked to calculate the new arrival time to them respectively. For vertexes locate between $v_{i+1}$ and $v_{j-1}$, compound travel time functions $\travel{(v_i,*,*)}$ can calculate the new arrival time to them based on the new arrival time to $v_i$. Similarly, compound travel time functions $\travel{(v_j, *, *)}$ and the new arrival time to $v_j$ can be utilized to calculate the new arrival times of vertexes locate after $v_j$. Specially, the time complexity of shortest travel time query and compound travel time function are both $O(1)$. Therefore, for each possible insertion position pair $(i,j)$, we can calculate the arrival time of each vertex belongs to $S_{R^+}$ in $O(1)$ time. 

\subsubsection{Checking Constrains and Computing Objective}
\
\newline
Given a possible $(i,j)$ pair for insertion, there are two constraints,the deadline constraint and capacity constraint. These constraints need to be checked to determine whether the newly generated route after inserting is feasible or not.

We use $\latest{[k]}$ to denote the latest arrival time at $v_k$ without violating any deadline constraints at $v_k$ and all vertexes after $v_k$. Similar with calculating the arrival time at the vertex, the $\latest{[k]}$ can be calculated from compound travel time functions between two vertices of the route. For successive sub-route $\langle v_{k}, v_{k+1} \rangle \in S_R$, to make sure the worker satisfy the constraints for all vertexes from $v_k$, we need to calculate $\latest{[k]}$ for $k$ from $v_n$ to $v_0$ in a reverse order, the latest arrival time of each vertex can be calculated by the following equation, and $\latest{[n]}$ is the deadline time of the request whose destination is the last vertex $v_n$:
\begin{equation}
\label{equ:latest}
            \begin{cases}
                \latest{[k]} + \travel{(v_k,v_{k+1},\latest{[k]})} = \latest{[k+1]}, & \\
                \latest{[n]} = e_r, \text{ when    } v_{n} = d_{r}
             \end{cases}
\end{equation}

\begin{figure}[t]
\centering
\includegraphics[width=0.3\textwidth]{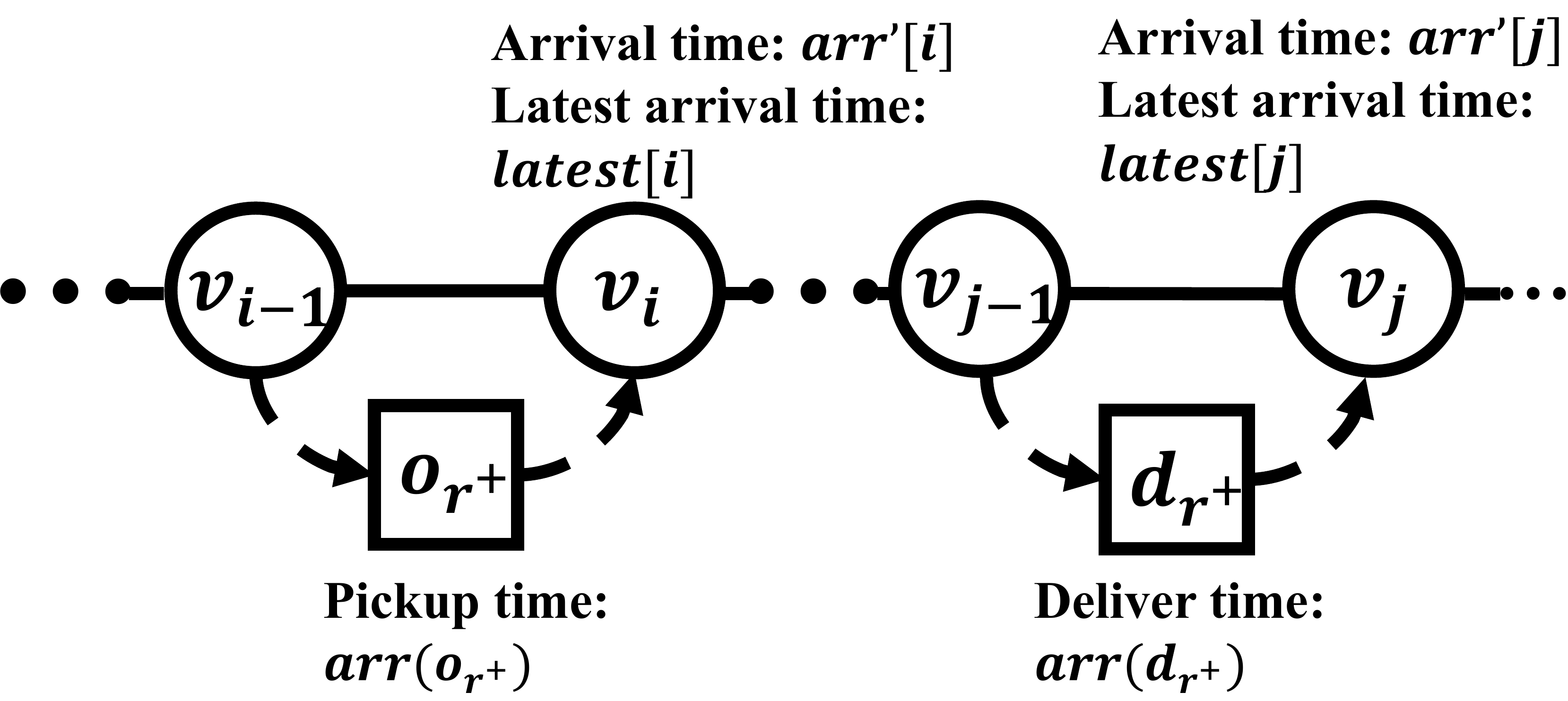}
\vspace{-1ex}
\caption{Checking deadline constraint.}
\vspace{-1ex}
\label{fig:prove}
\end{figure}

\begin{lemma}
\label{lem:latest}
The deadline constraint will not be violated if and only if 
(1) $\pick{(o_{r^+})} \leq e_{r^+}$;
(2) $\pick{(o_{r^+})} + \query{(o_{r^+}, v_i, \pick{(o_{r^+})}} \leq \latest{[i]}$;
(3) $\deliver{(d_{r^+})} \leq e_{r^+}$;
(4) $\deliver{(d_{r^+})} + \query{(d_{r^+}, v_j, \deliver{(d_{r^+})})}$    $\leq$ $latest{[j]}$;
\end{lemma}

\begin{proof}
\vspace{-1ex}
    We prove this lemma by using the general insertion case as \figref{fig:prove} shows. For this new request $r^+$, we first calculate the pickup and deliver time in condition (1) and (3) respectively. These two conditions guarantees this new potential route will not violate the deadline constraint of $r^+$. Condition (2) calculates the new arrival time to $v_i$ after picking up $r^+$, $\arr{'[i]} = \pick{(o_{r^+})} + \query{(o_{r^+}, v_i, \pick{(o_{r^+})}}$, to keep the deadline constraints after $v_i$ satisfied, $\arr'{[i]}$ can not exceed the latest arrival time to $v_i$.  Similarly, for vertex $v_j$, condition (4) calculates the new arrival time to it after serving the new request, $\arr{'[j]} = \pick{(d_{r^+})} + \query{(d_{r^+}, v_j, \deliver{(d_{r^+})}}$. In order not to violate the deadline constraints for vertexes locate after $v_j$, condition (4) guarantees $\arr{'[j]}$ is earlier than the latest arrival time to $v_j$.
\end{proof}

To check the capacity constraint in $O(1)$ time, we follow the previous work \cite{yuxiang2018vldb} and use $\num{[k]}$ to represents the number of requests picked up but not delivered at $v_k$. Then we calculate this array in the following rules, where $v_k$ is either a pickup location $o_r$ or a deliver location $d_r$ of a request $r \in R$:
\begin{equation}
\label{equ:num}
    \num{[k]} = \begin{cases}
                \num{[k-1]} + c_r, & v_{k} = o_{r}\\
                \num{[k-1]} - c_r, & v_{k} = d_{r}
                \end{cases}
\end{equation}

\begin{lemma} \cite{yuxiang2018vldb}
\label{lem:num}
Along the new route, the capacity constrains will not be violated if and only if (1) $ \num{[i-1]} \leq c_w - c_{r^+}$; (2) $ \num{[k]} \leq c_w - c_{r^+}, i < k  \leq j$.
\end{lemma}

Please refer to Ref.\cite{yuxiang2018vldb} for the proof of Lemma 2.

\subsubsection{Algorithm Details}
\
\newline
\fakeparagraph{Main Idea} The algorithm enumerates all possible $(i,j)$ pairs, and updates the new arrival time according \equref{equ:pickup} - \equref{equ:arr} in $O(1)$ time. For $v_i$, \lemref{lem:latest} (1)-(2) are utilized to check the deadline constraint will violate or not after inserting $o_{r^+}$, and the capacity constraint is checked by \lemref{lem:num} (1). \lemref{lem:latest} (3)-(4) and \lemref{lem:num} (2) are used to check the deadline constraint and capacity constraint at $v_j$ after inserting $d_{r^+}$ before it. The time cost for all these feasibility checking are $O(1)$.

\begin{algorithm}
\caption{Quadratic Time Insertion}
\label{alg:n2}
\LinesNumbered 
\KwIn{a worker $w$ with a feasible route $S_R$ and a request $r^+$}
\KwOut{new route $S^*$ for $w$ and its travel cost $obj^*$}
$S^* \gets S_R, obj^* \gets +\infty$ \;
Initialize $\arr{}$, $\latest{}$, $\num{}$ by  $\equref{equ:arr}$ - $\equref{equ:num}$ \;
Compound travel time functions $\travel{(i,j, \cdot)},\forall i,j, 0< i < j < n$ by \defref{def:TravelFunction} \;
\For{$i \gets $ 1 to $n+1$}{
    \lIf{\lemref{lem:latest} (1) or (2) violated}{continue}
    \lIf{\lemref{lem:num} (1) violated}{continue}
    \For{$j \gets $ $i$ to $n+1$}{
        \lIf{\lemref{lem:latest} (3) or (4) violated}{continue}
        \lIf{\lemref{lem:num} (2) violated}{break}
        
        \If{j = n+1}{$obj \gets \deliver{(d_{r^+})}$}
        \Else{        
        calculate $\arr{'[j]}$ by $\equref{equ:arr}$ \;
        $obj \gets \arr{'[j]} + \travel{(v_j, v_n, \arr{'[j]})}$ \;}
        \If{$obj < obj^*$}{$obj^* \gets obj, i^* \gets i, j^* \gets j$}
    }
    \If{$obj^* < +\infty$}{ $S^* \gets$ insert $o_r$ before $i^*$-th and $d_r$ before $j^*$-th in $S_R$}
}
\KwRet{$S^*$, $obj^*$}\;
\end{algorithm}

\fakeparagraph{Complexity Analysis} Line 2 initializes $\arr{}, \latest{}, \num{}$. Line 3 compounds the travel time function between any two vertexes along $S_R$. Lines 4 and 7 enumerate possible positions to insert $o_{r^+}$ and $d_{r^+}$ respectively, and each line takes $O(n)$ time. The deadline and capacity constraints are checked for inserting $o_{r^+}$ in line 5 and line 6. For a potential $j$-th position, lines 8 and 9 check the deadline and capacity constraint. Each line of constraints checking takes $O(1)$ time. If $j = n+1$ in line 10, as shown in \figref{fig:insertionCase3} and \figref{fig:insertionCase4},  $r^+$ is the last delivered request, the total travel time of the new route is the deliver time of the new request. Otherwise, we can get the new arrival time at $v_j$ after inserting origin and destination of the new request in $O(1)$ time based on Eq.\ref{equ:arr} in line 13, and then line 14 takes $O(1)$ time to get the objective from $v_j$ to $v_n$ based on the compound travel time function and new arrival time at $v_j$. Therefore, the total time cost of \algoref{alg:n2} is $O(n^2)$. For the space cost, one compound travel time function need to be initialized between any two vertexes in line 3, therefore one worker should keep $O(n^2)$ travel time functions, and the space cost is $O(n^2)$.

\begin{table}[t]
	\centering
	\caption{$\travel{}$ of $S_R$ in \algoref{alg:n2}}
    \vspace{-1ex}
	\label{tab:alg2compound}
	{\scriptsize
	\begin{tabular}{|c|c|c|c|c|c|}
		\hline
		 & $o_1$ & $o_2$ & $d_2$ & $d_1$ \\
		\hline
		$o_1$ &   $\times$     & $\travel{(o_1, o_2, t)}$	 & $\travel{(o_1, d_2, t)}$	& $\travel{(o_1, d_1, t)}$ \\
		\hline
		$o_2$ &   $\times$     & $\times$ & $\travel{(o_2, d_2, t)}$ & $\travel{(o_2, d_1, t)}$ \\
		\hline
		$d_2$ &   $\times$    & $\times$ & $\times$ & $\travel{(d_2, d_1, t)}$ \\
		\hline
		$d_1$ &   $\times$     & $\times$ & $\times$ & $\times$ \\
		\hline
	\end{tabular}}
  \vspace{-2ex}
\end{table}

\fakeparagraph{Example 3} \textit{Back to the settings in Example 1. For the feasible route with shortest travel time $S_R =  \langle o_1, o_2, d_2, d_1 \rangle$, the compound travel time functions along $S_R$ are shown in \tabref{tab:alg2compound}. For each vertex in $S_R$, the algorithm compounds and maintains travel time functions to all vertices after it. Therefore, for each vertex, we can update the arrival time to any vertex after it with these functions instead of applying the shortest path query algorithm over large-scale time-dependent road networks. For $r_3$, we also consider the insertion of $(1,3)$. As \figref{fig:example3} shows, after inserting $o_3$ before $o_2$, same as the previous example, 2 $\query{}$ functions are inevitable to invoke to calculate $\arr{'[1]})$. However, because $o_2$ maintains one compound travel time function to $d_2$ after it, the new arrival time to $d_2$ can be updated as $\arr{'[1]} + \travel{(o_2, d_2, \arr{'[1]})}$ directly, which is shown as the dot line in \figref{fig:example3}. To calculate the final travel time $\arr{'[2]}$, 2 $\query{}$ are invoked from $d_2$ based on the updated $\arr{'[2]}$. Therefore, \algoref{alg:n2} compounds 3 travel time functions, and for insertion $(1,3)$, it takes 4 times $\query{}$ function.} 

\begin{figure}[t]
\centering
\includegraphics[width=0.3\textwidth]{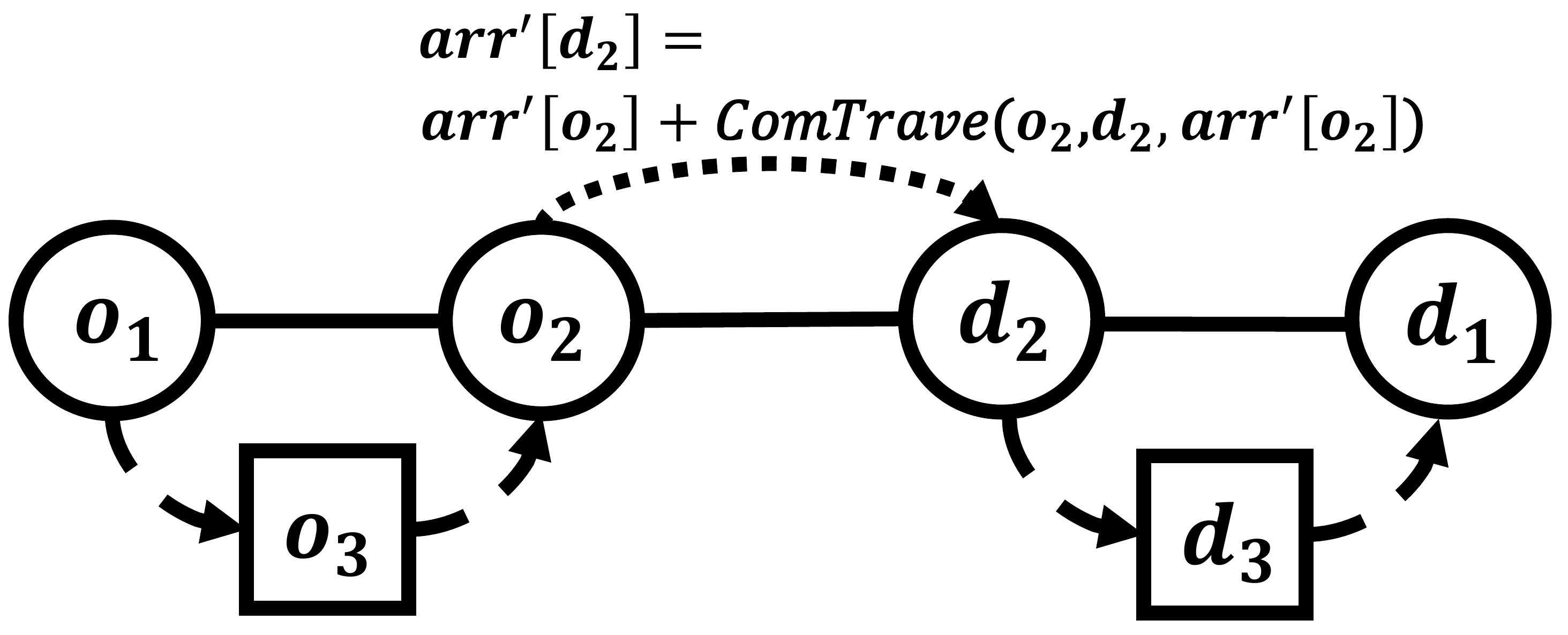}
\vspace{-1ex}
\caption{Update arrival time with $\travel{}$.}
\vspace{-1ex}
\label{fig:example3}
\end{figure}

\subsection{Linear Time Algorithm}
\fakeparagraph{Basic Idea}
Instead of enumerating possible positions to insert new request's origin and destination, maintaining compound travel time functions for total all $O(n^2)$ number of $(i,j)$ pairs, we enumerate the value of $j$ and find the optimal $i$. Then we further calculate the new travel time in $O(1)$ time based on the compound travel time functions. Therefore, the worker only need to store compound travel time functions from $v_j$ to both the last vertex $v_n$ and to the next vertex of $v_j$, where $j$ is each possible $j$-th position of current route to insert new request's destination, the space cost for compound travel time functions can be reduced from $O(n^2)$ to $O(n)$.

\subsubsection{Updating Optimal Pickup Location as Enumerating $j$}
\
\newline
Given a vertex $v_k$ with route $S_R$, if the destination of the new request $d_{r^+}$ is inserted before it, based on the order constraint there are only two cases for inserting pickup vertex $o_{r^+}$. Case 1 ($i = j = k$), as shown in \figref{fig:insertionCase2} and \figref{fig:insertionCase3} , the origin is also inserted before $v_k$ in the current route. Case 2 ($i<j$), as shown in \figref{fig:insertionCase1} and \figref{fig:insertionCase4}, the destination is inserted before $v_k$, and the origin $o_{r^+}$ is inserted in the sub-route $\langle v_1, v_2,\cdots, v_{k-1} \rangle$ before $v_k$. 

Based on the FIFO property of the time-dependent road network, if we can get a earlier arrival time to vertex $v_k$ therefore we can get a earlier arrival time to the last vertex $v_n$, which is a better route with shorter total travel time. Then we can update the optimal position to insert origin when enumerating value of $j$, after inserting $d_{r^+}$ before $v_j$, we further check whether inserting $o_{r^+}$ before $v_j$ is better choose than inserting $o_{r^+}$ in the sub-route $\langle v_1, v_2, \cdots, v_{j-1} \rangle$. We use $\pickloc{[k]}$ to indicates the best position to insert $o_{r^+}$ when insert $d_{r^+}$ before $v_k$, if this new potential $k$-th position do not satisfy either deadline constraint and capacity constraint for inserting $o_{r^+}$ then $\pickloc{[k]} = \nil$, otherwise we can calculate the value based this lemma:
\begin{lemma}
\label{lem:pickloc}
Given a potential vertex $v_k$ for inserting $d_{r^+}$ before it, we can update $\pickloc{[k]} = k$ if and only if (1) $\num{[k-1]} \leq c_w - c_{r^+}$; (2) $\pick{(o_{r^+})} <= e_{r^+}$ and $\pick{(o_{r^+})} + \query{(o_{r^+}, v_k, \pick{(o_{r^+})})} \leq \latest{[k]}$;  
(3) $\pick{(o_{r^+})} + \query{(o_{r^+}, v_k, \pick{(o_{r^+})})} <  \arr{'_{o_{r^+}}[k]}$, where $\arr{'_{o_{r^+}}[k]}$ is the arrival time to $v_k$ if inserting $o_{r^+}$ in positions of sub-route $\langle v_1, v_2, \cdots, v_{k-1} \rangle$, which can be calculated based on the travel time function from $v_{k-1}$ to $v_k$ as 

\vspace{-3ex}
\begin{gather*}
     \arr{'_{o_{r^+}}[k]} = \arr{'_{o_{r^+}}[k-1]} + \travel{(v_{k-1}, v_k,\arr{'_{o_{r^+}}[k-1]})}
\end{gather*}

\end{lemma}

\begin{proof}
\vspace{-1ex}
    Condition (1) checks the capacity constraint, guarantees at $v_k$ the worker has enough capacity to pickup $r^+$. Condition (2) checks the deadline constraints for inserting $o_{r^+}$ before $v_k$. Condition (3) guarantees that 
    inserting $o_{r^+}$ at the new potential $k$-th position will cause shorter travel time than inserting it in sub-route $\langle v_1, v_2, \cdots, v_{k-1} \rangle$, we can prove it based on the contradiction. For a new potential position $k$, if we insert $o_{r^+}$ before $v_k$, the new arrival time to $v_k$ is equal to  $\pick{(o_{r^+})} + \query{(o_{r^+}, v_k, \pick{(o_{r^+})})}$. Before taking $v_k$ into consideration, the best position for picking up the new request $r^+$ is $\pickloc{[k-1]}$, if insert $o_{r^+}$ at $\pickloc{[k-1]}$, we can get the arrival time at $v_k$ based on the compound travel time function from $v_{k-1}$, where $\arr{'_{o_{r^+}}[k]} = \arr{'_{o_{r^+}}[k-1]} + \travel{(v_{k-1}, v_k,\arr{'_{o_{r^+}}[k-1]})}$. If condition (3) is satisfied, we can get an earlier arrival time at $v_k$. If we assume $\pickloc{[k-1]}$ is the best position to insert $o_{r^+}$, therefore we can get a route with shorter travel time compare with inserting $o_{r^+}$ before $v_k$. However, based on the FIFO property, if $\pickloc{[k-1]}$ is chosen then we can get a route arrival later to $v_k$ and to the final vertex $v_n$ with larger travel time, which contradicts the assumption. The assumption does not hold, so $k$ is a better position to insert $o_r$, we can update $\pickloc{[k]} = k$.
\end{proof}

\subsubsection{Inserting new request's origin $o_{r^+}$ at $\pickloc{[k]}$}
\
\newline
 When enumerating each potential $k$-th position, we can update the best position to insert the new request's origin in $O(1)$ time based on \lemref{lem:pickloc}. To further check the feasibility for inserting $d_{r^+}$, we need to update the new arrival time to $v_k$ first based on the value of $\pickloc{[k]}$. If the value is $\nil$, there is no feasible positions to pickup and deliver $r^+$ until vertex $v_k$ in the current route, the arrival time remains the original value. If $k$ is the optimal position to insert $o_{r^+}$, that is $\pickloc{[k]} = k$, inserting $o_{r^+}$ before $v_k$, two shortest travel time queries $\query{}$ are invoked to get the pickup time from $v_{k-1}$ to $o_{r^+}$ and the new arrival time from $o_{r^+}$ to $v_k$. Otherwise, the value of $\pickloc{[k]}$ indicates the best position to insert the origin in sub-route before $v_k$. The new arrival time to $v_k$ can be updated based the compound travel time function and arrival time from previous vertex $v_{k-1}$.

We can update the new arrival time of $v_k$ based on the following rules in $o(1)$ time if $\pickloc{[k]} \neq \nil$. Case 1 indicates that $\pickloc{[k]} = k$, $o_{r^+}$ is also inserted before $v_k$. Case 2 represents that $\pickloc{[k]} = \pickloc{[k-1]}$, $o_{r^+}$ is inserted into sub-route $\langle v_1, v_2, \cdots, v_{k-1} \rangle$:

$\arr{'_{o_{r^+}}[k]} =$
\begin{equation}
\label{equ:pickloc}
    \begin{cases}
    \pick{(o_{r^+})} + \query{(o_{r^+}, v_k, \pick{(o_{r^+})})},& Case 1\\
    \arr{'_{o_{r^+}}[k-1]} + \travel{(v_{k-1}, v_k, \arr{'_{o_{r^+}}[k-1]})},& Case 2\\
    \end{cases}
\end{equation}

\subsubsection{Inserting new request's $d_{r^+}$ at $k$-th Position}
\
\newline
For each $k$-th position, we should consider the capacity constraint of the worker, the time constraint of new request and the new arrival time at $v_k$ to check the feasibility of inserting $d_{r^+}$ before $v_k$.
\begin{lemma}
\label{lem:ddlk}
The $k$-th position is a feasible position for inserting $d_{r^+}$ if and only if
(1) $num[k-1] \leq c_w - c_{r^+}$;
(2) $\deliver{(d_{r^+})} \leq e_{r^+}$;
(3) $\deliver{(d_{r^+})} + \query{(d_{r^+}, v_k, \deliver{(d_{r^+})})} \leq \latest{[k]}$
\end{lemma}
\begin{proof}
\vspace{-1ex}
    Condition (1) checks the capacity constraint. Condition (2) guarantees that deliver new request at $k$-th position will not violate the deadline constraint of the new request. After inserting $d_{r^+}$ before $v_k$, condition (3) checks the new arrival time to $v_k$ will not violate any deadline constraint for assigned requests delivered after $v_k$. 
\end{proof}

\vspace{-2ex}
If $k$ is a feasible position to insert $d_{r^+}$, we first update the new arrival time at $v_{k-1}$ which is affected by inserting $o_{r^+}$. Next we can get the deliver time $\deliver{(d_{r^+})}$ of the new request as:
\begin{equation}
\label{equ:deliverloc1}
\deliver{(d_{r^+})} = \arr{'_{o_{r^+}}[k-1]} + \query{(v_{k-1},d_{r^+}, \arr{'_{o_{r^+}}[k-1]})}
\end{equation}

Then we can update the new arrival time to $v_k$ after delivering the new request based on the query function $\query{}$. We use $\Delta_k$ to denote the arrival time at $v_k$, if we insert $o_{r^+}$ at $\pickloc{[k]}$ and insert $d_{r^+}$ at $k$,
\begin{equation}
\label{equ:deliverloc2}
\Delta_k = \deliver{(d_{r^+})} + \query{(d_{r^+}, v_k, \deliver{(d_{r^+})})}
\end{equation}

Since $v_k$ maintains the compound travel time function to the last vertex $v_n$, we can calculate the new arrival time to $v_n$ which is the $objective$ of the insertion $(\pickloc{[k]},k)$, 
\begin{equation}
\label{equ:obj}
obj = \Delta_k + \travel{(v_k, v_n, \Delta_k)}
\end{equation}

\begin{table}[t]
	\centering
	\caption{$\travel{}$ of $S_R$ in \algoref{alg:n}}
    \vspace{-1ex}
	\label{tab:alg3compound}
	{\scriptsize
	\begin{tabular}{|c|c|c|c|c|c|}
		\hline
		      & $ 1$ & $  2 $ \\
		\hline
		$o_1$ &   $\travel{(o_1, o_2, t)}$     & $\travel{(o_1, d_1, t)}$ \\
		\hline
		$o_2$ &   $\travel{(o_2, d_2, t)}$     & $\travel{(o_2, d_1, t)}$ \\
		\hline
		$d_2$ &   $\travel{(d_2, d_1, t)}$     & $\times$ \\
		\hline
		$d_1$ &   $\times$     & $\times$  \\
		\hline
	\end{tabular}}
  \vspace{-2ex}
\end{table}

\begin{algorithm}
\caption{Linear Time Insertion}
\label{alg:n}
\LinesNumbered 
\KwIn{a worker $w$ with a feasible route $S_R$ and a request $r^+$}
\KwOut{new route $S^*$ for $w$ and its travel cost $obj^*$}
$S^* \gets S_R, obj^* \gets +\infty$ \;
Initialize $\arr{}$, $\latest{}$, $\num{}$ by  $\equref{equ:arr}$ - $\equref{equ:num}$ \;
Compound travel time functions $\travel{(k,k+1, \cdot)}, \travel{(k,n,\cdot)} ,\forall k, 0< k < n$ by \defref{def:TravelFunction} \;
\For{$k \gets $ 1 to $n+1$}{
    \lIf{\lemref{lem:pickloc} (1) or (2) violated}{continue}
    \lIf{\lemref{lem:pickloc} (3)}{$\pickloc{[k]} \gets k$}
    Update $\arr{'_{o_{r^+}}[k]}$ after inserting $o_{r^+}$ based on \equref{equ:pickloc} \;
    
    \If{\lemref{lem:ddlk} (1)(2)(3)}{
        Calculate $\deliver{(d_{r^+})}$ after inserting $o_{r^+}$ based on \equref{equ:deliverloc1} \;
        Calculate $\Delta_k$ of insertion $(\pickloc{[k]},k)$ based on \equref{equ:deliverloc2} \;

        \If{k = n+1}{$obj \gets \Delta_k$ \;}
        \Else{$obj \gets calculate by $ \equref{equ:obj} \;}        

        \If{$obj < obj^*$}{$obj^* \gets obj, i^* \gets \pickloc{[k]}, j^* \gets k$ \;} 
    }

}
\If{$obj^* < +\infty$}{ $S^* \gets$ insert $o_r$ at $i^*$-th and $d_r$ at $j^*$-th in $S_R$ \;}
\KwRet{$S^*$, $obj^*$}\;
\end{algorithm}

\subsubsection{Algorithm Details}
\
\newline
\fakeparagraph{Main Idea} We enumerate each $k$-th position along $v_0$ to $v_n$ in the current route, we update $\pickloc{[k]}$ value to indicates the best position to insert origin of $r^+$ until $v_k$ in the current route. Simultaneously, for each $k$, we check the feasibility and calculate the objective of this new generated route by inserting origin of $r^+$ at $\pickloc{[k]}$ and inserting destination at $k$ position. For this new potential route where $r^+$ is picked up at $\pickloc{[k]}$ and delivered before $v_k$, the feasibility is checked by \lemref{lem:ddlk} and the objective of the new total travel time is calculated by \equref{equ:deliverloc2} and \equref{equ:obj}. It take $O(1)$ time to find the optimal position to insert the new request's origin, check feasibility and calculate the objective of this new route.

\fakeparagraph{Complexity Analysis}
Line 2 initializes $\arr{}, \latest{}, \num{}$. Line 3 compounds the travel time functions between $v_k$ to $ v_{k+1}$ and $v_k$ to $v_n$ for $v_k \in S_R$. Line 4 enumerates each potential position to insert destination of $r^+$, it takes $O(n)$ time. The deadline and capacity constraints checking for inserting $o_{r^+}$ at the new potential $k$-th position will take in line 5, if it violates constraints the best position to insert new request's origin is in sub-route $\langle v_1, v_2, \dots, v_{k-1} \rangle$. Otherwise the $k$ will be the best position to insert the origin in line 6. After inserting origin, we need to update the arrival time affected by it in line 7. Then we check the feasibility and calculate the new arrival time to $v_k$ after inserting $d_{r^+}$ before $v_k$ in lines 8-10, and each line takes $O(1)$ time. After inserting both origin and destination of the new request, we calculate the travel time of new route from lines 11 - 14. Therefore, the total time cost of \algoref{alg:n} is $O(n)$. For the space cost, in line 2 for each vertex $v_k$, there are two compound travel time functions need to be initialized, therefore one worker should maintain $O(n)$ compound travel time functions, and the space cost is $O(n)$.

\begin{figure}[t]
	\centering
    \begin{subfigure}[b]{0.25\textwidth}
		\includegraphics[width=\textwidth]{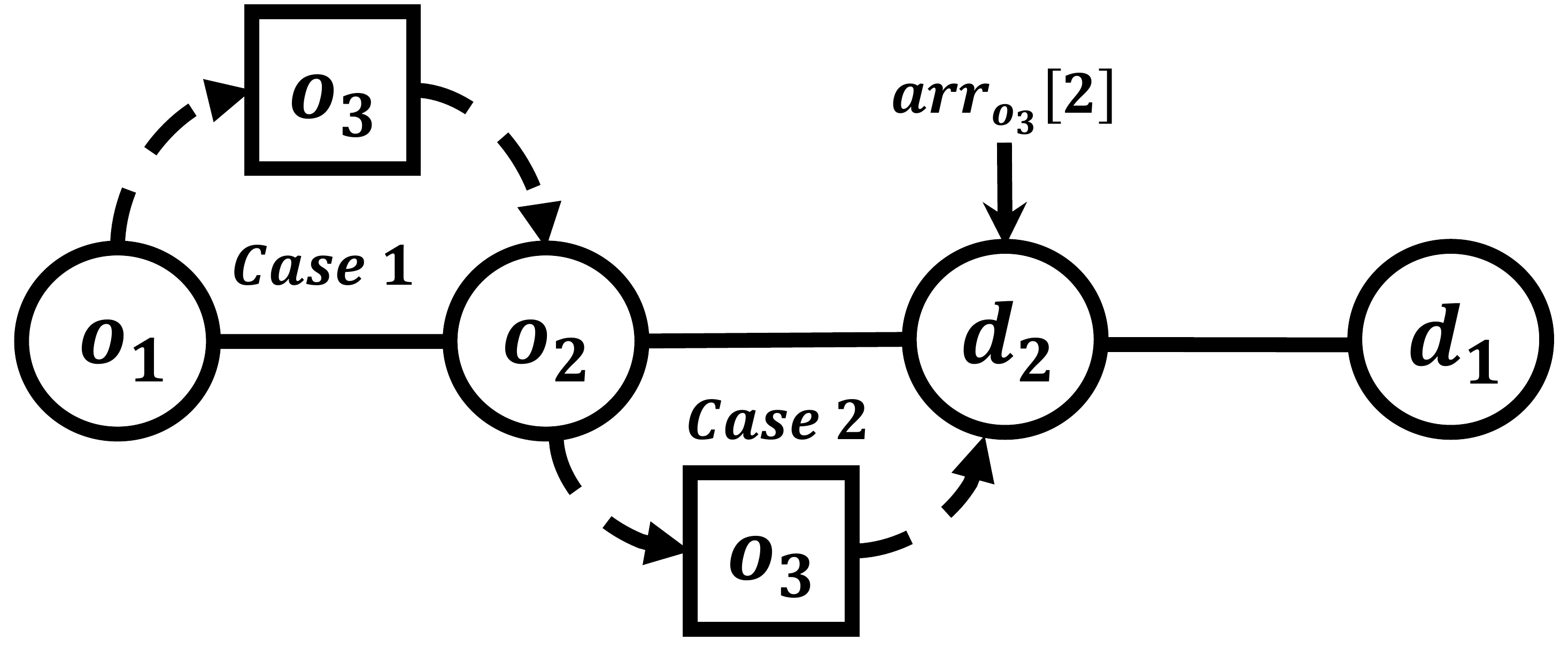}
		\vspace{-4ex}
		\caption{\footnotesize{When $k=2$, update $\pickloc{[2]}$}}
		\label{fig: example5a}
		\vspace{1ex}
	\end{subfigure}

	\begin{subfigure}[b]{0.25\textwidth}
		\includegraphics[width=\textwidth]{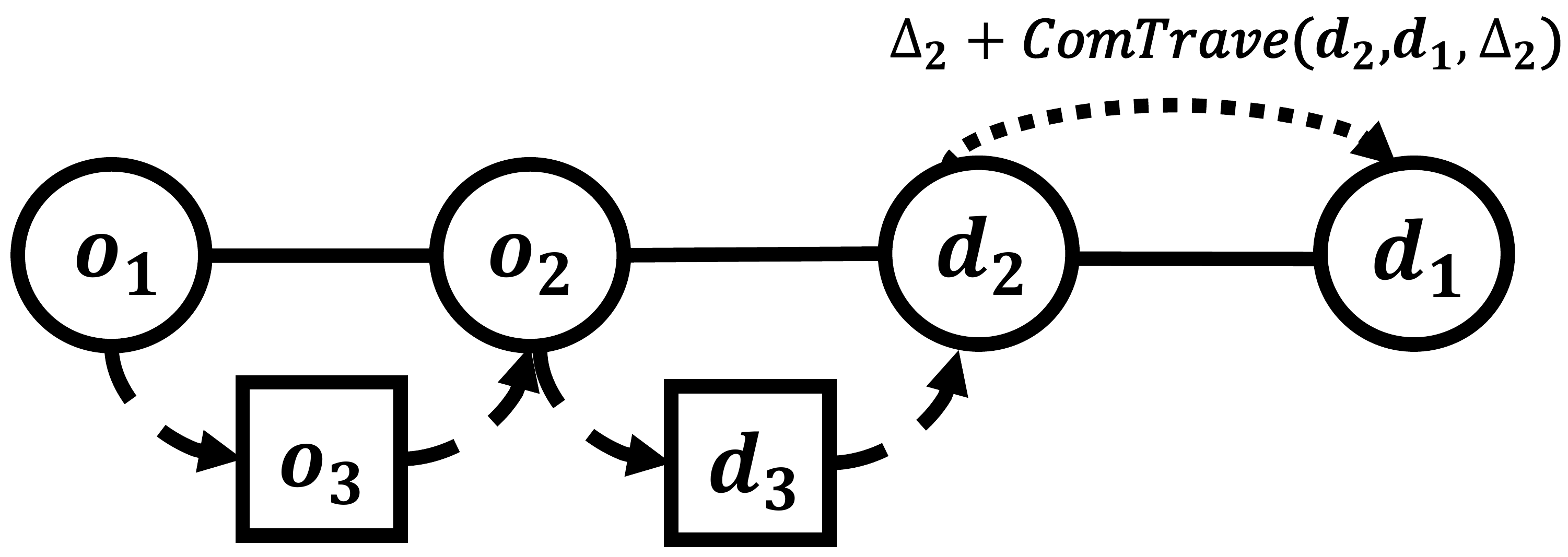}
		\vspace{-4ex}
		\caption{\footnotesize{$obj$ of insertion $(\pickloc{[2]},2)$}}
		\label{fig:example5b}
	\end{subfigure}
	\vspace{-2ex}
	\caption{Insertion of $(\pickloc{[2]},2)$.}
	\label{fig:example4}
	\vspace{-4ex}
\end{figure}

\fakeparagraph{Example 4} \textit{ Back to the settings in Example 3. For the feasible route with the shortest travel time $S_R =  \langle o_1, o_2, d_2, d_1 \rangle$, the compound travel time functions along $S_R$ are shown in \tabref{tab:alg3compound}. Each vertex compounds and maintains at most 2 functions, one for the next vertex and another for the last vertex in $S_R$. For easy of presentation, we first make the assumption that insertion $(1,2)$ is the optimal location pair to serve $r_3$. When $k=1$, it means $d_3$ is inserted before $o_2$. Due to the order constraint, $o_3$ can only inserted before $o_2$. We first calculate new arrival time at $o_2$. If $o_3$ is inserted before it, $\arr{'_{o_{3}}[1]} = \pick{(o_{3})} + \query{(o_{3}, o_2, \pick{(o_{3})})}$ and $\pickloc{[1]} = 1$. When $k=2$, as \figref{fig: example5a} shows, there is a new possible location before $d_2$ for inserting $o_3$. Because of the FIFO property, if we can get an earlier arrival time at $d_2$, then we can get an earlier arrival time at $d_1$. Therefore, when $k=2$, we need to check which location is the best for inserting $o_3$, \ie{ $\pickloc{[2]} = 2$ or $1$}. In case 1, $o_3$ is inserted before $o_2$, the new arrival time at $d_2$ can be calculated as  $\arr{'_{o_{3}}[1]}+\travel{(o_2, d_2, \arr{'_{o_{3}}[1]})}$. Otherwise, in case 2, $o_3$ is inserted before $o_2$, the new arrival time at $d_2$ can be calculated as $\pick{(o_{3})} + \query{(o_{3}, d_2, \pick{(o_{3})})}$. Based on our assumption, $\arr{'_{o_{3}}[1]}+\travel{(o_2, d_2, \arr{'_{o_{3}}[1]})}$ is less than $\pick{(o_{3})} + \query{(o_{3}, d_2, \pick{(o_{3})})}$, so $\arr{'_{o_{3}}[2]}) = \arr{'_{o_{3}}[1]}+\travel{(o_2, d_2, \arr{'_{o_{3}}[1]})}$ and $\pickloc{[2]} = 1$. Then we try to insert $d_3$ before $d_2$, and we can first get $\deliver{(d_{3})}$ as $\arr{'_{o_{3}}[1]}+ \query{}$ ${(o_2,o_3,\arr{'_{o_{3}}[1]})}$. Consider inserting $o_3$ before $o_2$ and inserting $d_3$ before $d_2$ in \figref{fig:example5b}, \ie{ the insertion of $(1,2)$}. We can get arrival time at $d_2$ as $\Delta_2 = \deliver{(d_{3})} + \query{(d_{3}, d_2, \deliver{(d_{3})})}$. Finally, based on the travel time function from $d_2$ to the last vertex, we can get the objective value as $\Delta_2 + \travel{(d_2, d_1, \Delta_2)}$.
}

%% file: 06_experiment.tex
\vspace{-1ex}
\section{Experiment Study}\label{sec:experiment}
In this section, we first present the experimental setup for our study, then we demonstrate the discuss our experimental results. 

\begin{table}[t]
\vspace{-1ex}
	\centering
	\caption{Statistics of datasets.}
    \vspace{-2ex}
	\label{table:dataset}
    \resizebox{0.45\textwidth}{!} {%
	\begin{tabular}{|c|c|c|c|}
		\hline
		Dataset & \#(Vertices) & \#(Edges) & \#(Interpolation Points)\\
		\hline														
		\textit{Chengdu} & $423,434$ & $913,718$ & $1,411,569$ \\
		\hline
		\textit{Haikou} & $41,542$ & $89,206$ & $138,083$ \\
		\hline
	\end{tabular}
	}
\vspace{-2ex}
\end{table}

\vspace{-2ex}
\subsection{Experimental Setup}
\fakeparagraph{Datasets} The evaluations of proposed methods are conducted on two real-world datasets. The first one is a public dataset collected in Chengdu City, China, which is published through GAIA \cite{Gaia} initiative by DiDi Chuxing \cite{Didi}. The second dataset is collected in Haikou City, China. We download the road networks of two cities from OpenStreetMap \cite{Openstreetmap}, then generate the time-dependent road networks follow settings and procedures in the previous work \cite{shortestquery}, the number of edges in the datasets varies from 80,000 to more than 900,000. We set the time domain as 86,400 seconds which is the whole day. For the edges of time-dependent road networks, the number of interpolation points of the associated time-dependent weight functions is 138,083 and 1,411,569 respectively. We use the TDSP query function in \cite{shortestquery} as the shortest travel time query. The statistics of the datasets are shown in \tabref{table:dataset}.

We choose the data from the date which has the largest number of requests during the daytime \ie{ 8:00am - 18:00pm} for evaluation. Two datasets are represented as $\chengdu$ and $\haikou$ respectively. Each request in the dataset is a tuple consists of a origin, destination and a release time. Origins and destinations of requests are pairs of latitudes and longitudes locate in the city, and we map the request's origins and destinations to the nearest vertex of the city road network. For each request, we vary the time period from release time to deadline from 10 to 30 minutes. During the time period $e_r - t_r$, the request can be served by the ridesharing service with other assigned requests. The start vertex of a worker is randomly chosen from the vertices of the road network, the capacity of each worker varies from 3 to 20. \tabref{table:parameters} summarizes the major parameters of this experiment, which are also used in existing work \cite{huangyan2014}, \cite{parameter-dasfaa2018}, \cite{parameter-aaai2019}, \cite{parameter-tkde2019}, \cite{parameter-vldb2020}, \cite{yuxiang2020tkde}. Recently, there is a popular ridesharing service called ``bus pooling'' attracts a lot of attentions \cite{BlaBlaCar}, which is a transport solution can be described as ``Uber for buses''. Therefore, we set the maximum capacity of a worker as 20.

\begin{table}[t]
\vspace{-1ex}
	\centering
	\caption{Compared Algorithms.}
    \vspace{-2ex}
	\label{table:algorithm}
    \resizebox{0.4\textwidth}{!} {%
	\begin{tabular}{|c|c|c|}
		\hline
		Algorithms & Time Complexity & Space Complexity \\
		\hline														
		\cubic & $O(n^3)$ & $O(n)$  \\
		\hline
		\quadra & $O(n^2)$ & $O(n^2)$ \\
		\hline
		\linear & $O(n)$ & $O(n)$ \\
		\hline
	\end{tabular}
	}
\vspace{-2ex}
\end{table}

\fakeparagraph{Compared Algorithms} We compare the following algorithms in this section. \tabref{table:algorithm} summaries these algorithms. 
\begin{itemize}
  \item
  \textit{Cubic Time Algorithm (\algoref{alg:n3} of this paper) \cite{yuxiang2019icde}.}  The baseline method implements the insertion operator over time-dependent road network. It enumerates all possible $(i,j)$ pairs to generate new potential routes. For each route, the shortest travel time query over a time-dependent road network is recursively invoked to calculate the arrival time and check the constraints of this route. For the space cost, the worker only needs to maintain the current feasible route.
  \item
  \textit{Quadratic Time Algorithm (\algoref{alg:n2} of this paper).} This method generates new potential routes by enumerating all possible $(i,j)$ pairs. To reduce the time complexity and avoid invoking large number of shortest travel time queries, the worker compounds the time-dependent edge weights and maintain the compound travel time functions. Except for the auxiliary arrays, the space cost is dominated by the compound travel functions between any two vertexes pair along the worker's feasible route, so the space cost is $O(n^2)$.
  \item
  \textit{Linear Time Algorithm (\algoref{alg:n} of this paper).} Instead of enumerating both $(i,j)$ values, this method enumerates $j$ and find the optimal $i$ in $O(1)$ time. The number of invoking the shortest travel time query is further decreased. The time and space cost both are reduced to $O(n)$, because this method maintains the compound travel time functions form each vertex to both its next vertex and the last vertex along the feasible route, instead  of checking all vertexes pairs.

\end{itemize}

\fakeparagraph{Metrics} The following metrics are critical for evaluating the performances of time-dependent insertion: (1) the impact of the capacity $c_w$ of the worker; (2) the impact of the time period $e_r - t_r$ for each request; (3) the impact of the number of requests; (4) the impact of the number of workers to study the scalability. We choose the metrics requiring evaluation: (1) the number of invoking shortest travel time queries $\query{}$; (2) insertion time; (3) response time; (4) memory cost. The shortest travel time query over the time-dependent road network is the bottleneck of time-dependent insertion, we avoid invoking large number of $\query{}$ to improve the efficiency. As for metric insertion time, in \textit{Cubic Time Algorithm} and \textit{Quadratic Time Algorithm}, it indicates the average time to check the feasibility and calculate the objective for each possible $(i,j)$ pair along the current route. In \textit{Linear Time Algorithm}, it represents the average time for each $j$ value to check the feasibility then find the optimal $i$ and calculate the objective along the current route. Response time and memory cost are both used as metrics in many ride-sharing applications\cite{yuzheng2013}\cite{huangyan2014}\cite{yuxiang2018vldb}. Response time is the average time to assign each request. In \textit{Cubic Time Algorithm}, the memory cost is the memory used for storing the route of the worker. The memory cost of both \textit{Quadratic Time Algorithm} and \textit{Linear Time Algorithm} is dominated by the maintained compound travel time functions of the worker, the worker also stores the route and auxiliary arrays.

\fakeparagraph{Implementation} The experiments are conducted on a server with 40 Intel(R)
Xeon(R) E5 2.30GHz processors with hyper-threading enabled and 128GB memory. Three proposed insertion operators are implemented in GNU C++. Following the setup in \cite{DBLP:conf/kdd/DuTZTZ18,DBLP:conf/waim/GaoTSSCX16,DBLP:conf/sigmod/SheT0S17}, each experiment is repeated 10 times and the average results are reported.

\subsection{Experimental Results}

\begin{table}[t]
\centering
\caption{Parameter settings.}
\vspace{-2ex}
\label{table:parameters}
\resizebox{0.4\textwidth}{!} {%
\begin{tabular}{|c|c|}
	\hline
	Parameters & Settings\\
	\hline
	Capacity $c_w$ & 3, \textbf{5}, 10, 15, 20\\
	\hline
	Time period : $e_r - t_r$ (minute) & 10, \textbf{15}, 20, 25, 30\\
	\hline
	Number of requests  & 20k, \textbf{40k}, 60k, 80k, 100k\\
	\hline
	Scalability: number of workers & 100, 200, 300, 400, 500 \\
	\hline
\end{tabular}
}
\vspace{-3ex}
\end{table}

\begin{figure}[t]
	\centering
	\begin{subfigure}[b]{0.22\textwidth}
		\includegraphics[width=\textwidth]{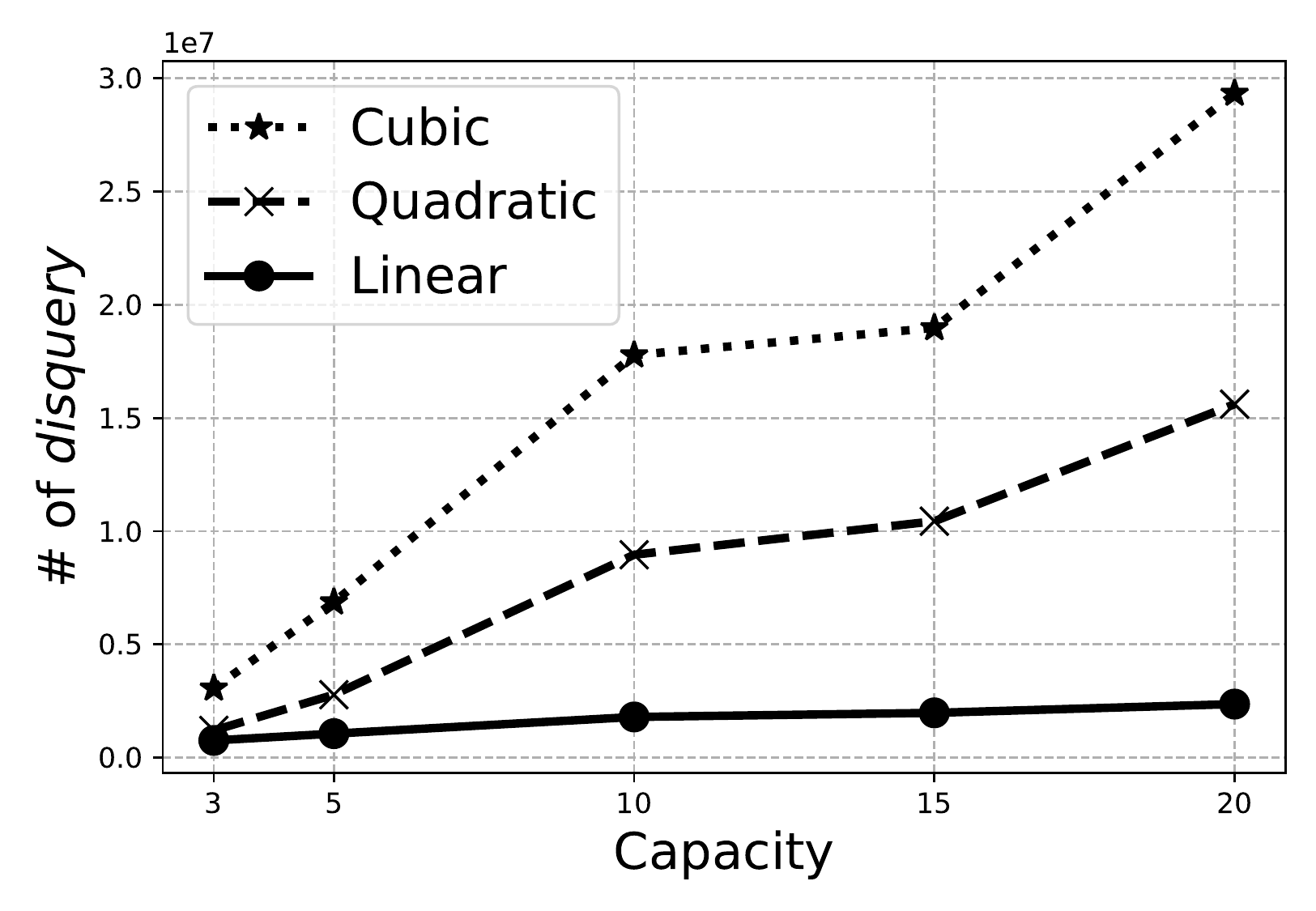}
		\vspace{-4ex}
		\caption{\footnotesize{Number of \query{} on \chengdu}}
	\end{subfigure}
	~~
    \begin{subfigure}[b]{0.22\textwidth}
		\includegraphics[width=\textwidth]{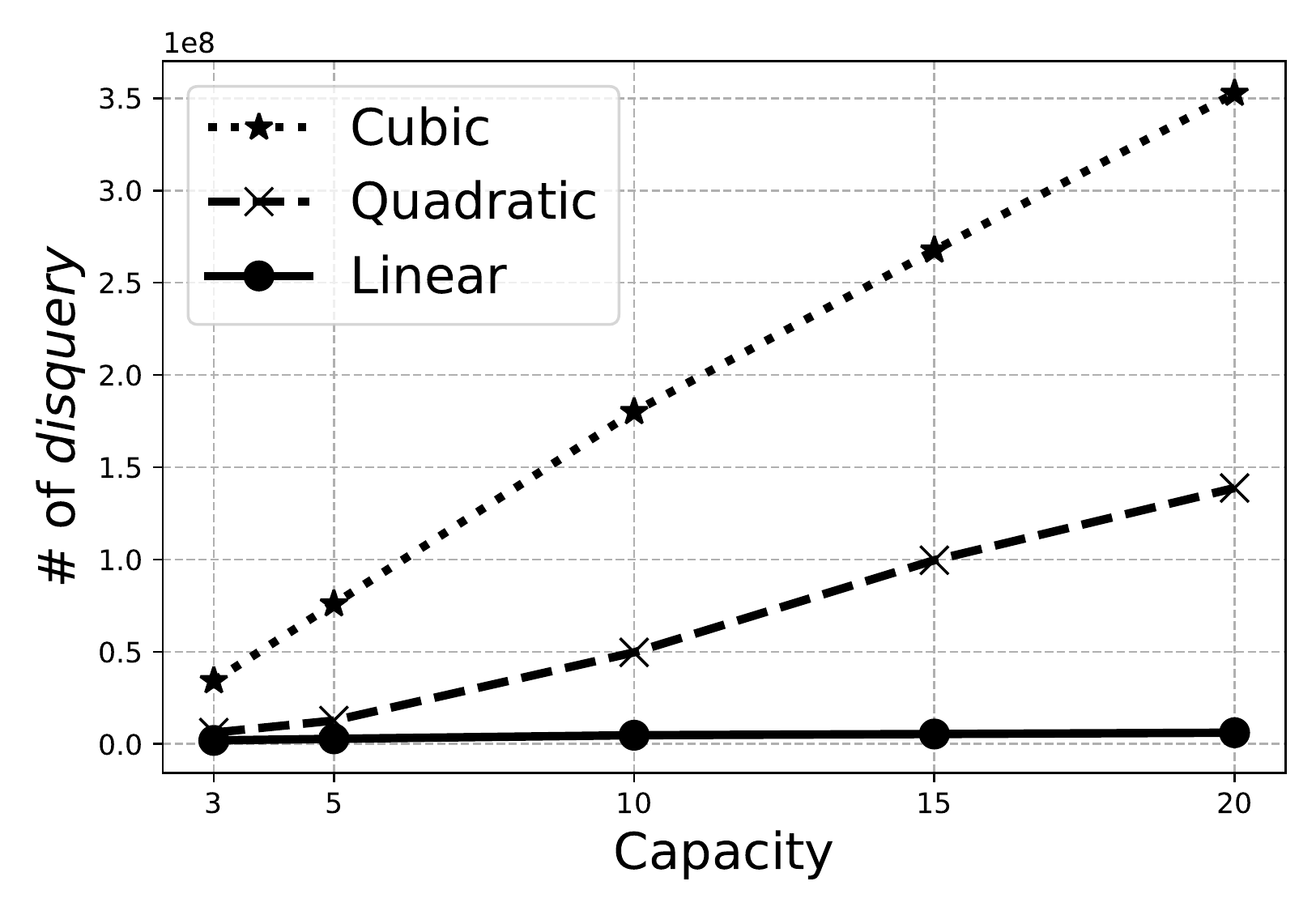}
		\vspace{-4ex}
		\caption{\footnotesize{Number of \query{} on \haikou}}
	\end{subfigure}
	
	\begin{subfigure}[b]{0.22\textwidth}
		\includegraphics[width=\textwidth]{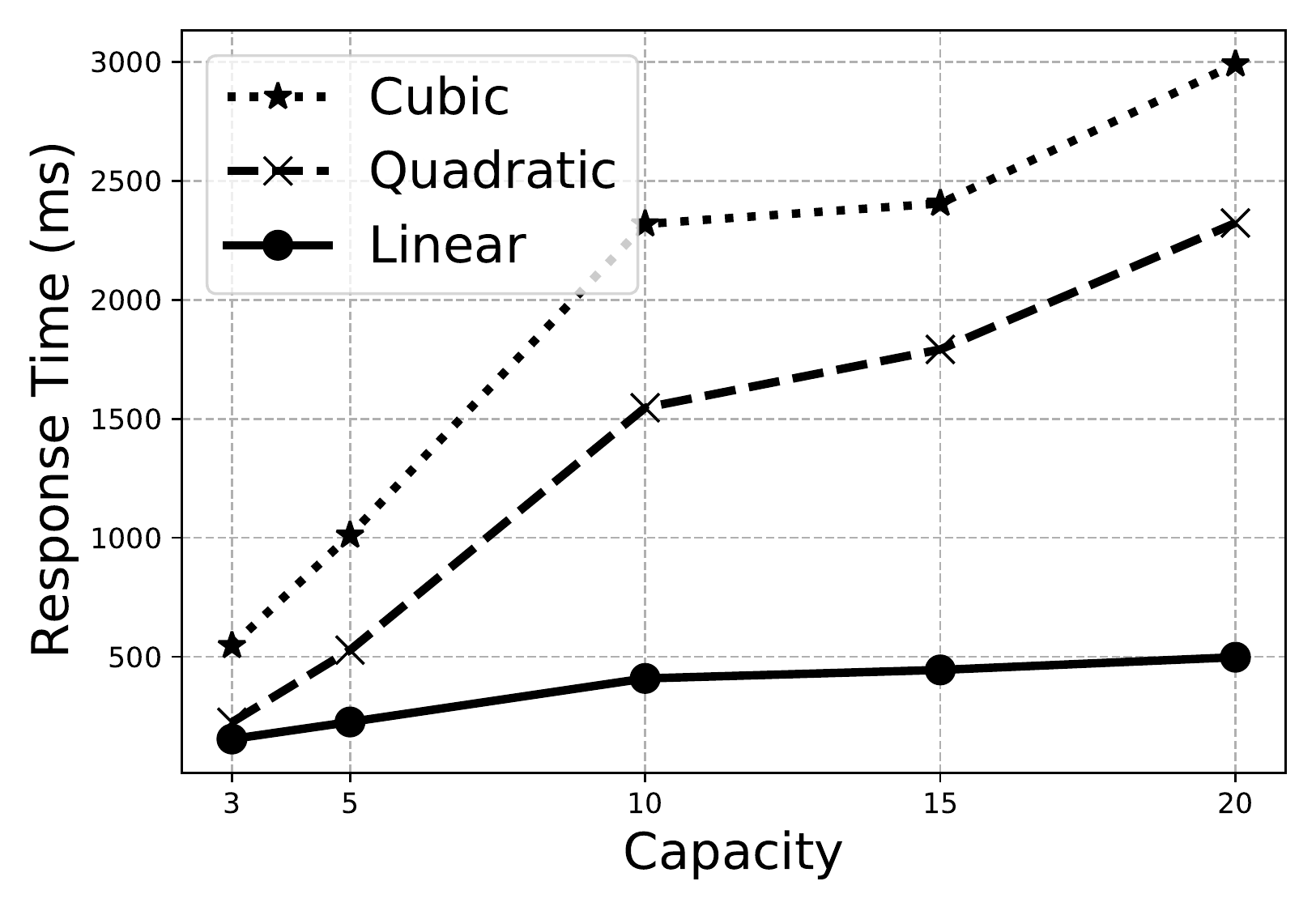}
		\vspace{-4ex}
		\caption{\footnotesize{Response time on \chengdu}}
	\end{subfigure}
	~~
    \begin{subfigure}[b]{0.22\textwidth}
		\includegraphics[width=\textwidth]{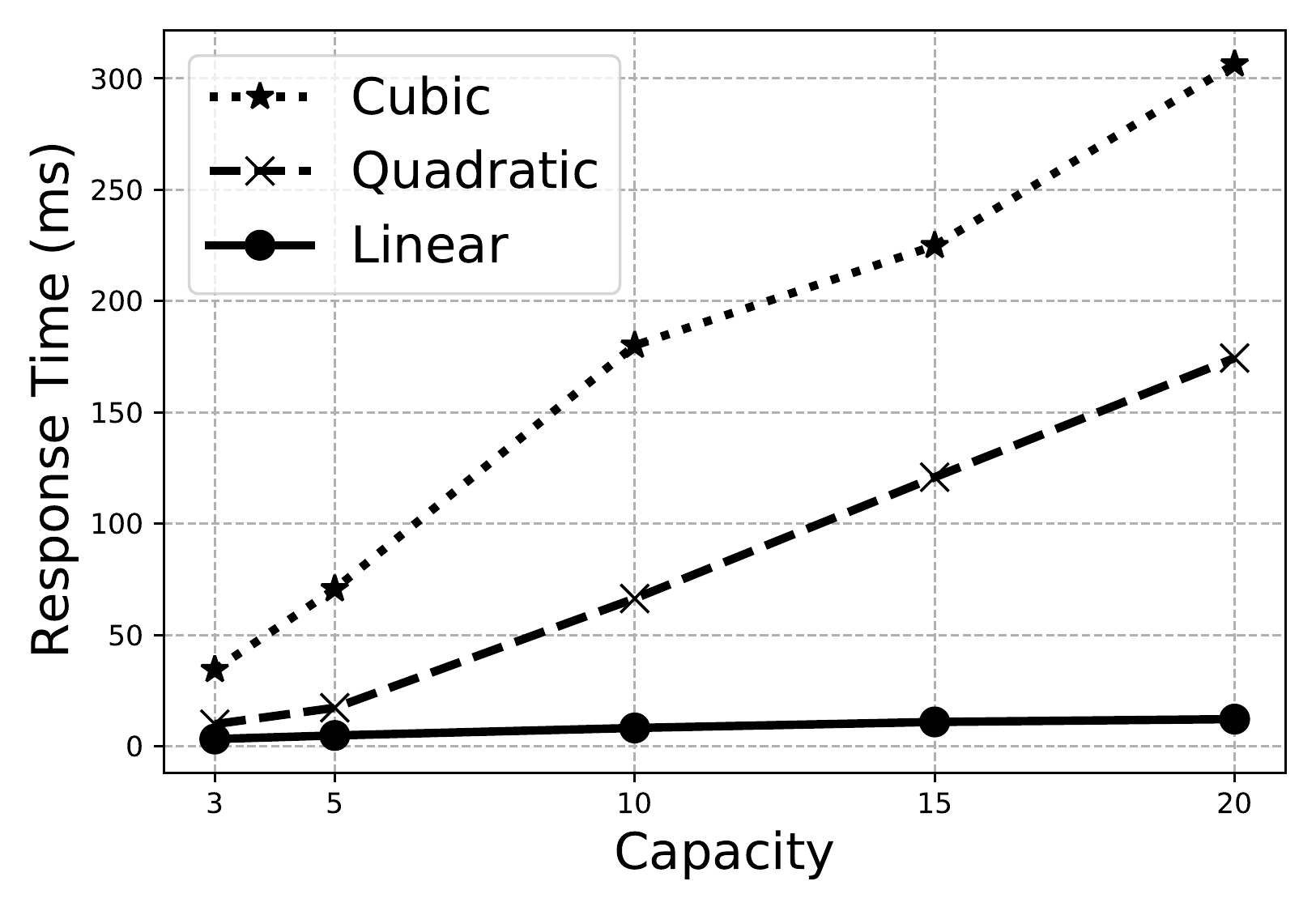}
		\vspace{-4ex}
		\caption{\footnotesize{Response time on \haikou}}
	\end{subfigure}
	
	\begin{subfigure}[b]{0.22\textwidth}
		\includegraphics[width=\textwidth]{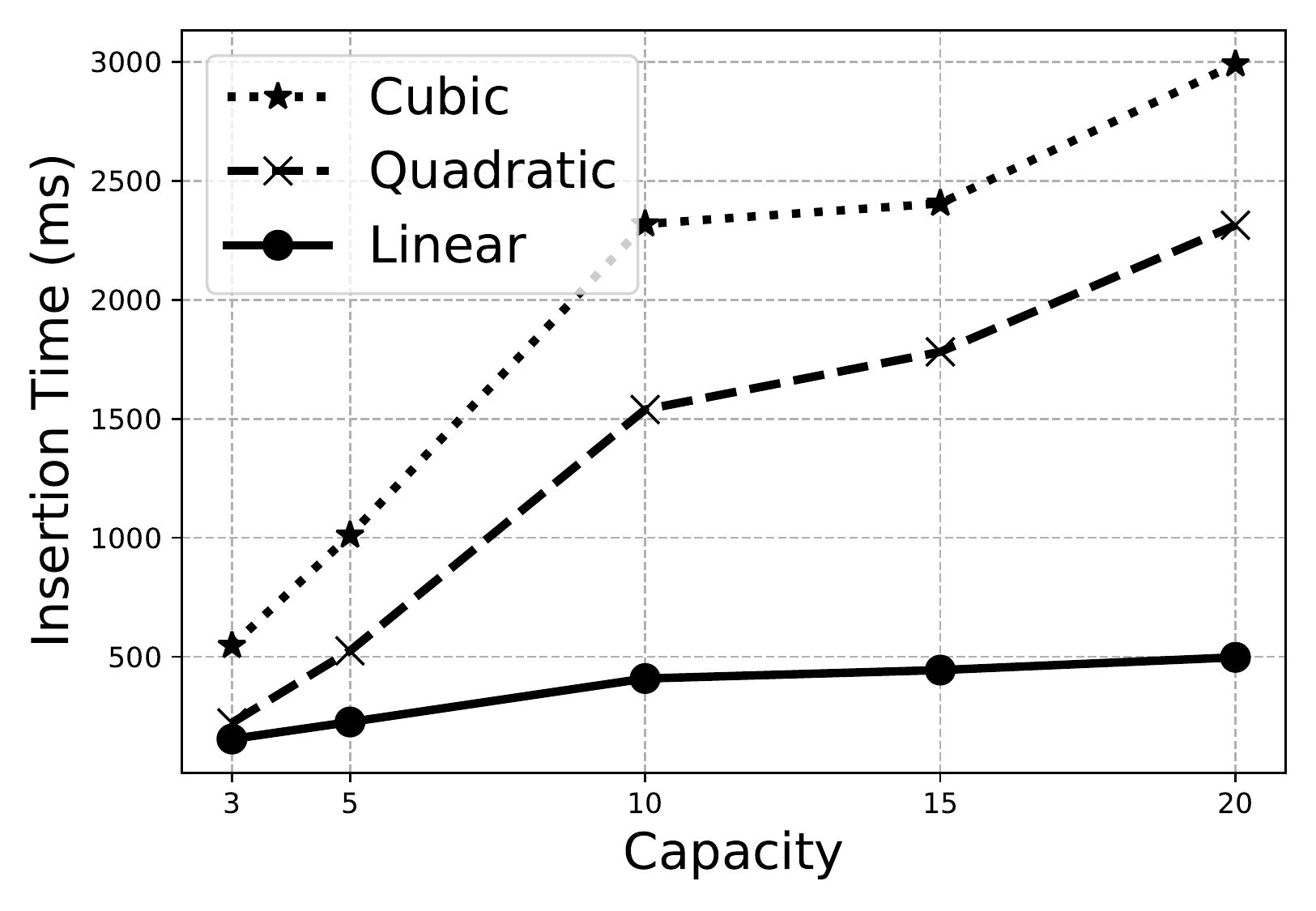}
		\vspace{-4ex}
		\caption{\footnotesize{Insertion time on \chengdu}}
	\end{subfigure}
	~~
    \begin{subfigure}[b]{0.22\textwidth}
		\includegraphics[width=\textwidth]{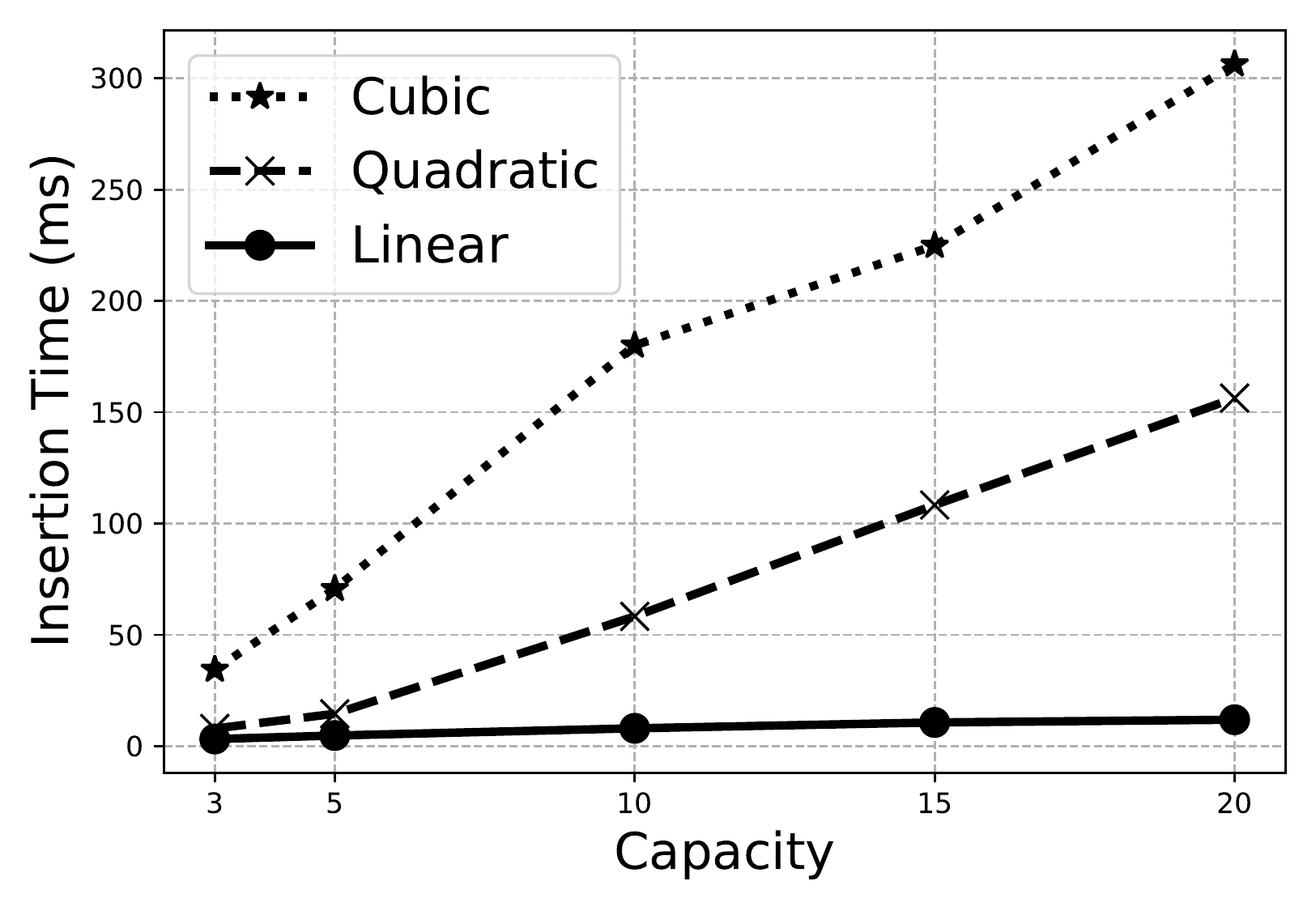}
		\vspace{-4ex}
		\caption{\footnotesize{Insertion time on \haikou}}
	\end{subfigure}	
	
	\begin{subfigure}[b]{0.22\textwidth}
		\includegraphics[width=\textwidth]{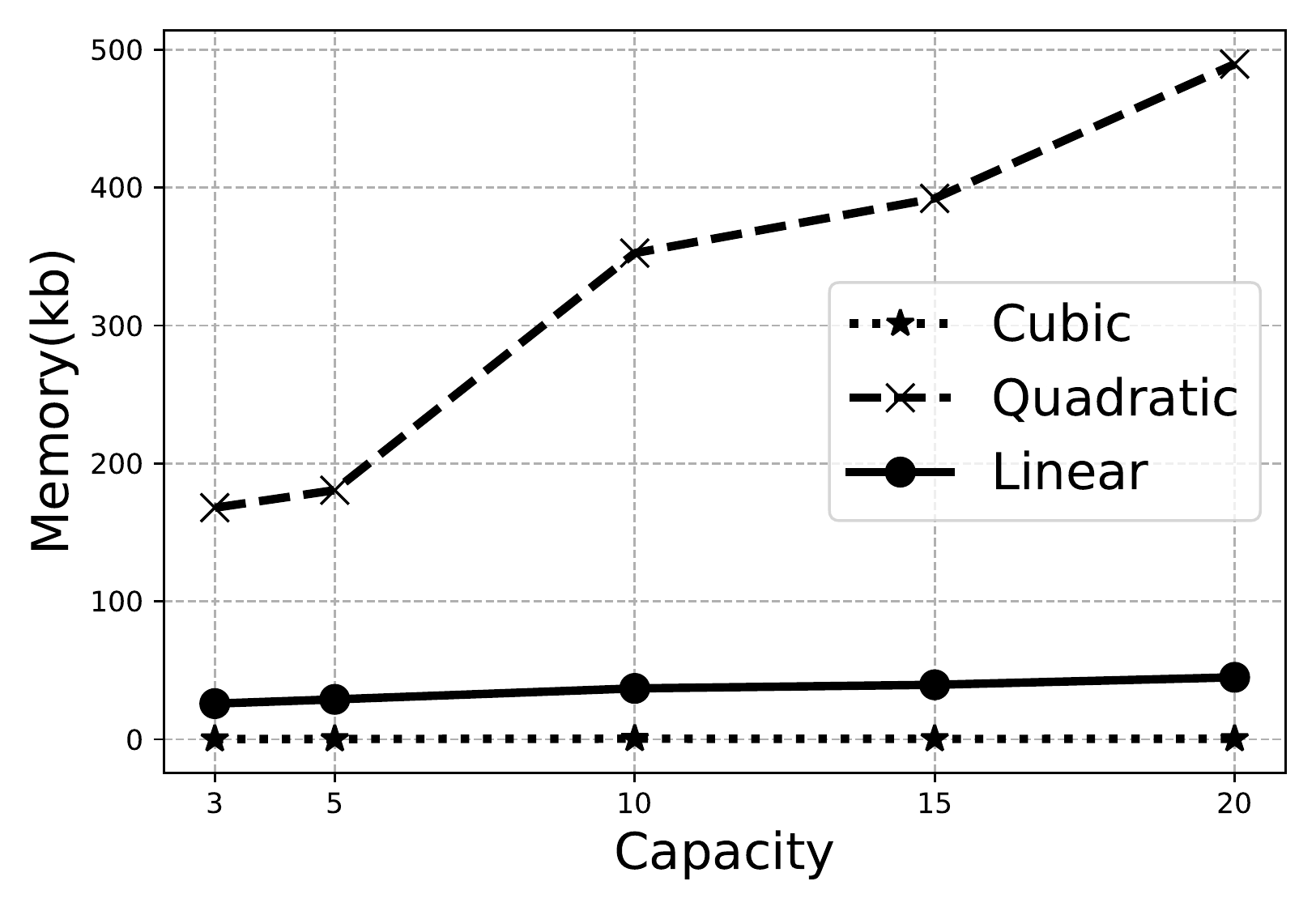}
		\vspace{-4ex}
		\caption{\footnotesize{Memory cost on \chengdu}}
	\end{subfigure}
	~~
    \begin{subfigure}[b]{0.22\textwidth}
		\includegraphics[width=\textwidth]{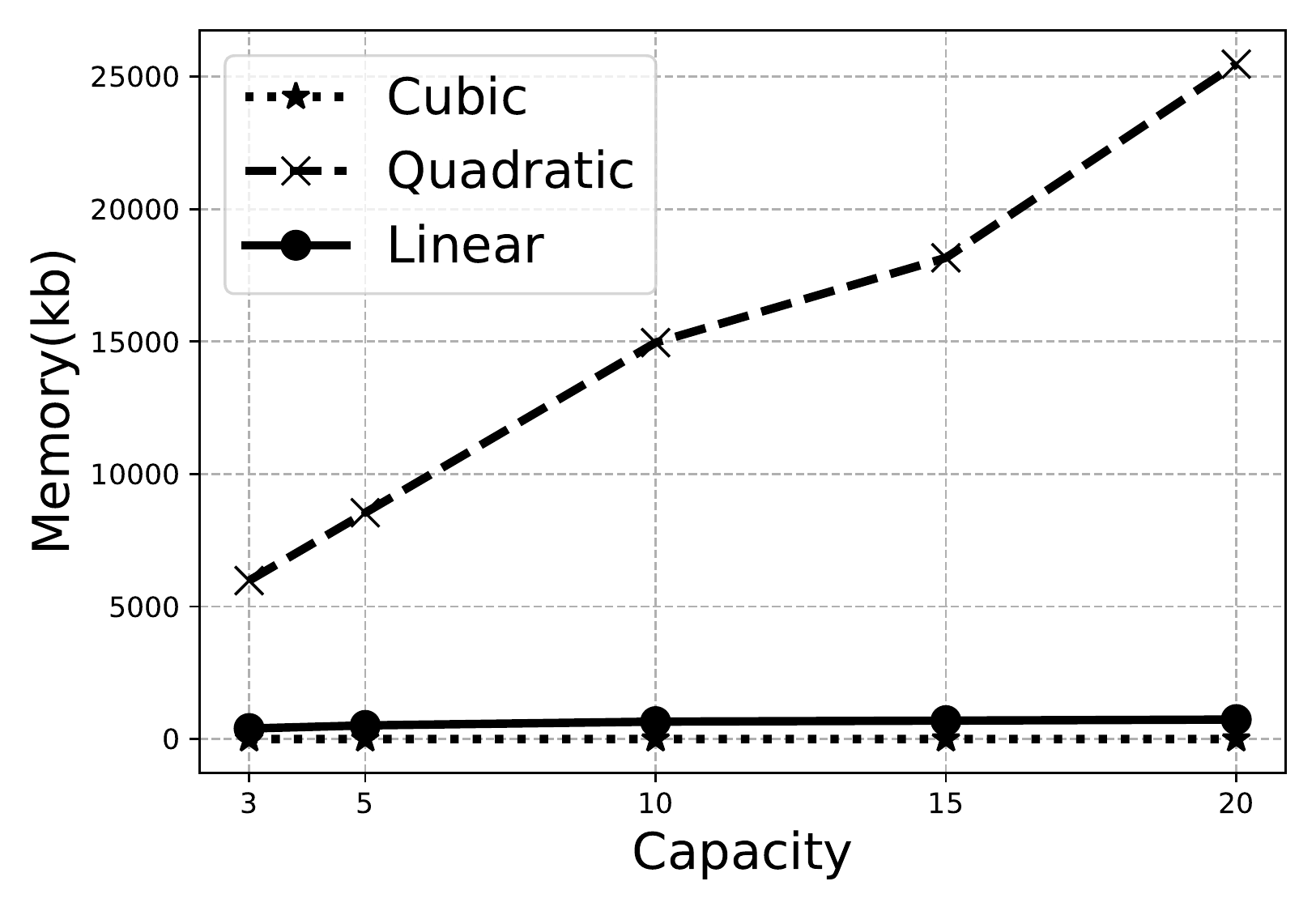}
		\vspace{-4ex}
		\caption{\footnotesize{Memory cost on \haikou}}
	\end{subfigure}	
	
	\vspace{-2ex}
	\caption{Results of varying capacity of the worker}
	\label{fig:capacity}
	\vspace{-3ex}
\end{figure}

\fakeparagraph{Impact of Capacity of Workers} \figref{fig:capacity} shows the results of varying the capacity of workers on  $\chengdu$ and $\haikou$. \textit{Linear Time Algorithm} has the minimum number of invoking the shortest travel time queries, comparing with the baseline method. \textit{Linear Time Algorithm} reduces 75.24\% - 91.94\% of the number of invoking $\query{}$ on these two datasets. The significant reduction in the number of invoking $\query{}$ proves the effectiveness of the compound travel time function. $\linear$ has the fastest insertion time correspondingly, which is up to 6.01 times faster on $\chengdu$, and 25.79 times faster on $\haikou$ than $\cubic$, respectively. We can observe that $\quadra$ consumes extremely larger memory cost than other two methods, this is because the memory cost of time-dependent insertion is dominated by the compound travel time functions. $\linear$ reduces 90.8\% and 97.1\% memory cost on $\chengdu$ and $\haikou$, respectively, compared to $\quadra$. With the increase in the capacity of the worker, the time cost and memory cost grow in the same trend, while $\linear$ always consumes the least time and consumes much less memory (no more than 45 KB on $\chengdu$) than $\quadra$. Note that the insertion time and response time of $\cubic$, the memory cost of $\quadra$ changes dramatically in both $\chengdu$ and $\haikou$ due to $O(n^3)$ time complexity and $O(n^2)$ space complexity, respectively.

\begin{figure}[t]
	\centering
	\begin{subfigure}[b]{0.22\textwidth}
		\includegraphics[width=\textwidth]{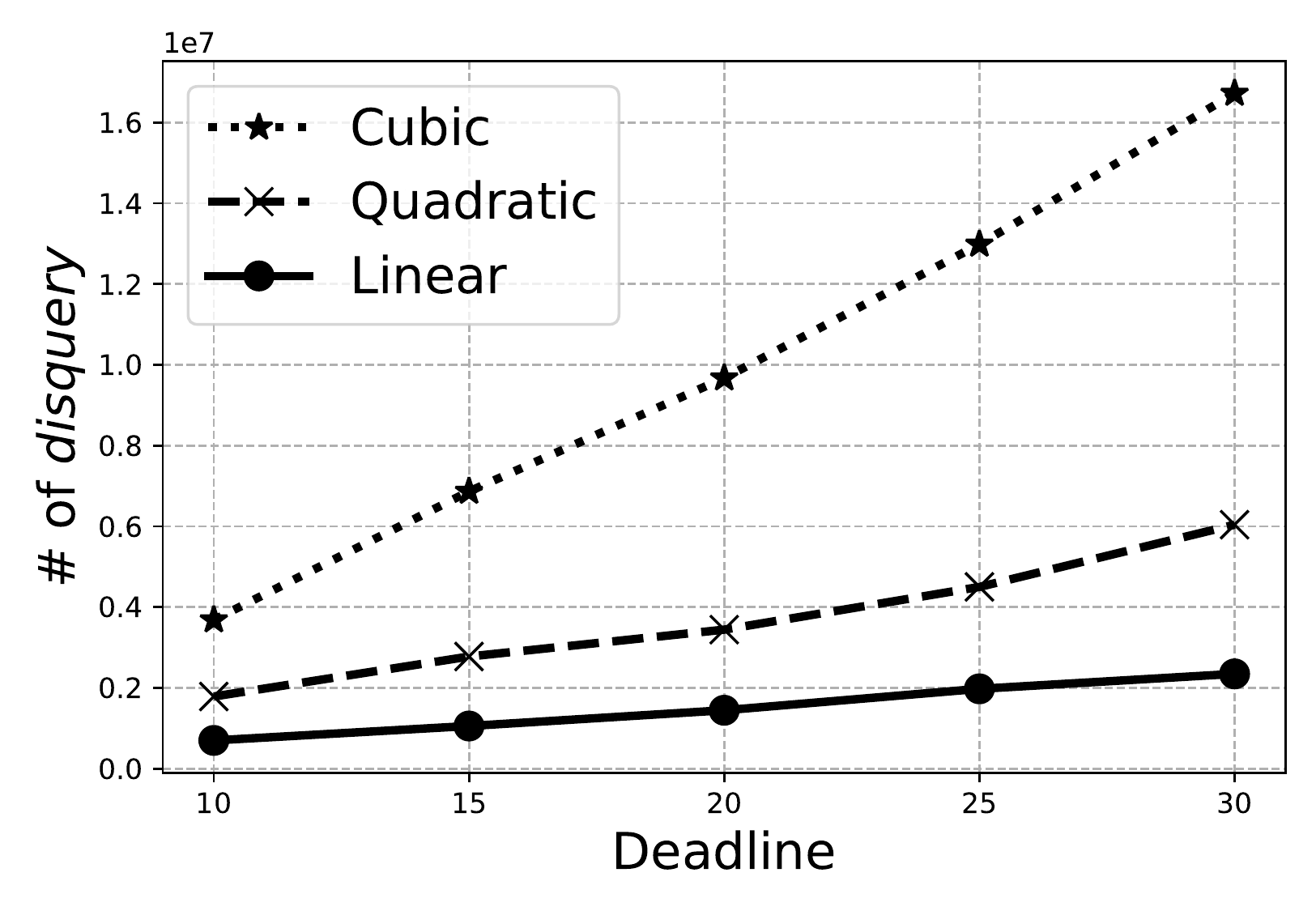}
		\vspace{-4ex}
		\caption{\footnotesize{Number of \query{} on \chengdu}}
	\end{subfigure}
	~~
    \begin{subfigure}[b]{0.22\textwidth}
		\includegraphics[width=\textwidth]{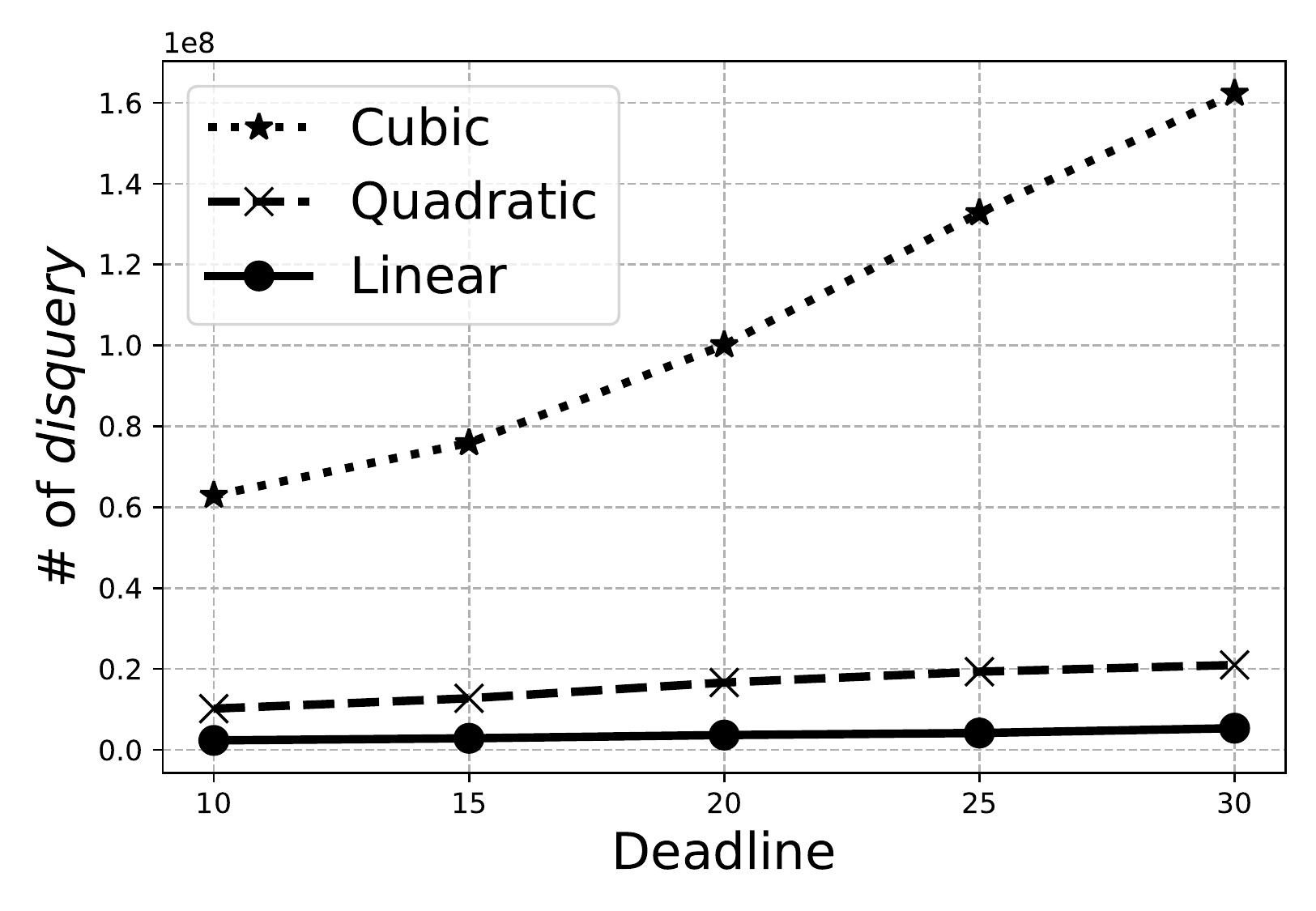}
		\vspace{-4ex}
		\caption{\footnotesize{Number of \query{} on \haikou}}
	\end{subfigure}
	
	\begin{subfigure}[b]{0.22\textwidth}
		\includegraphics[width=\textwidth]{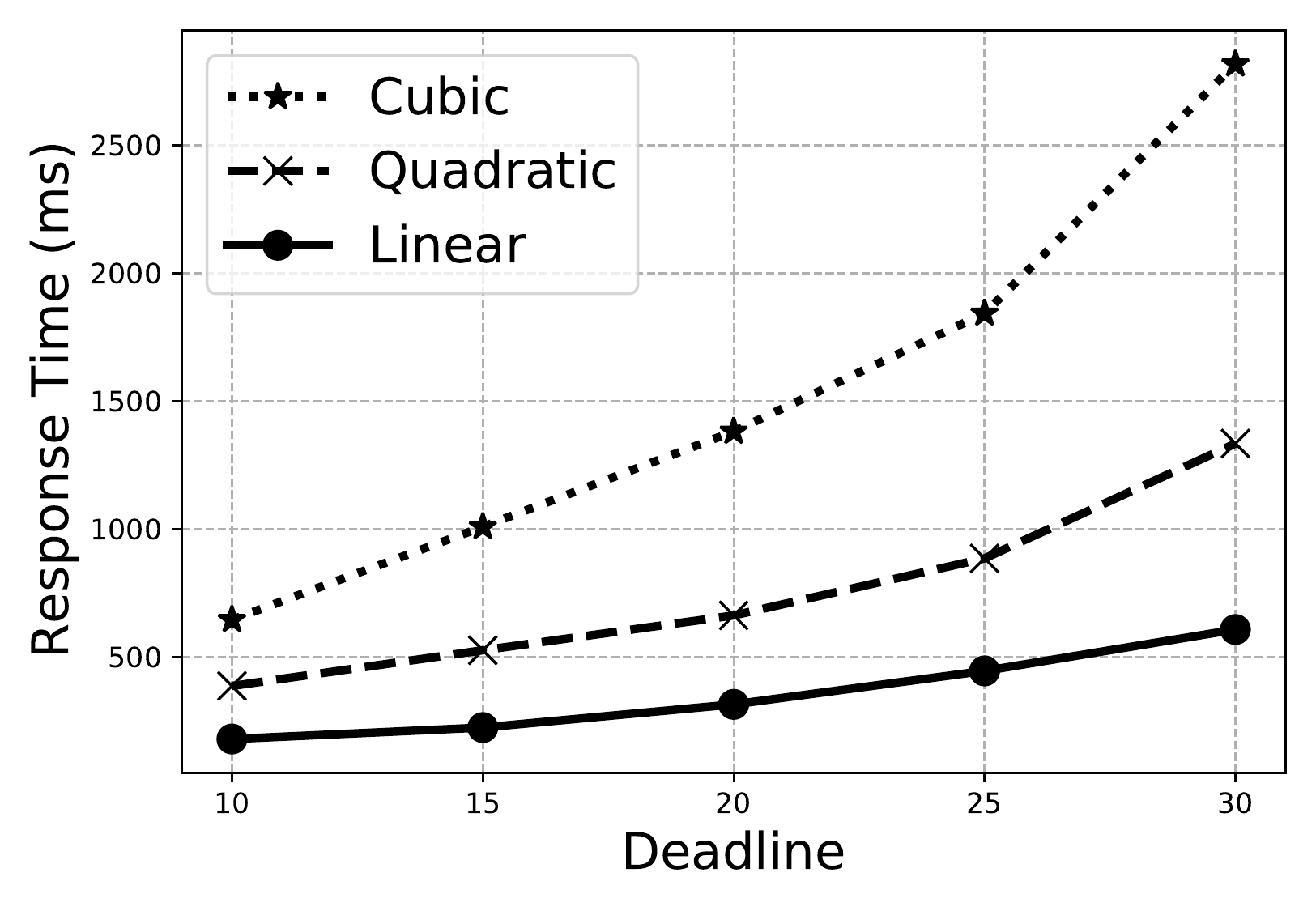}
		\vspace{-4ex}
		\caption{\footnotesize{Response time on \chengdu}}
	\end{subfigure}
	~~
    \begin{subfigure}[b]{0.22\textwidth}
		\includegraphics[width=\textwidth]{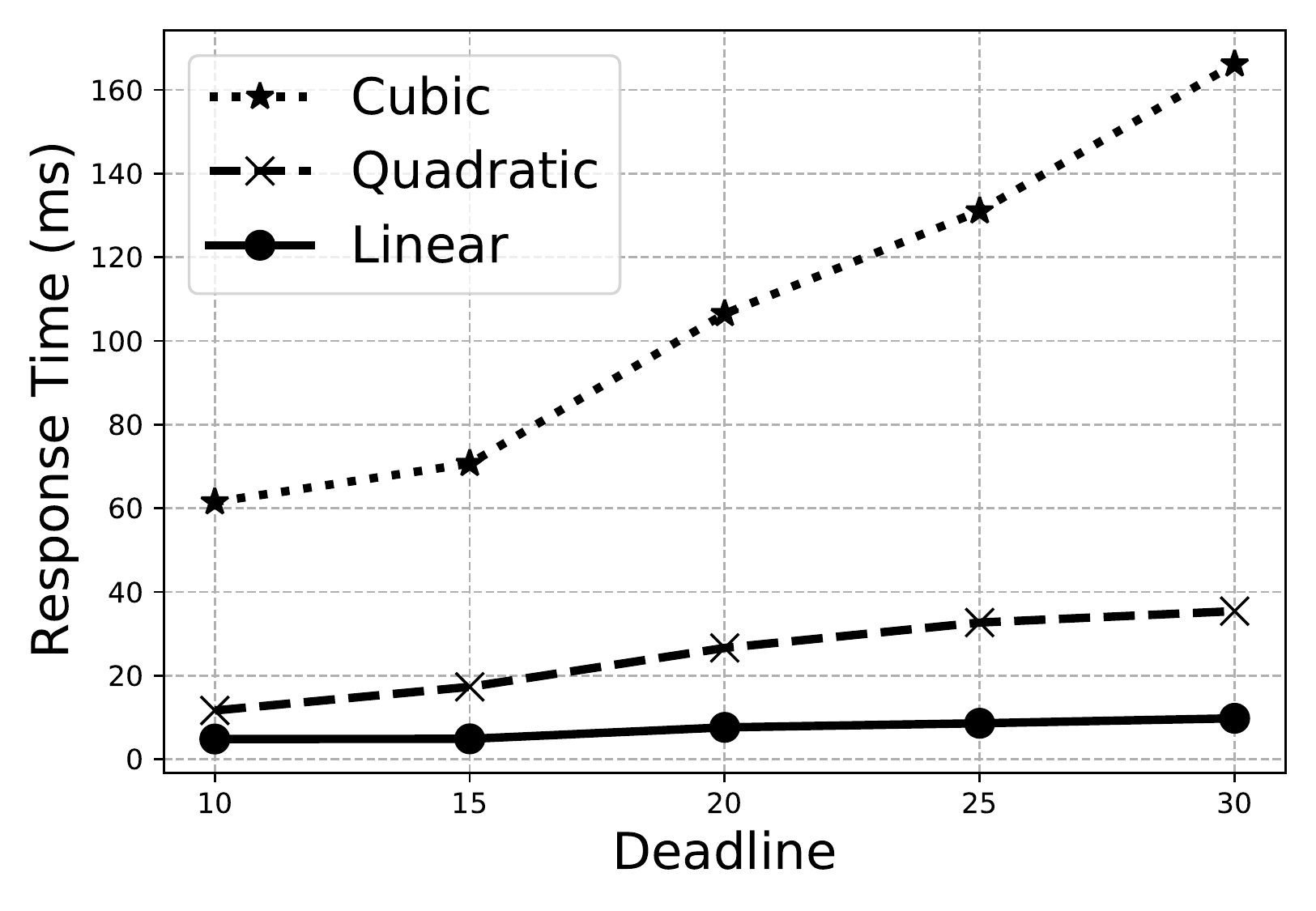}
		\vspace{-4ex}
		\caption{\footnotesize{Response time on \haikou}}
	\end{subfigure}
	
	\begin{subfigure}[b]{0.22\textwidth}
		\includegraphics[width=\textwidth]{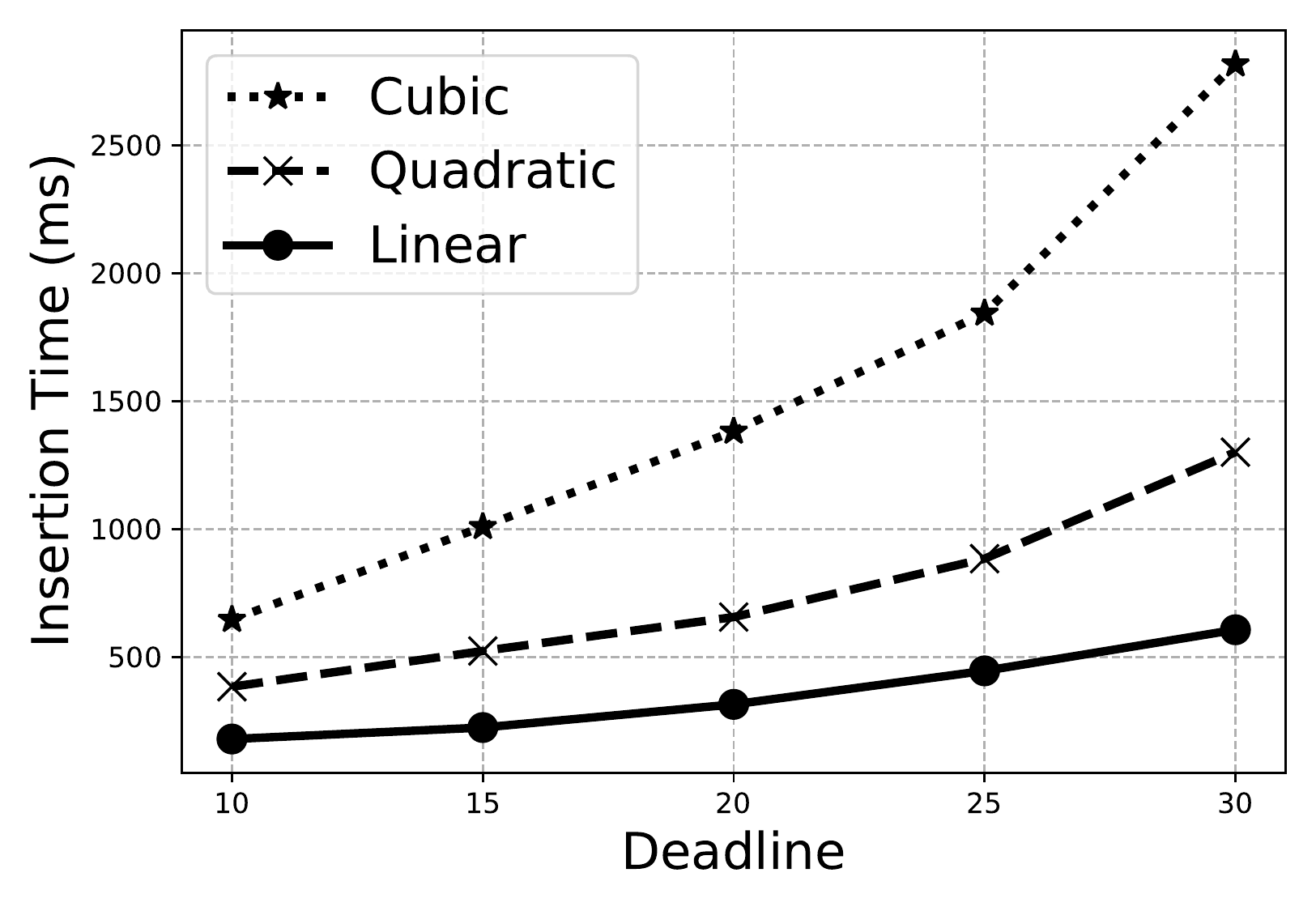}
		\vspace{-4ex}
		\caption{\footnotesize{Insertion time on \chengdu}}
	\end{subfigure}
	~~
    \begin{subfigure}[b]{0.22\textwidth}
		\includegraphics[width=\textwidth]{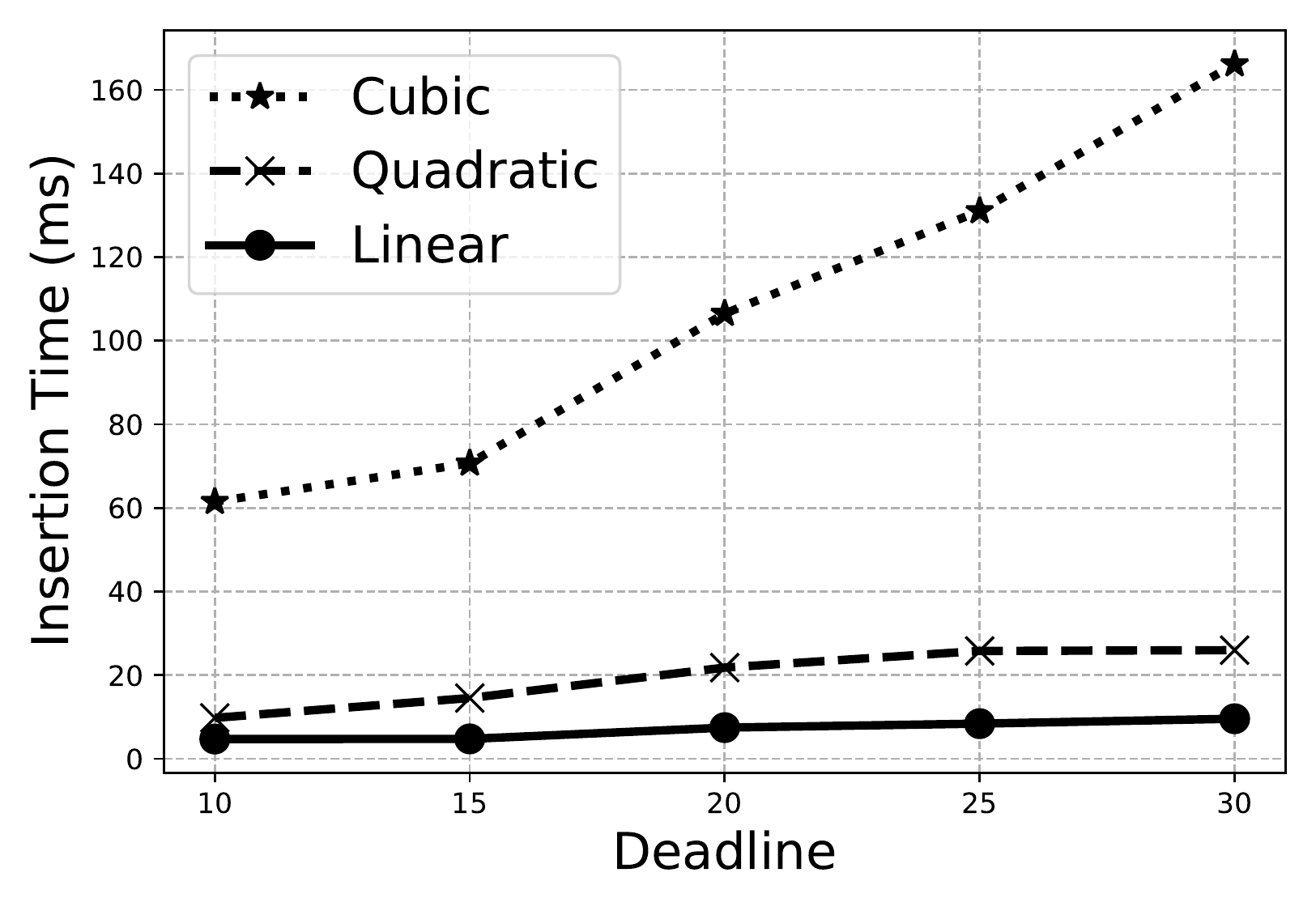}
		\vspace{-4ex}
		\caption{\footnotesize{Insertion time on \haikou}}
	\end{subfigure}	
	
	\begin{subfigure}[b]{0.22\textwidth}
		\includegraphics[width=\textwidth]{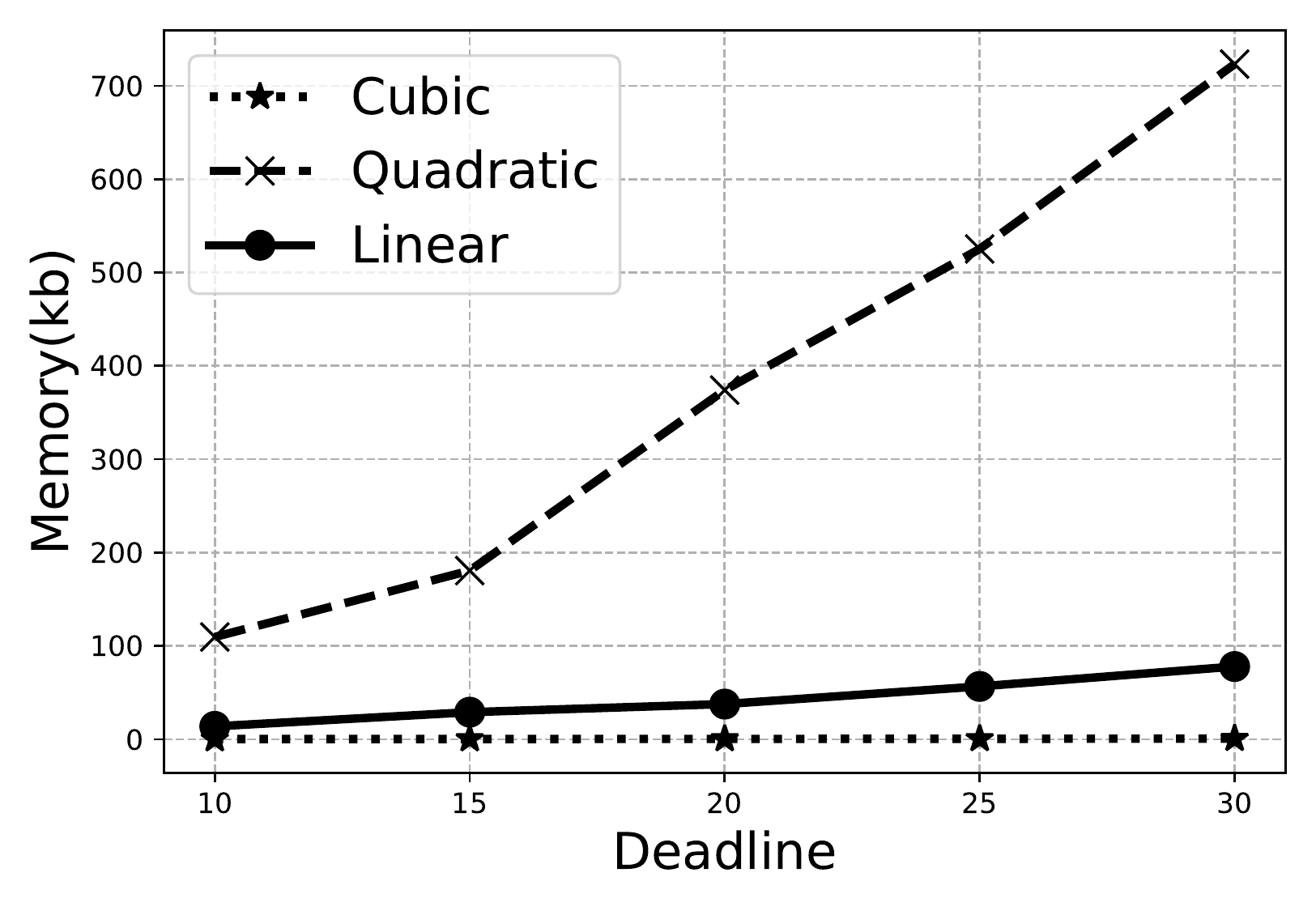}
		\vspace{-4ex}
		\caption{\footnotesize{Memory cost on \chengdu}}
	\end{subfigure}
	~~
    \begin{subfigure}[b]{0.22\textwidth}
		\includegraphics[width=\textwidth]{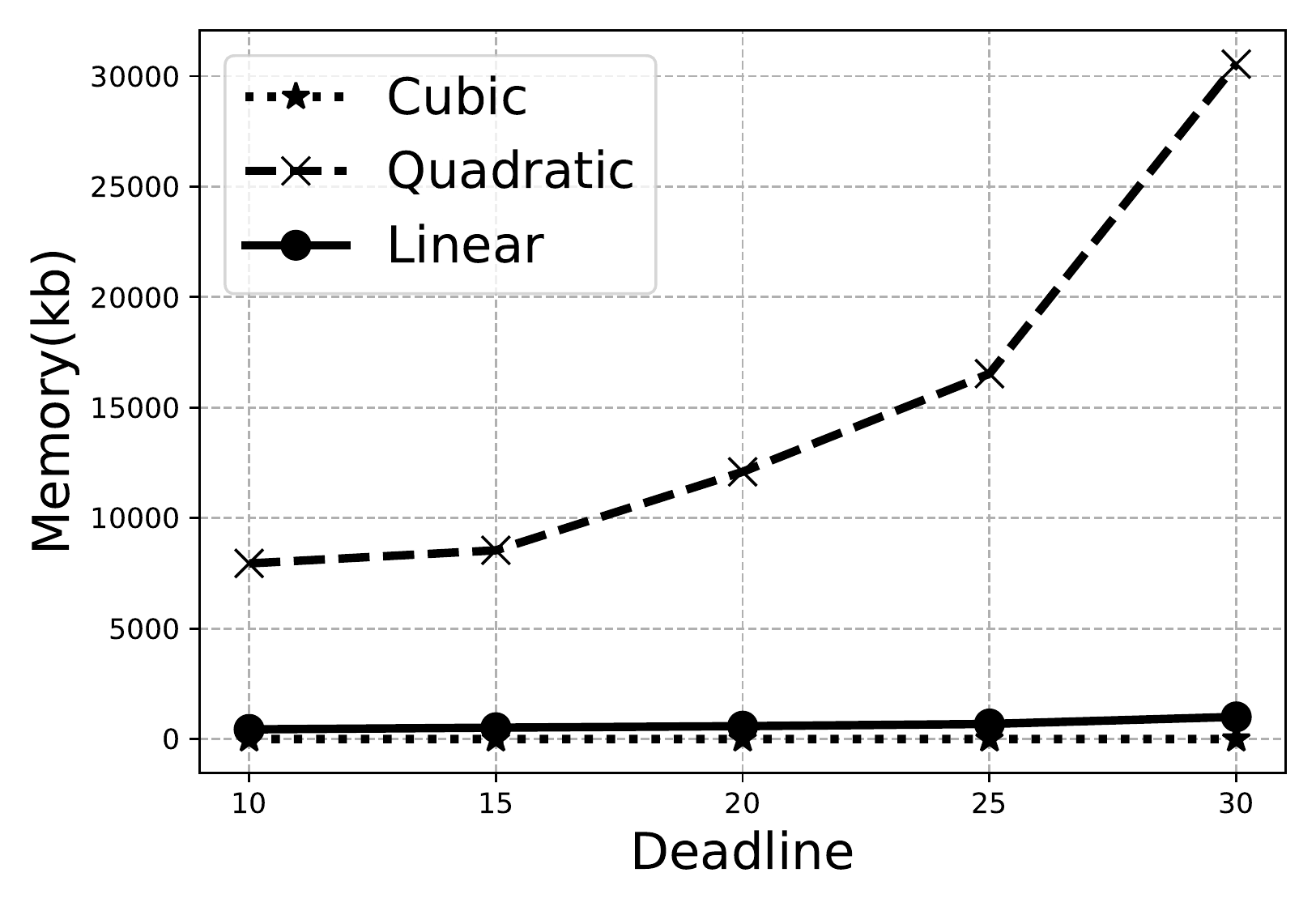}
		\vspace{-4ex}
		\caption{\footnotesize{Memory cost on \haikou}}
	\end{subfigure}	
	
	\vspace{-2ex}
	\caption{Results of varying $e_r - t_r$}
	\label{fig:deadline}
	\vspace{-3ex}
\end{figure}

\fakeparagraph{Impact of Deadline $e_r$} \figref{fig:deadline} shows the results of varying the deadline of requests on $\chengdu$ and $\haikou$. The values of $e_r - t_r$ are the scale values of the x-axis. Again, $\linear$ outperforms other two algorithms in the term of the time cost, which is up to 4.62 times faster on $\chengdu$ and 16.96 times faster on $\haikou$ in response time. As for the space cost, it still has a great advantage over $\quadra$. With the increase of $e_r - t_r$, the requests have larger deadlines, more requests can be inserted into the original feasible route, time cost of $\cubic$ and space cost of $\quadra$ increase significantly. Meanwhile, the number of invoking $\query{}$ functions, insertion time and response time of $\linear$ increase much slower than other two methods. For the memory cost, $\linear$ remains relatively stable with the increase of requests' deadlines. Comparing with the baseline method whose space cost is $O(n)$, $\linear$ consumes only slightly higher memory (no more than 77 KB, which is acceptable).

\begin{figure}[t]
	\centering
	\begin{subfigure}[b]{0.22\textwidth}
		\includegraphics[width=\textwidth]{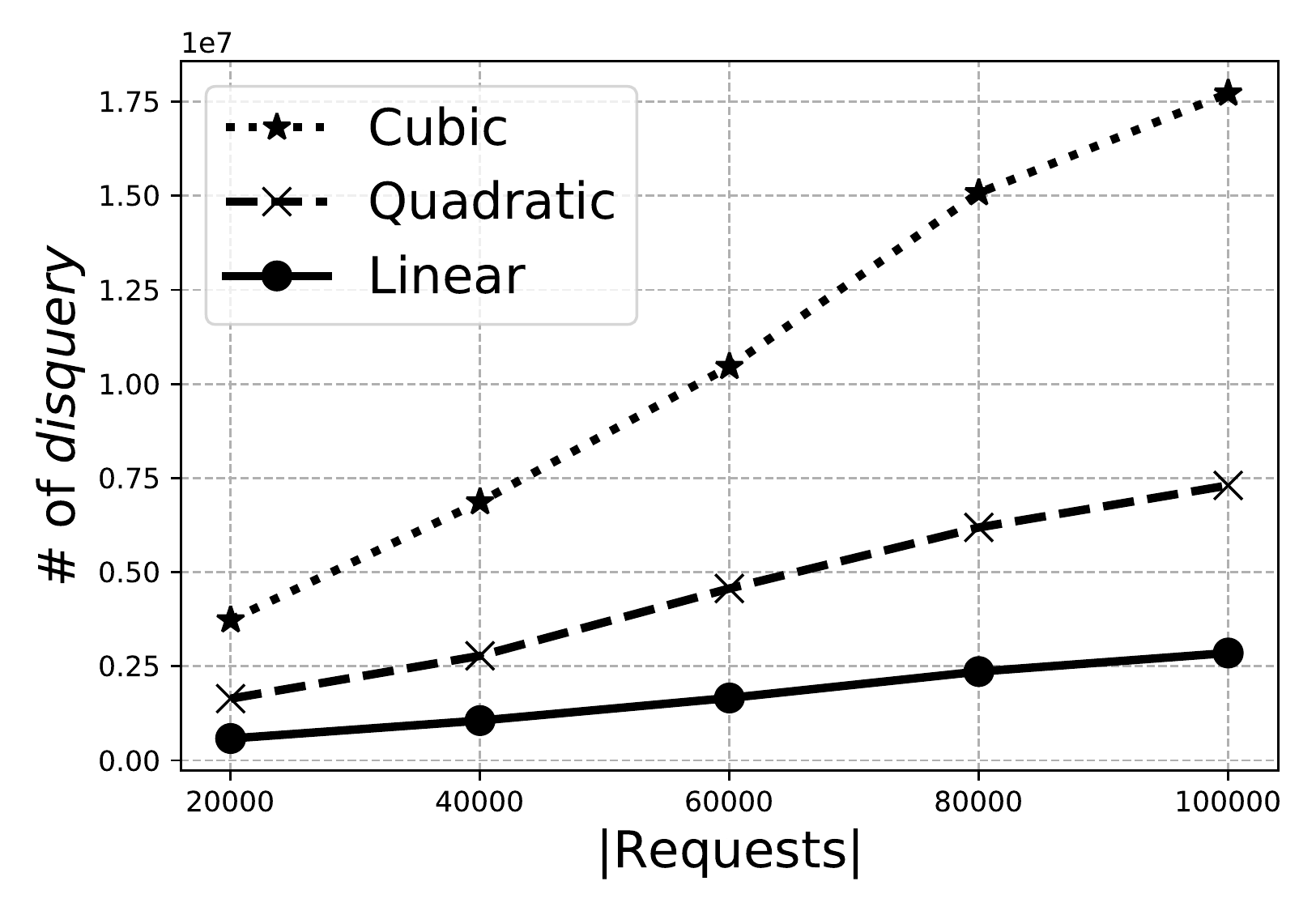}
		\vspace{-4ex}
		\caption{\footnotesize{Number of \query{} on \chengdu}}
	\end{subfigure}
	~~
    \begin{subfigure}[b]{0.22\textwidth}
		\includegraphics[width=\textwidth]{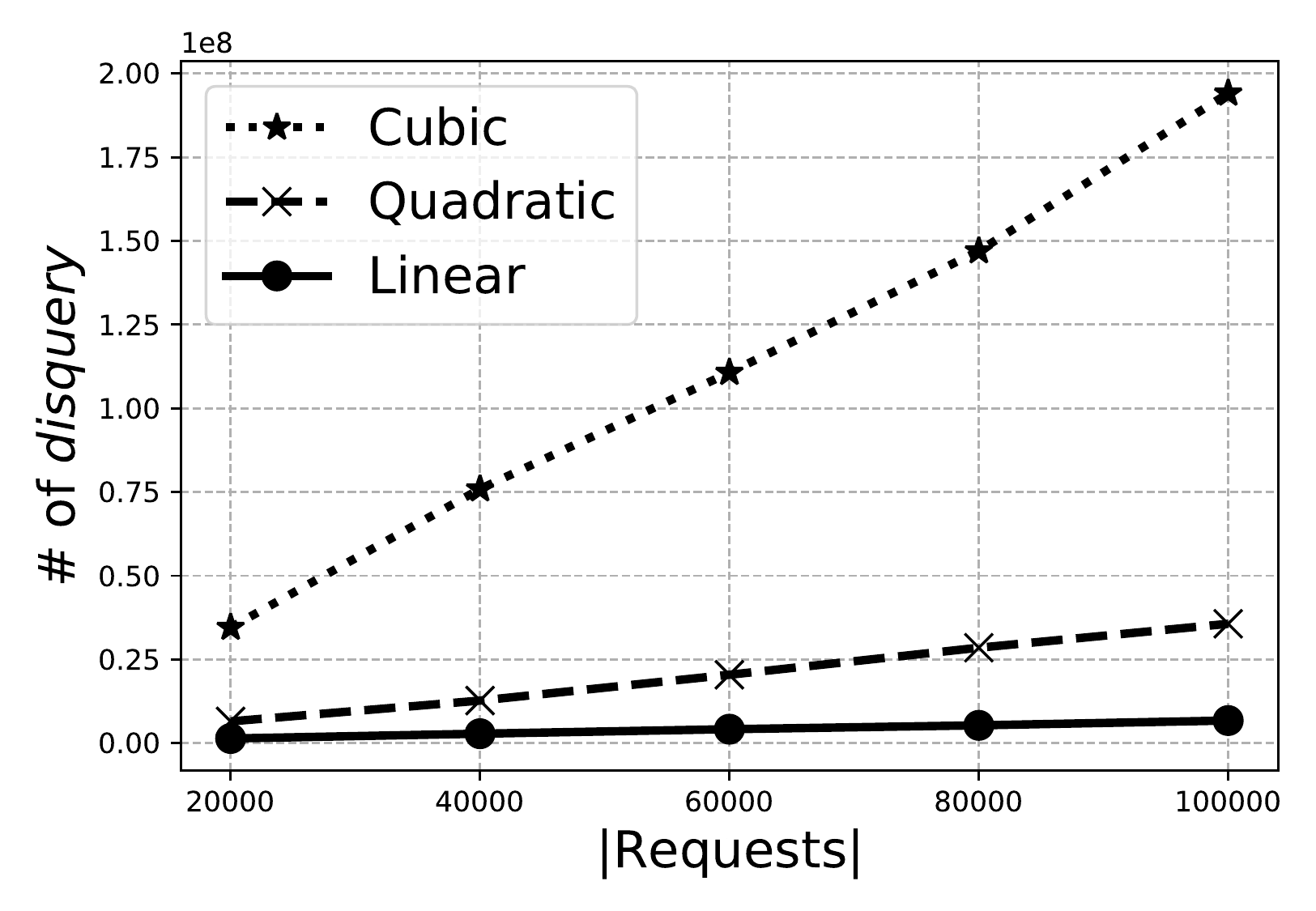}
		\vspace{-4ex}
		\caption{\footnotesize{Number of \query{} on \haikou}}
	\end{subfigure}
	
	\begin{subfigure}[b]{0.22\textwidth}
		\includegraphics[width=\textwidth]{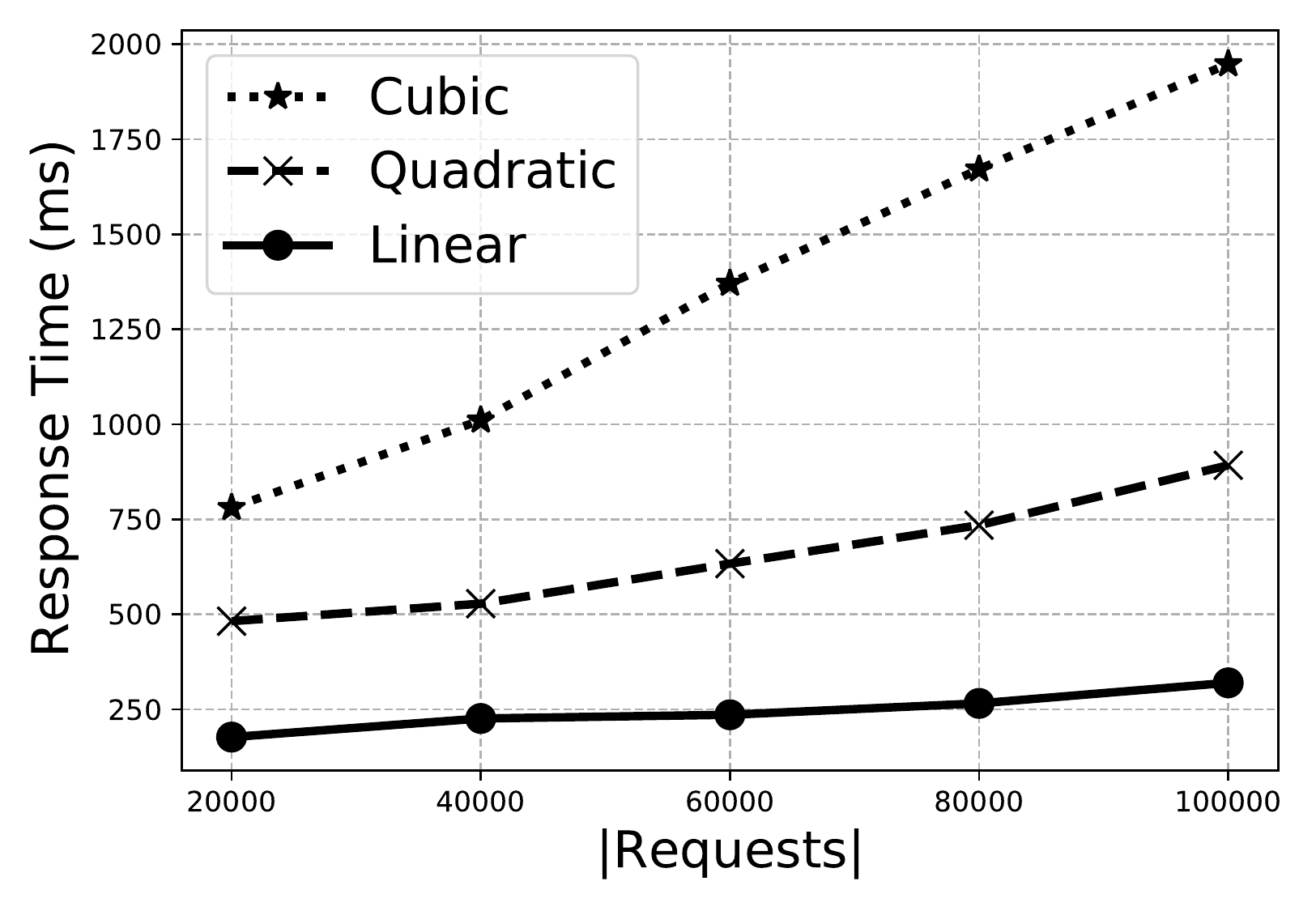}
		\vspace{-4ex}
		\caption{\footnotesize{Response time on \chengdu}}
	\end{subfigure}
	~~
    \begin{subfigure}[b]{0.22\textwidth}
		\includegraphics[width=\textwidth]{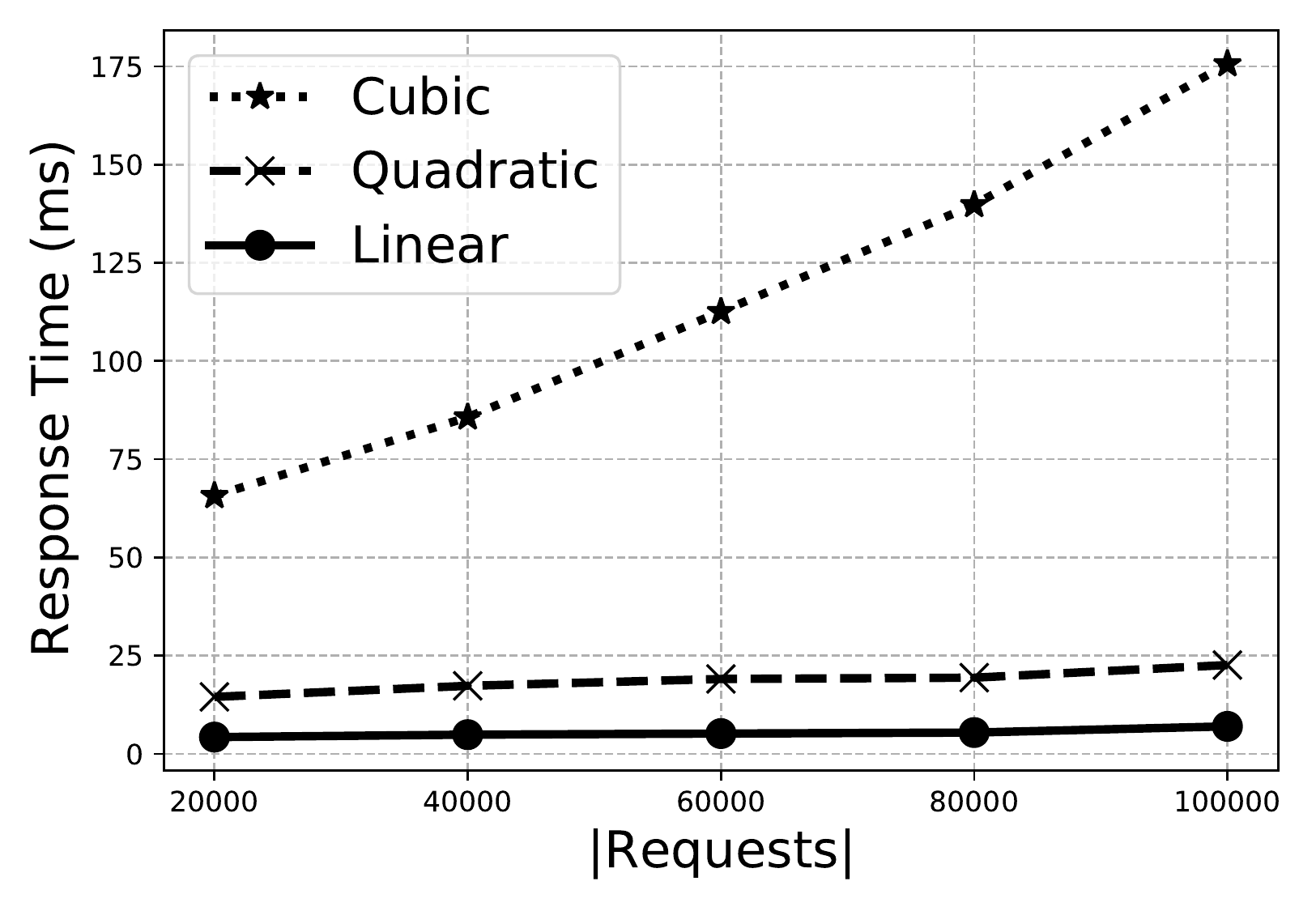}
		\vspace{-4ex}
		\caption{\footnotesize{Response time on \haikou}}
	\end{subfigure}
	
	\begin{subfigure}[b]{0.22\textwidth}
		\includegraphics[width=\textwidth]{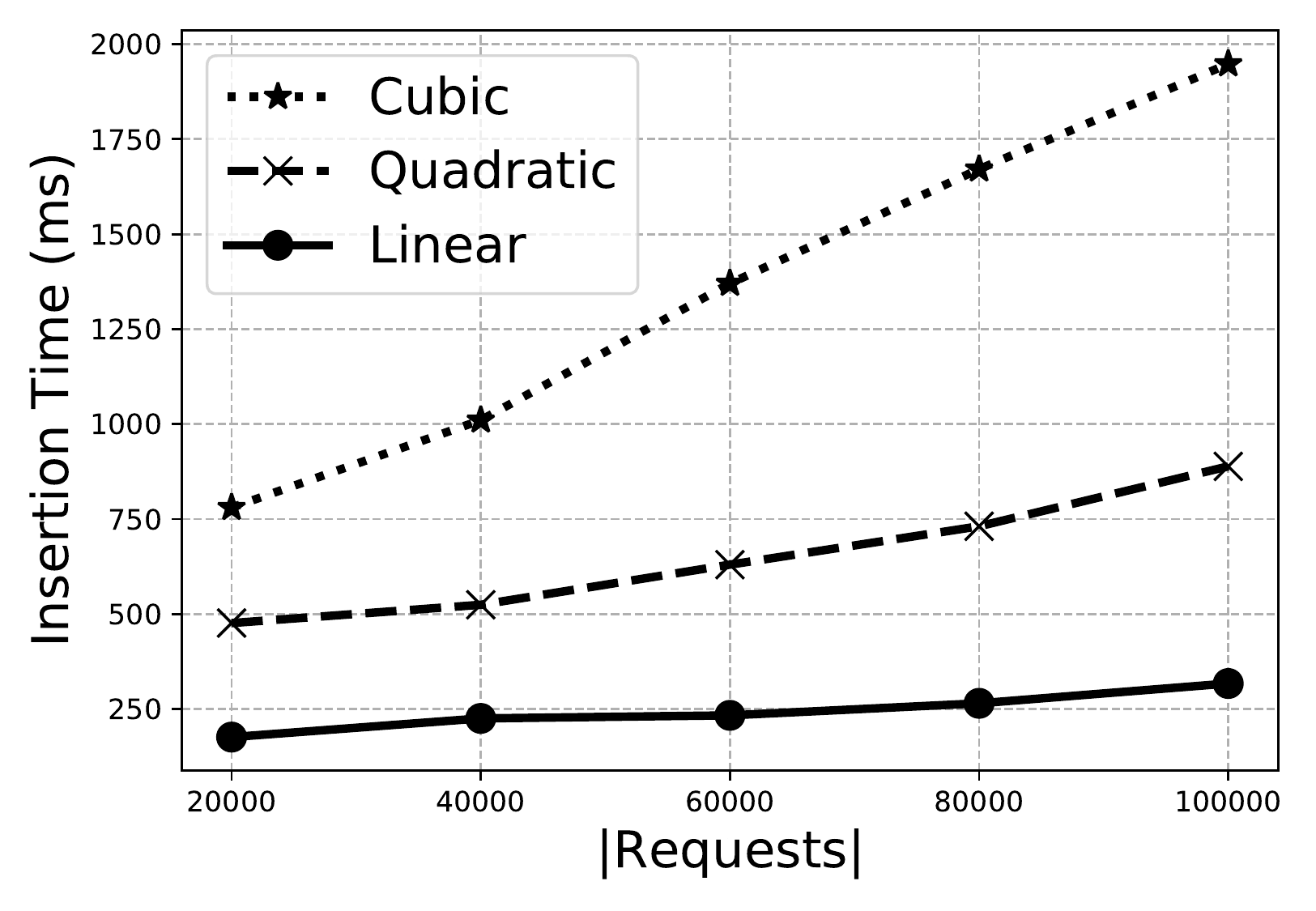}
		\vspace{-4ex}
		\caption{\footnotesize{Insertion time on \chengdu}}
	\end{subfigure}
	~~
    \begin{subfigure}[b]{0.22\textwidth}
		\includegraphics[width=\textwidth]{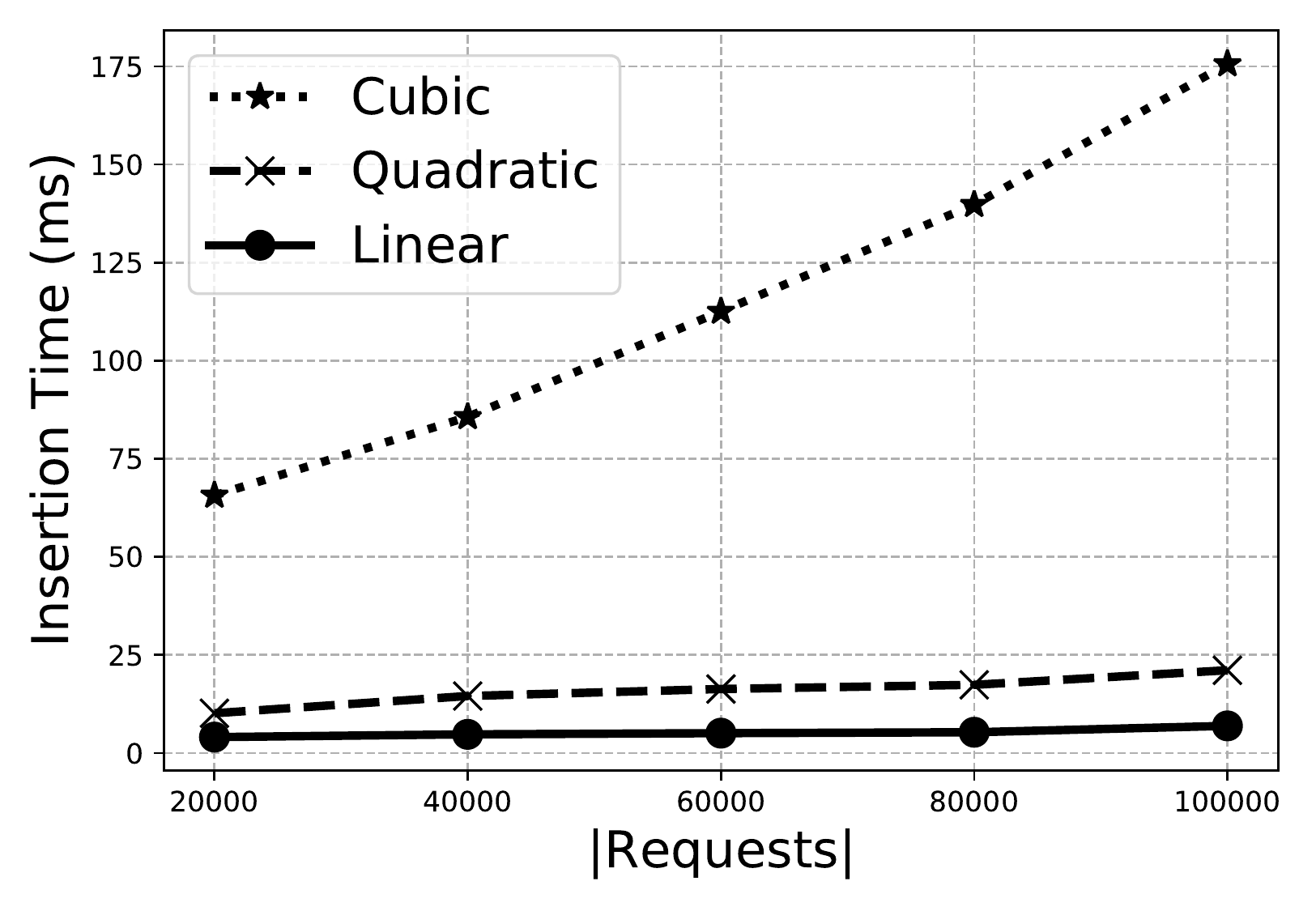}
		\vspace{-4ex}
		\caption{\footnotesize{Insertion time on \haikou}}
	\end{subfigure}	
	
	\begin{subfigure}[b]{0.22\textwidth}
		\includegraphics[width=\textwidth]{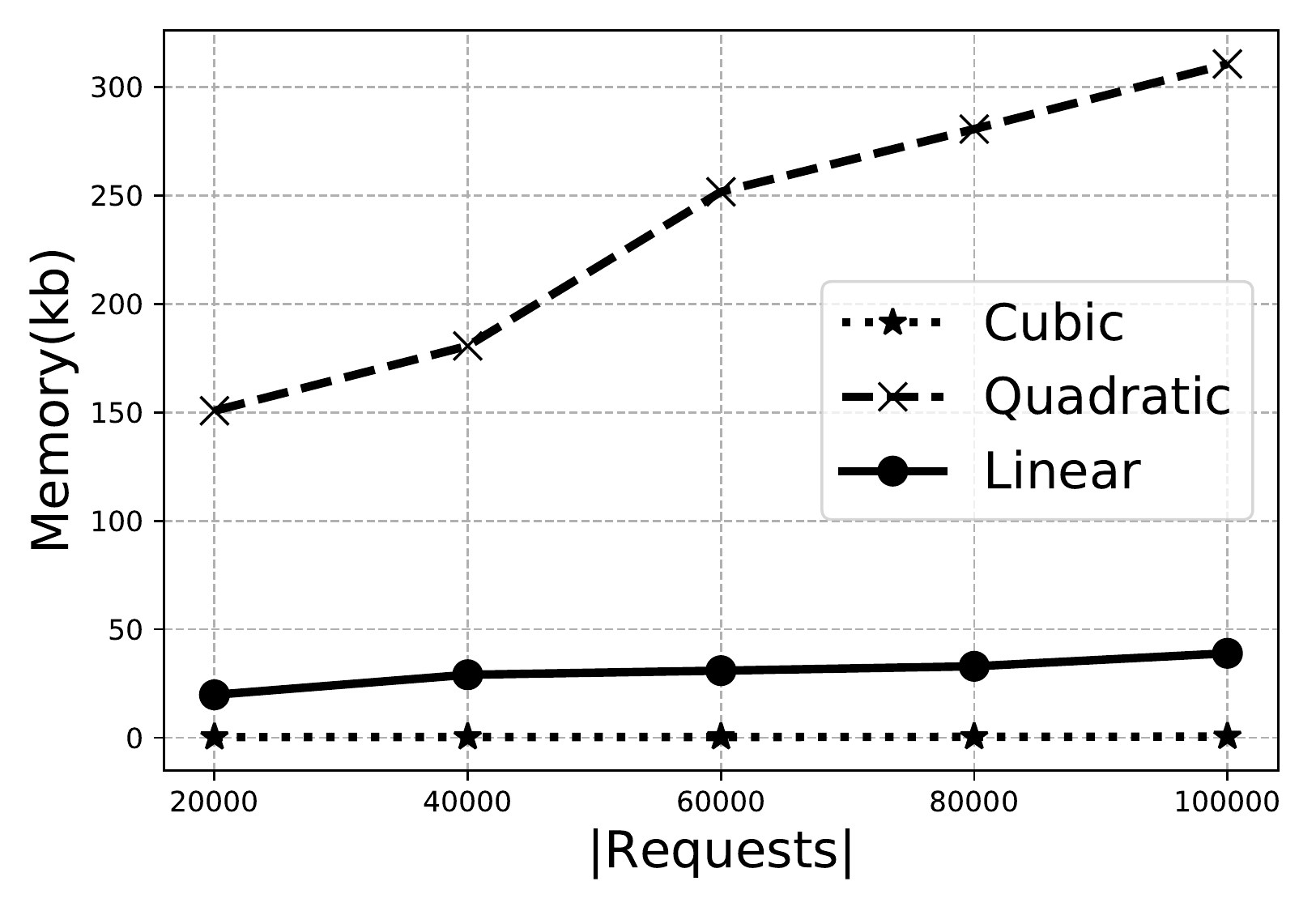}
		\vspace{-4ex}
		\caption{\footnotesize{Memory cost on \chengdu}}
	\end{subfigure}
	~~
    \begin{subfigure}[b]{0.22\textwidth}
		\includegraphics[width=\textwidth]{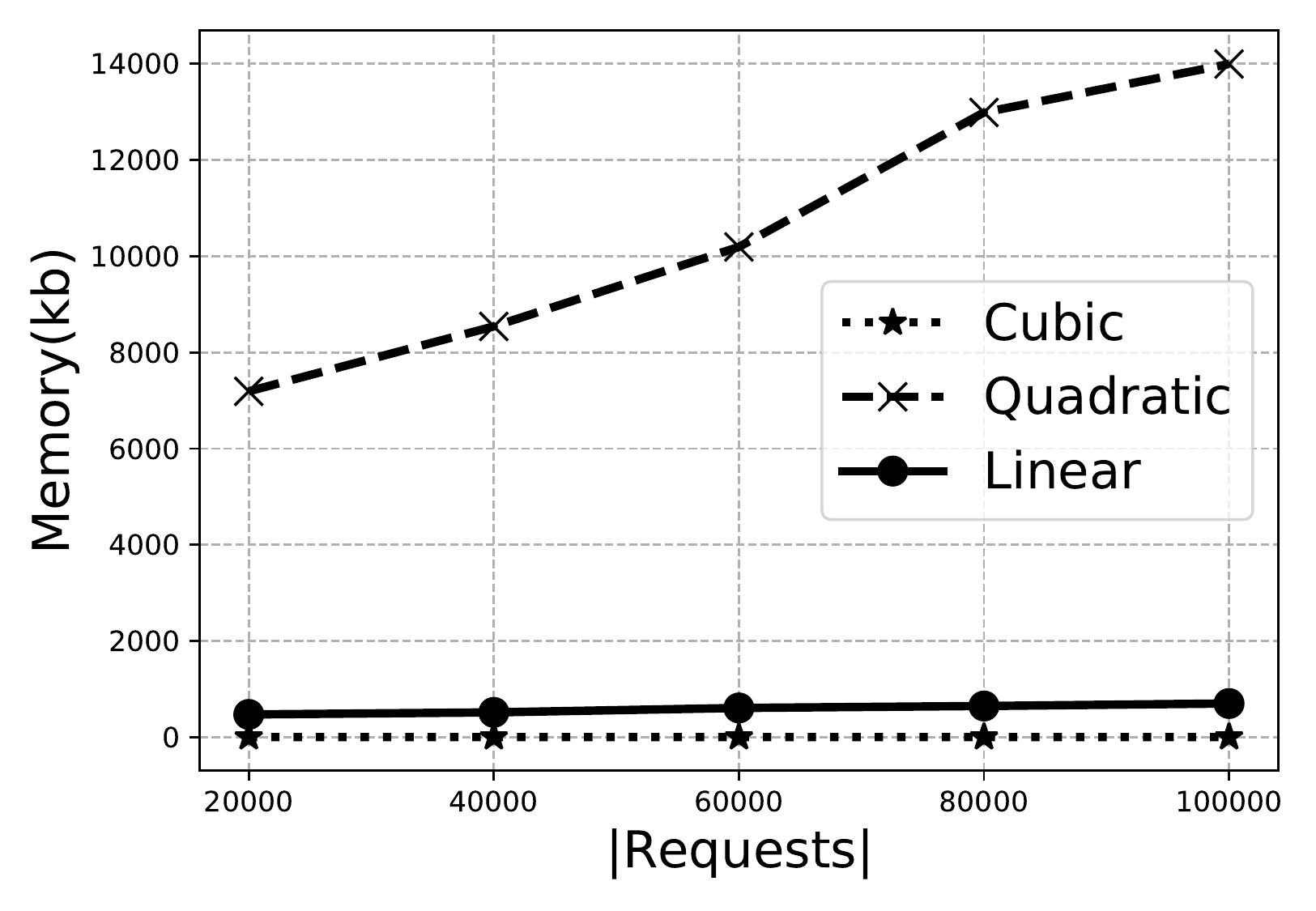}
		\vspace{-4ex}
		\caption{\footnotesize{Memory cost on \haikou}}
	\end{subfigure}	
	
	\vspace{-2ex}
	\caption{Results of varying number of requests}
	\label{fig:requests}
	\vspace{-3ex}
\end{figure}

\fakeparagraph{Impact of Number of Requests} \figref{fig:requests} shows the results of varying the number of requests on $\chengdu$ and $\haikou$. \textit{Linear Time Algorithm} is still the fastest among all three algorithms, \ie{ the number of $\query{}$ is 6.21 and 25.71 times smaller and response time is 6.09 and 12.22 times faster than baseline on two dataset.} For both $\chengdu$ and $\haikou$, with the increasing of the number of requests, the worker needs to obtain a route with longer travel time to serve the requests, so the time cost of all three algorithms increase. As for the memory cost, $\quadra$ still performs the worst due to the high space complexity, although it runs at least 4.59 times faster than $\cubic$. The memory cost of $\linear$ increases slightly with the increasing number of requests in both two datasets.

\begin{figure}[t]
	\centering
	\begin{subfigure}[b]{0.22\textwidth}
		\includegraphics[width=\textwidth]{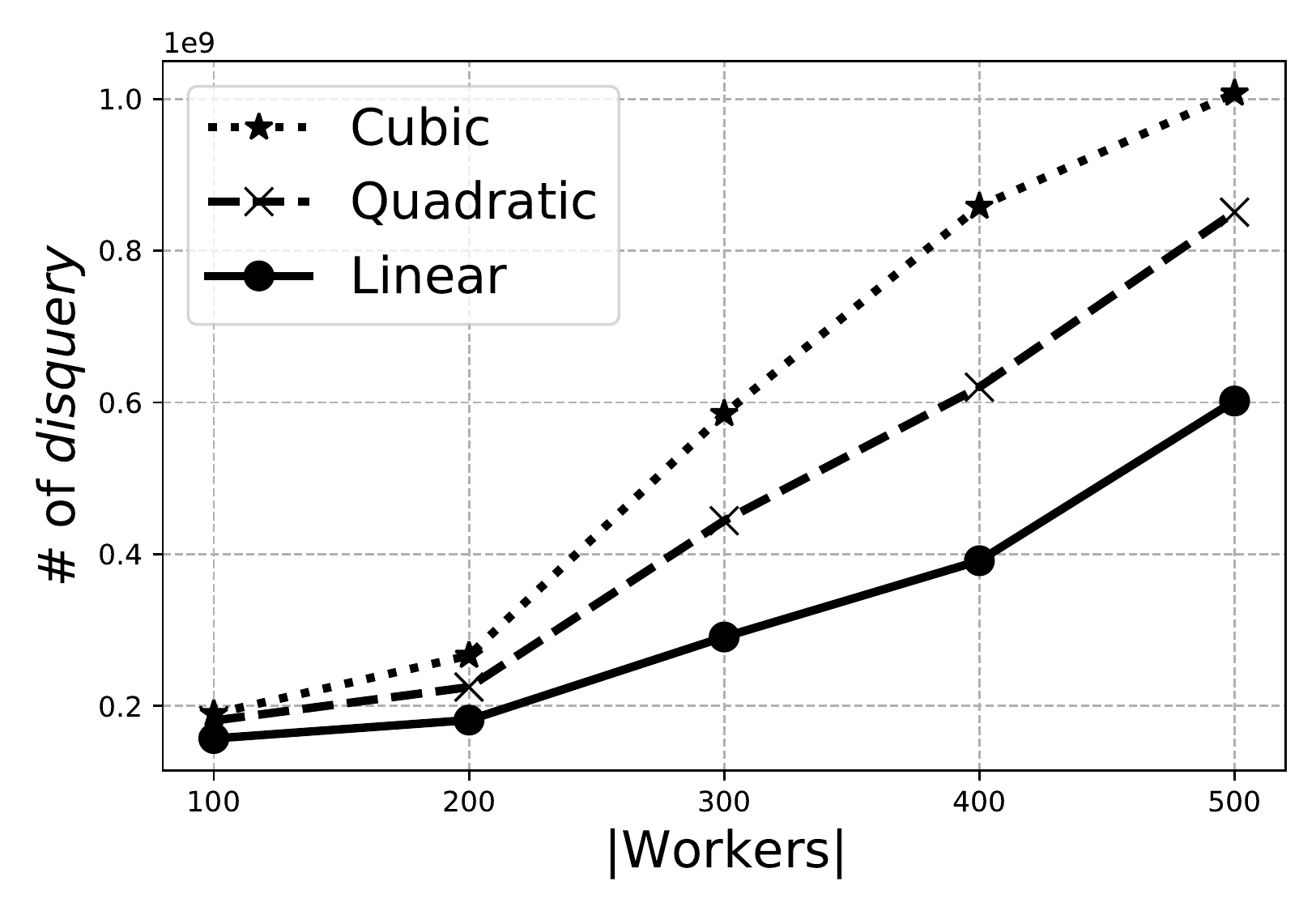}
		\vspace{-4ex}
		\caption{\footnotesize{Number of \query{}}}
	\end{subfigure}
	~~
    \begin{subfigure}[b]{0.22\textwidth}
		\includegraphics[width=\textwidth]{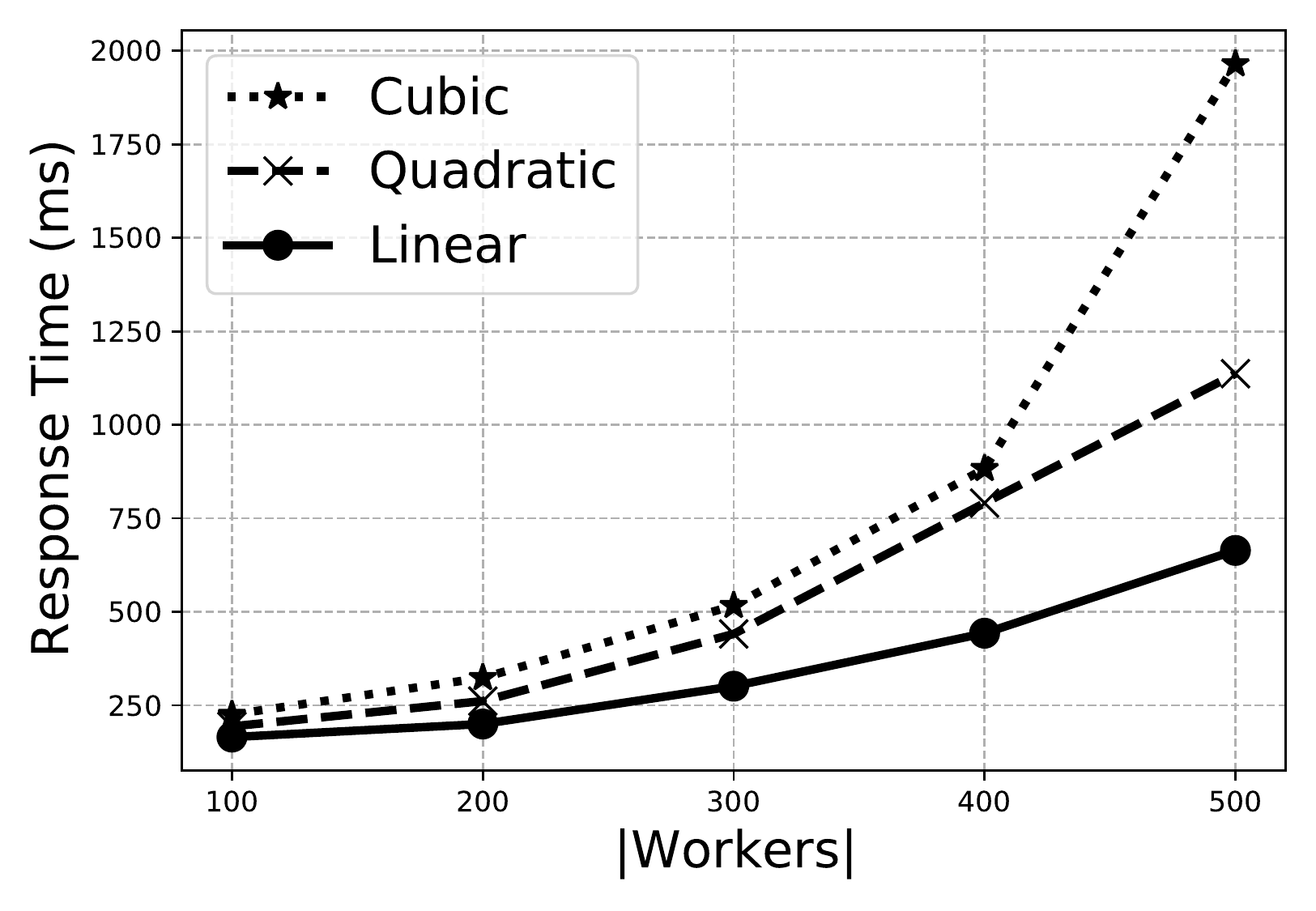}
		\vspace{-4ex}
		\caption{\footnotesize{Response time}}
	\end{subfigure}
	
	\begin{subfigure}[b]{0.22\textwidth}
		\includegraphics[width=\textwidth]{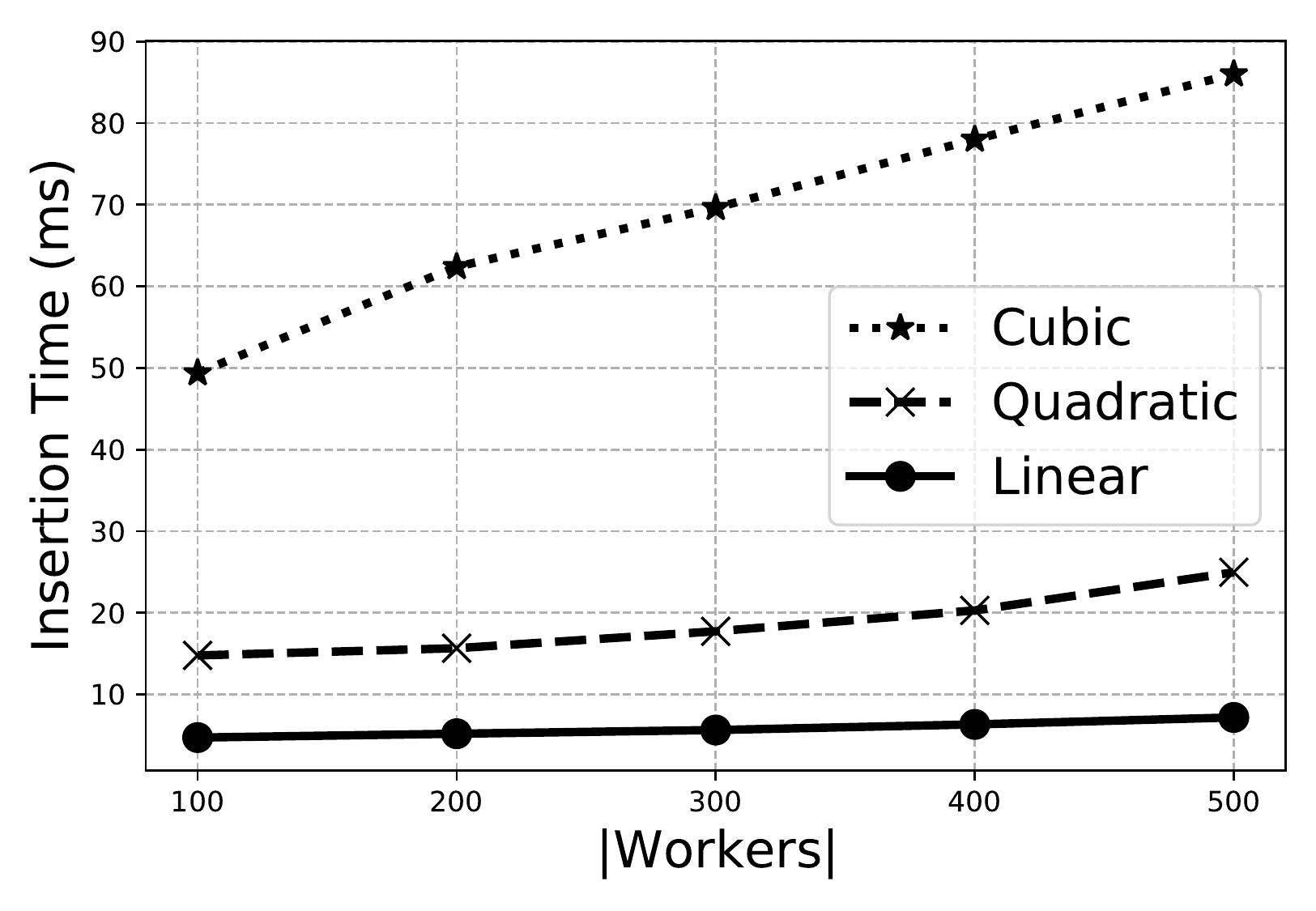}
		\vspace{-4ex}
		\caption{\footnotesize{Insertion time}}
	\end{subfigure}
	~~
    \begin{subfigure}[b]{0.22\textwidth}
		\includegraphics[width=\textwidth]{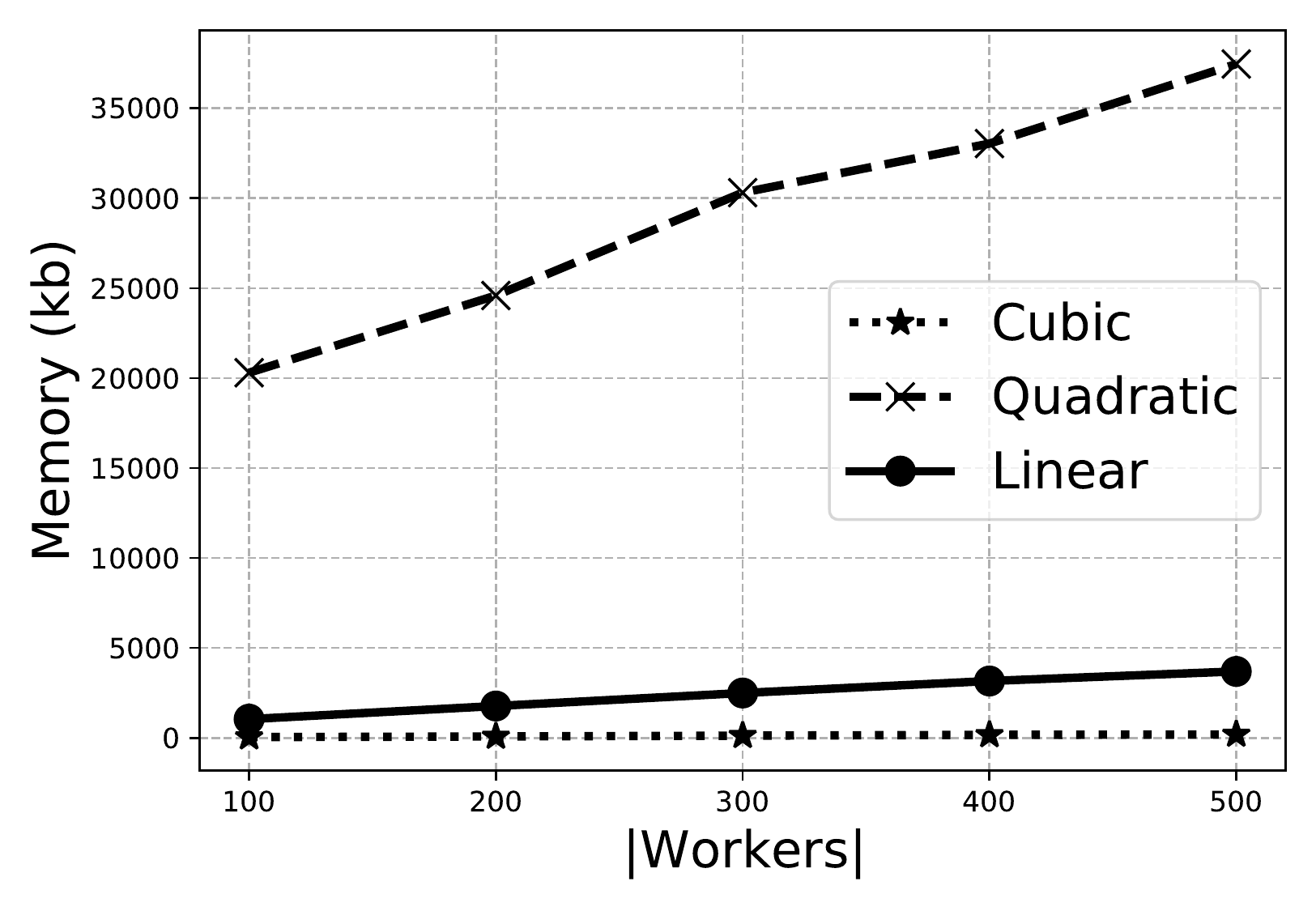}
		\vspace{-4ex}
		\caption{\footnotesize{Memory cost}}
	\end{subfigure}

	\vspace{-2ex}
	\caption{Results of scalability test on \haikou}
	\label{fig:scability}
	\vspace{-3ex}
\end{figure}

\fakeparagraph{Scalability Test} This experiment verifies the scalability of the time-dependent insertion. We show the results on $\haikou$ in \figref{fig:scability}, and the results on $\chengdu$ have the same trend. As for \textit{Cubic Time Algorithm}, with the increasing of the number of workers, the number of invoking $\query{}$ grows exponentially, which cause the extraordinary slower response time. The results prove that existing insertion operator could become the bottleneck over the time-dependent road network, especially in a larger-scale city like $\chengdu$. For \textit{Quadratic Time Algorithm}, the memory cost becomes easily notable when the number of worker is larger. Except the memory cost, \textit{Linear Time Algorithm} outperforms other two method in all metrics, this method can assign the request in real-time. Compared to \textit{Cubic Time Algorithm}, the $O(n)$ space complexity of \textit{Linear Time Algorithm} is also acceptable (less than 3.7 MB). The results show that our proposed method is suitable for the real large-scale road networks.

\fakeparagraph{Comparison between Datasets} We have following observations by comparing the results on $\chengdu$ and $\haikou$:
\begin{itemize}
  \item
   $\linear$ runs extremely fast in both datasets, this method can handle the request in real-time. Compared to $\haikou$, the time-dependent insertion operator runs slower than on $\chengdu$, it is because the shortest travel time query $\query{}$ becomes a bottleneck on a larger road network.
  \item
  For memory cost, all three methods consume more memory on $\haikou$ (0.496 - 30541 KB) than on $\chengdu$ (0.238 - 723 KB). This may be because requests on $\haikou$ have more potential to share a route, the time-dependent insertion operator has more feasible insertion positions for each origin-destination pair of one request and then increase the memory cost.
\end{itemize}

\fakeparagraph{Summary of Experiments} We summarize the results as follows:
\begin{itemize}
  \item
   It is impractical to extend the existing insertion operator straightforwardly to ridesharing service over the time dependent road networks. In other words, the response time of baseline method is more than 2990 milliseconds on $\chengdu$, it is impractical to handle the request in real-time.
  \item
    Our optimized algorithms $\quadra$, $\textit{Linear}$ $\textit{Time Algorithm}$ are 2.16 and 6.09 times faster on $\chengdu$ than the baseline, as for $\haikou$ these two algorithms are 6.39 to 25.79 times faster than the baseline, due to their improvement in the time complexity.
  \item
    $\linear$ reduces the memory cost up to 97.1\% compared with $\quadra$. The memory consumed by $\linear$ is only slightly larger than $\cubic$ , \eg{ no more than 77 KB  on \chengdu}. 
\end{itemize}

%% file: 04_related.tex
\vspace{-3ex}
\section{Related Work}\label{sec:related}

In this section, we discuss the related works. Specifically, the related works can be classified into two categories: route processing over time-dependent road networks and rout planing for rideshairng.

\textbf{Route processing over time-dependent road networks.}
The time-dependent road network has attracted many 
research efforts in spatial
databases. Piecewise linear functions are widely adopted to fit the time-dependent edge weight functions to model the dynamic traffic in reality \cite{shortestquery} \cite{yuanyeicde2019} \cite{yuanyeicde2021}. For the fundamental query problem over time-dependent road networks, \cite{tist2021} introduces a novel dual-level path index, when a route planning query is given, the filter-and-refine strategy based on the index is utilized to enhance the efficiency of the route planning. \cite{shortestquery} splits the road network into hierarchical partitions and constructed a balanced tree index. When it comes a query, it compounds the piecewise linear functions of border vertexes in partitions to answer the query.

Route planning related problem over time-dependent road network is also critical. \cite{yuanyeicde2019} proposes an online route planning over time-dependent road network problem, and develop a request-inserted algorithm with complexity $O(|R|m + |R|nlogn)$ to reduce the competitive ratio. Last mile delivery problem is extended over on time-dependent road network in \cite{sigspatial2021}, a courier can take multiple parcels start from a same warehouse and each parcel can be delivered to alternative locations depending on the time, to maxmize the number of delivered parcels an optimized insertion heuristic method is proposed to insert delivery location into the courier’s path. However, in existing works they only deal with one location for each request, \cite{yuanyeicde2019} makes the assumption that all passengers have the same destination and \cite{sigspatial2021} makes the assumption that all parcels have the same origin. Besides that, both \cite{yuanyeicde2019} and \cite{sigspatial2021} omit checking the feasibility of the worker. Compared to these problems, our time-dependent insertion problem is much harder. As shown in example 1, each request is picked up from his own origin and delivered to his own destination by the worker $w$. For the new request $r_3$, we plan a new feasible route $\langle o_1, o_2, o_3, d_2, d_3, d_1 \rangle$, including both his origin $o_3$ and destination $d_3$ to serve him. Thus, without the assumption that requests have the same origin or destination, existing algorithms cannot be applied to our problem. 

\textbf{Route planning for ridesharing.} Ridesharing service has also been studied in many domains vary from Database to AI. Different studies focus on different objectives, \cite{yuzheng2013}\cite{zengyangicde} focus on maxmize the number of requests been severed, \cite{gecco2012}\cite{ijcai2013}\cite{huangyan2014}\cite{icde2018}\cite{yuxiang2019icde} focus on another essential objective of minimizing the total travel time which is similar to our optimization objective. From the perspective of the platform, the total revenue is also an objective that needs to be taken into consideration \cite{libinvldb2018}\cite{libinicde2019}. 

The insertion operator is an efficient method to plan a new possible route based on the optimized objective by inserting a new request into the current route of a worker proposed in \cite{parameter-dasfaa2018}. Recently, there are also some studies that have have proven the effectiveness and efficiency in large scale real-world static road network \cite{yuxiang2018vldb}\cite{yuxiang2019icde}\cite{yuxiang2020tkde}. \cite{yuxiang2019icde} \cite{yuxiang2020tkde} extensively study insertion operator, a dynamic programming based novel partition framework reduces the time complexity from $O(n^3)$ to $O(n^2)$, then fenwick tree index helps to further reduce the complexity to $O(n)$. A unified route planning problem for shared mobility is defined in \cite{yuxiang2018vldb}, it proves that there is no polynomial-time algorithm with constant competitive ratio for solving this problem, a novel two-phased solution based on dynamic programming insertion is proposed to solve this problem approximately. However, none of these stuides have consider the time-varying travel cost characteristics of road network in dynamic traffic situation. Back to example 1, the existing insertion operator can calculate the travel time of new route $\langle o_1, o_2, o_3, d_2, d_3, d_1 \rangle$ by adding four static edge weights directly ($\langle o_2, o_3 \rangle$, $\langle o_3, d_2 \rangle$, $\langle d_2, d_3 \rangle$ and $\langle d_3, d_1 \rangle$) over the static road network. However, over the time-dependent road network, the travel time of this new route is dependent on the time when $r_3$ is picked up by the worker \ie{$arr(o_3)$}. Therefore, the method of calculating new route's travel time in existing insertion operator is no longer applicable, the existing insertion operator cannot handle the route-planning related problem efficiently over the time-dependent road network.

%% file: 07_conclusion.tex
\section{conclusion}	\label{sec:conclusion}

In this paper, we first study a widely-adopted insertion operator in real-time ridesharing service over time-dependent road networks. We extend the existing insertion operator to time-dependent scenario as a baseline by invoking large number of the shortest travel time queries, and the time complexity is $O(n^3)$. We further compound time-dependent edge weight functions to get compound travel time functions along the feasible route of the worker, the time complexity reduced to $O(n^2)$. To maintain the compound travel time functions between any vertexes pair that belongs to the worker's route, the space cost is also $O(n^2)$. Based on the FIFO property of a time-dependent road network, we propose an efficient time-dependent insertion operator that takes linear time and space, by proving that given a position of the route to insert the destination of the new request, we can find the optimal position to insert the origin in $O(1)$ time. Extensive experiments on real datasets validate the efficiency and scalability of our time-dependent insertion operator. Particularly, the time-dependent insertion operator can be up to 25.79 times faster than baseline method under different settings on two real-world datasets.